\documentclass[draft]{article}
\usepackage{amsfonts,amsmath,amssymb,color,graphicx,latexsym}
\newtheorem{lemma}{Lemma}
\newtheorem{theorem}{Theorem}
\newtheorem{proposition}{Proposition}
\newtheorem{corollary}{Corollary}
\newenvironment{proof}{{\bf Proof:}}{~$\dashv$\\}
\def\Z{\mathbb{Z}}
\def\N{\mathbb{N}}
\def\lsem{{\lbrack\mid}}
\def\rsem{{\mid\rbrack}}
\def\qd{\mathtt{qd}}
\def\qdd{\mathtt{qdd}}
\def\dust{\mathtt{dust}}
\def\Th{\mathtt{Th}}
\def\Fr{\mathtt{Fr}}
\def\size{\mathtt{len}}
\def\kernel{\mathtt{kernel}}
\def\root{\mathtt{root}}
\def\coNP{\mathbf{coNP}}
\def\NEXPTIME{\mathbf{NEXPTIME}}
\def\PSPACE{\mathbf{PSPACE}}
\def\EXPSPACE{\mathbf{EXPSPACE}}
\def\vvar{\mathtt{var}}
\def\L{\mathbf{L}}
\def\S{\mathbf{S}}
\def\K{\mathbf{K}}
\def\KD{\mathbf{KD}}
\def\PVAR{\mathbf{PVAR}}
\def\IVAR{\mathbf{IVAR}}
\def\fiv{\mathtt{fiv}}
\def\qd{\mathtt{qd}}
\begin{document}
\title{Modal definability in Euclidean modal logics}
\author{Philippe Balbiani$^{a}$\footnote{Email: philippe.balbiani@irit.fr.}
\hspace{0.24cm}
Tinko Tinchev$^{b}$\footnote{Email: tinko@fmi.uni-sofia.bg.}}
\date{$^{a}$Toulouse Institute of Computer Science Research
\\
CNRS~---~INPT~---~UT, Toulouse, France
\\
$^{b}$Faculty of Mathematics and Informatics
\\
Sofia University St.~Kliment Ohridski, Sofia, Bulgaria}
\maketitle
\begin{abstract}
This paper is about the computability of the modal definability problem in classes of frames determined by Euclidean modal logics.
We characterize those Euclidean modal logics such that the classes of frames they determine give rise to an undecidable modal definability problem.
\end{abstract}
Model Theory of Classical First-Order Logic.
Correspondence Theory.
Modal definability.
Stable classes of frames.
\section{Introduction}\label{section:introduction}
A modal formula is first-order definable in a class of frames if there exists a sentence that is true exactly in the frames of that class validating the modal formula.
A sentence is modally definable in a class of frames if there exists a modal formula that is valid exactly in the frames of that class satisfying the sentence.
A modal formula and a sentence correspond in a class of frames if they are respectively valid and true in the same frames of that class.
%
%
\\
\\
The question of the correspondence between modal formulas and sentences is concomitant with the creation of the relational semantics of modal logic.
Kripke~\cite{Kripke:1963} already observed that some sentences are modally definable: transitivity vs $\Box p\rightarrow\Box\Box p$, symmetry vs $p\rightarrow\Box\Diamond p$, etc.
Less than $20$ years have elapsed between Kripke's observation and the development of Correspondence Theory culminating in the publication of the book ``Modal Logic and Classical Logic''~\cite{vanBenthem:1983}: in 1975, Sahlqvist~\cite{Sahlqvist:1975} isolated a large set of modal formulas which guarantee completeness with respect to first-order definable classes of frames, whereas van Benthem~\cite{vanBenthem:1975} and Goldblatt~\cite{Goldblatt:1975} independently noticed that McKinsey formula $\Box\Diamond p\rightarrow\Diamond\Box p$ has no first-order correspondent.
\\
\\
%
%
Since the first-order conditions corresponding to Sahlqvist formulas are effectively computable~\cite{Sahlqvist:1975}, then it is natural to ask whether Sahlqvist fragment modally defines all first-order correspondents of modal formulas.
This question has received a negative answer, the conjunction $(\Box\Diamond p\rightarrow\Diamond\Box p)\wedge(\Box p\rightarrow\Box\Box p)$ possessing a first-order correspondent not modally definable by a Sahlqvist formula~\cite[Chapter~$3$]{Blackburn:deRijke:Venema:2001}.
However, notice that this statement should be carefully considered, seeing that the ``Sahlqvist formulas'' considered by Blackburn {\it et al.}\/ constitute a strict subset of the formulas considered by Sahlqvist.
Indeed, we do not know whether the first-order correspondent of $(\Box\Diamond p\rightarrow\Diamond\Box p)\wedge(\Box p\rightarrow\Box\Box p)$ is modally definable by a formula considered by Sahlqvist.
\\
\\
In any case, it is natural to ask whether the first-order definability problem, the modal definability problem and the correspondence problem are decidable.
%
%
%
%
The answers to these questions have been firstly obtained by Chagrova in her doctoral thesis~\cite{Chagrova:1989} and then further developed in~\cite{Chagrov:Chagrova:1995,Chagrov:Chagrova:2006,Chagrov:Chagrova:2007,Chagrova:1991}:
%
%
%
%
%
%
%
%
%
%
the first-order definability problem, the modal definability problem and the correspondence problem are undecidable.
%
%
%
%
Chagrova's results concern the class of all frames.
They have been obtained by reductions from accessibility problems in Minsky machines~\cite{Chagrov:Chagrova:2006}.
As a result, it is not obvious how similar results can be obtained with respect to restricted classes of frames (the class of all transitive frames, the class of all symmetric frames, etc).
%
%
\\
\\
With respect to the class of frames determined by $\S5$, it follows from $\S5$~Reduction Theorem in~\cite[Page~$51$]{Hughes:Cresswell:1968} and Lemma~$9.7$ in~\cite[Page~$99$]{vanBenthem:1983} that every modal formula is first-order definable.
Moreover, as proved in~\cite{Balbiani:Tinchev:2005,Balbiani:Tinchev:2006}, modal definability with respect to the class of frames determined by $\S5$ is $\PSPACE$-complete.
%
%
%
%
With respect to the classes of frames determined by $\KD45$ and $\K45$, similar results hold as well~\cite{Georgiev:2017a,Georgiev:2017b}.
\\
\\
Therefore, a first question comes: is there an extension of $\K5$ such that the class of frames it determines gives rise both to a trivial first-order definability problem~---~i.e. every modal formula is first-order definable in that class~---~and an undecidable modal definability problem?
This question has been positively answered: with respect to the class of frames determined by $\K5$, first-order definability is trivial, whereas modal definability is undecidable~\cite{Balbiani:Georgiev:Tinchev:2018}.
\\
\\
The truth is that the classes of frames determined by all extensions of $\K5$ give rise to a trivial first-order definability problem.
%
%
%
%
Therefore, a second question comes: characterize those extensions of $\K5$ such that the classes of frames they determine give rise to an undecidable modal definability problem.
We answer this question in Section~\ref{section:about:computability:of:modal:definability} where we
also prove that one can solve in non-deterministic exponential time the following decision problem: determine whether the problem of deciding the modal definability of sentences with respect to a given Euclidean modal logic is decidable.
\\
\\
As for the rest of the paper, it is organized as follows.
Although we assume the reader is at home with basic tools and techniques in Propositional Modal Logic and Model Theory of Classical First-Order Logic, we include in Sections~\ref{section:about:preliminaries}--\ref{section:first:order:syntax:and:semantics} the needed material about these topics.
The definability problems are introduced in Section~\ref{section:definability:problems}.
Standard theorems about relative interpretations and relativizations are adapted to our context in Sections~\ref{section:about:relative:interpretations} and~\ref{section:about:relativization}.
Theorems~\ref{proposition:finale} and~\ref{proposition:finale:correspondence:problem} state in Section~\ref{section:about:computability:of:modal:definability} the main results of this paper.
This paper includes the proofs of some of our results.
Some of these proofs are relatively simple and we have included them here just for the sake of the completeness.
\section{Preliminaries}\label{section:about:preliminaries}
For all $k,l{\in}\Z$, let $\lsem k,l\rsem{=}\{m{\in}\Z:\ k{\leq}m{\leq}l\}$.
\\
\\
For all sets $S$, $\wp(S)$ denotes the set of all subsets of $S$.
\\
\\
For all sets $S$, ${\parallel}S{\parallel}$ denotes the cardinality of $S$.
%
%
\\
\\
For all sets $A,B$, for all functions $\rho:\ A{\longrightarrow}\wp(B)$ and for all $D{\in}\wp(B)$, let $\rho^{{-}1}(D){=}\{s{\in}A:\ \rho(s){=}D\}$.
%
%
\\
\\
For all sets $W$, for all binary relations $R$ on $W$ and for all $s,t{\in}W$, let $R(s){=}\{w{\in}
$\linebreak$
W:\ s{R}w\}$ and $R^{{-}1}(t){=}\{w{\in}W:\ w{R}t\}$.
%
%
%
%
%
%
%
%
%
%
%
%
%
%
%
%
%
%
%
%
%
%
%
%
%
%
\\
\\
Let $\mathbf{N}^{+}{=}\N{\setminus}\{0\}$ and $\mathbf{N}^{-}{=}\N{\cup}\{{-}1\}$.
%
%
%
%
%
%
\\
\\
Obviously, for all $m{\in}\mathbf{N}^{+}{\cup}\{0\}$ and for all $n{\in}\mathbf{N}^{-}$, $(m,n){\in}\{(0,0)\}{\cup}(\mathbf{N}^{+}{\times}\mathbf{N}^{-})$ if and only if either $m{\not=}0$, or $n{=}0$.
\\
\\
Let $\ll$ be the well-founded partial order on $\mathbf{N}^{+}{\times}\mathbf{N}^{-}$ such that for all $m,m^{\prime}{\in}\mathbf{N}^{+}$ and for all $n,n^{\prime}{\in}\mathbf{N}^{-}$,
%
%
%
%
%
%
%
%
%
%
%
%
%
%
%
%
$(m,n){\ll}(m^{\prime},n^{\prime})$ if and only if $m{\leq}m^{\prime}$ and $n{\leq}n^{\prime}$.
%
%
%
%
%
%
%
%
%
%
%
%
%
%
%
%
%
%
%
%
%
%
%
%
%
%
%
%
%
%
%
%
%
%
\begin{lemma}\label{lemma:property:of:the:well:founded:order:between:pairs}
For all $m{\in}\mathbf{N}^{+}$ and for all $n{\in}\mathbf{N}^{-}$, if $m{\geq}2$ and $n{\geq}2$ then for all $m^{\prime}{\in}\mathbf{N}^{+}$ and for all $n^{\prime}{\in}\mathbf{N}^{-}$, if $m^{\prime}{<}m$ then $(m^{\prime},2){\ll}(m,n)$ and if $n^{\prime}{<}n$ then $(2,n^{\prime}){\ll}(m,n)$.
%
%
\end{lemma}
%
%
For all $m{\in}\mathbf{N}^{+}$ and for all $n{\in}\mathbf{N}^{-}$, let
\begin{itemize}
\item $\pi(m,n){=}\{(m^{\prime},n^{\prime}):\ m^{\prime}{\in}\mathbf{N}^{+}$, $n^{\prime}{\in}\mathbf{N}^{-}$ and $(m^{\prime},n^{\prime}){\ll}(m,n)\}$.
\end{itemize}
A subset $\mathtt{S}$ of $\mathbf{N}^{+}{\times}\mathbf{N}^{-}$ is {\it closed}\/ if for all $m{\in}\mathbf{N}^{+}$ and for all $n{\in}\mathbf{N}^{-}$, if $(m,n){\in}\mathtt{S}$ then $\pi(m,n){\subseteq}\mathtt{S}$.
\\
\\
Obviously, for all $m{\in}\mathbf{N}^{+}$ and for all $n{\in}\mathbf{N}^{-}$, $\pi(m,n)$ is a closed subset of $\mathbf{N}^{+}{\times}\mathbf{N}^{-}$.
\begin{lemma}\label{lemma:ll:finitely:many:smaller:pairs}
For all $m{\in}\mathbf{N}^{+}$ and for all $n{\in}\mathbf{N}^{-}$, $\pi(m,n)$ is finite.
\end{lemma}
%
%
%
%
%
%
%
%
\begin{lemma}\label{lemma:if:S:L:setminus:some:pairs:is:infinite:then:two:consequences}
For all closed subsets $\mathtt{S}$ of $\mathbf{N}^{+}{\times}\mathbf{N}^{-}$, if $\mathtt{S}{\setminus}((\{1\}{\times}\mathbf{N}^{-}){\cup}(\mathbf{N}^{+}{\times}\{{-}1,0,
$\linebreak$
1\}))$ is infinite then either $\{2\}{\times}\mathbf{N}^{-}{\subseteq}\mathtt{S}$, or $\mathbf{N}^{+}{\times}\{2\}{\subseteq}\mathtt{S}$.
\end{lemma}
\begin{proof}
Let $\mathtt{S}$ be a closed subset of $\mathbf{N}^{+}{\times}\mathbf{N}^{-}$.
Suppose $\mathtt{S}{\setminus}((\{1\}{\times}\mathbf{N}^{-}){\cup}(\mathbf{N}^{+}{\times}
$\linebreak$
\{{-}1,0,1\}))$ is infinite.
For the sake of the contradiction, suppose neither $\{2\}{\times}
$\linebreak$
\mathbf{N}^{-}{\subseteq}\mathtt{S}$, nor $\mathbf{N}^{+}{\times}\{2\}{\subseteq}\mathtt{S}$.
Hence, let $m{\in}\mathbf{N}^{+}$ and $n{\in}\mathbf{N}^{-}$ be such that $(m,2){\not\in}\mathtt{S}$ and $(2,n){\not\in}\mathtt{S}$.
Since $\mathtt{S}{\setminus}((\{1\}{\times}\mathbf{N}^{-}){\cup}(\mathbf{N}^{+}{\times}\{{-}1,0,1\}))$ is infinite, then by
\linebreak
Lemma~\ref{lemma:ll:finitely:many:smaller:pairs}, $\mathtt{S}{\setminus}((\{1\}{\times}\mathbf{N}^{-}){\cup}(\mathbf{N}^{+}{\times}\{{-}1,0,1\})){\not\subseteq}\pi(m,n)$.
Thus, let $m^{\prime}{\in}\mathbf{N}^{+}$ and $n^{\prime}{\in}\mathbf{N}^{-}$ be such that $(m^{\prime},n^{\prime}){\in}\mathtt{S}{\setminus}((\{1\}{\times}\mathbf{N}^{-}){\cup}(\mathbf{N}^{+}{\times}\{{-}1,0,1\}))$ and $(m^{\prime},n^{\prime}){\not\ll}
$\linebreak$
(m,n)$.
Consequently, $(m^{\prime},n^{\prime}){\in}\mathtt{S}$, $m^{\prime}{>}1$ and $n^{\prime}{>}1$.
Hence, $m^{\prime}{\geq}2$ and $n^{\prime}{\geq}2$.
Since $(m^{\prime},n^{\prime}){\not\ll}(m,n)$, then either $m{<}m^{\prime}$, or $n{<}n^{\prime}$.
Since $m^{\prime}{\geq}2$ and $n^{\prime}{\geq}2$, then by Lemma~\ref{lemma:property:of:the:well:founded:order:between:pairs}, either $(m,2){\ll}(m^{\prime},n^{\prime})$, or $(2,n){\ll}(m^{\prime},n^{\prime})$.
Since $\mathtt{S}$ is a closed subset of $\mathbf{N}^{+}{\times}\mathbf{N}^{-}$ and $(m^{\prime},n^{\prime}){\in}\mathtt{S}$, then either $(m,2){\in}\mathtt{S}$, or $(2,n){\in}\mathtt{S}$: a contradiction.
Thus, either $\{2\}{\times}\mathbf{N}^{-}{\subseteq}\mathtt{S}$, or $\mathbf{N}^{+}{\times}\{2\}{\subseteq}\mathtt{S}$.
\medskip
\end{proof}
\section{Frames, galaxies and flowers}\label{section:preliminary:definitions}
In this section, we introduce $3$~types of relational structures: frames, galaxies and flowers.
\begin{figure}[ht]
\begin{center}
\begin{picture}(-50,50)(-50,50)
\put(-120.00,50.00){\circle*{3}}
\put(-115.00,50.00){$f$}
\put(-160.00,50.00){\circle*{3}}
\put(-155.00,50.00){$e$}
\put(-200.00,50.00){\circle*{3}}
\put(-195.00,50.00){$d$}
\put(-120.00,90.00){\circle*{3}}
\put(-115.00,90.00){$c$}
\put(-160.00,90.00){\circle*{3}}
\put(-155.00,90.00){$b$}
\put(-200.00,90.00){\circle*{3}}
\put(-195.00,90.00){$a$}
\put(-160.00,90.00){\vector(+1,-1){40.00}}
\put(-160.00,90.00){\vector(-1,-1){40.00}}
\put(-120.00,90.00){\vector(-1,-1){40.00}}
\put(-120.00,90.00){\vector(0,-1){40.00}}
\put(-220.00,40.00){\line(+1,0){120.00}}
\put(-220.00,60.00){\line(+1,0){120.00}}
\put(-220.00,40.00){\line(0,+1){20.00}}
\put(-100.00,40.00){\line(0,+1){20.00}}
\put(+40.00,50.00){\circle*{3}}
\put(45.00,50.00){$l$}
\put(00.00,50.00){\circle*{3}}
\put(05.00,50.00){$k$}
\put(-40.00,50.00){\circle*{3}}
\put(-35.00,50.00){$j$}
\put(+40.00,90.00){\circle*{3}}
\put(+45.00,90.00){$i$}
\put(00.00,90.00){\circle*{3}}
\put(05.00,90.00){$h$}
\put(-40.00,90.00){\circle*{3}}
\put(-35.00,90.00){$g$}
\put(-40.00,90.00){\vector(0,-1){40.00}}
\put(00.00,90.00){\vector(+1,-1){40.00}}
\put(00.00,90.00){\vector(-1,-1){40.00}}
\put(+40.00,90.00){\vector(-1,-1){40.00}}
\put(00.00,90.00){\vector(0,-1){40.00}}
\put(-60.00,40.00){\line(+1,0){120.00}}
\put(-60.00,60.00){\line(+1,0){120.00}}
\put(-60.00,40.00){\line(0,+1){20.00}}
\put(+60.00,40.00){\line(0,+1){20.00}}
\end{picture}
\caption{Example of an Euclidean frame where, on one hand, $d$, $e$ and $f$ are pairwise connected and, on the other hand, $j$, $k$ and $l$ are pairwise connected.}
\end{center}
\end{figure}
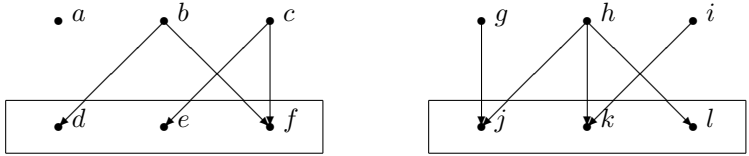
\paragraph{Frames}
A {\it frame}\/ is a tuple of the form $(W,R)$ where $W$ is a non-empty set and $R$ is a binary relation on $W$.
\\
\\
A frame $(W,R)$ is {\it irreflexive}\/ if for all $s{\in}W$, not $s{R}s$.
A frame $(W,R)$ is {\it symmetric}\/ if for all $s,t{\in}W$, if $s{R}t$ then $t{R}s$.
A frame $(W,R)$ is {\it transitive}\/ if for all $s,t,u{\in}W$, if $s{R}t$ and $t{R}u$ then $s{R}u$.
A frame $(W,R)$ is {\it Euclidean}\/ if for all $s,t,u{\in}W$, if $s{R}t$ and $s{R}u$ then $t{R}u$ and $u{R}t$.
%
%
%
%
\\
\\
Obviously, for all Euclidean frames $(W,R)$ and for all $w{\in}W$, if $R(w){=}\emptyset$ then $R^{{-}1}(w){=}\emptyset$.
%
%
%
%
%
%
\paragraph{Dusts, roots and kernels}
%
%
Let the {\it dust of the Euclidean frame $(W,R)$}\/ be the set (denoted $\dust(W,R)$) of all $w$ in $W$ such that $R(w){=}\emptyset$.
Let the {\it root of the Euclidean frame $(W,R)$}\/ be the set (denoted $\root(W,R)$) of all $w$ in $W$ such that $R(w){\not=}\emptyset$ and $R^{{-}1}(w){=}\emptyset$.
Let the {\it kernel of the Euclidean frame $(W,R)$}\/ be the set (denoted $\kernel(W,R)$) of all $w$ in $W$ such that $R^{{-}1}(w){\not=}\emptyset$.
\begin{figure}[ht]
\begin{center}
\begin{picture}(-50,50)(-50,50)
\put(-40.00,50.00){\circle*{3}}
\put(-35.00,50.00){$f$}
\put(-80.00,50.00){\circle*{3}}
\put(-75.00,50.00){$e$}
\put(-120.00,50.00){\circle*{3}}
\put(-115.00,50.00){$d$}
\put(-40.00,90.00){\circle*{3}}
\put(-35.00,90.00){$c$}
\put(-80.00,90.00){\circle*{3}}
\put(-75.00,90.00){$b$}
\put(-120.00,90.00){\circle*{3}}
\put(-115.00,90.00){$a$}
\put(-80.00,90.00){\vector(+1,-1){40.00}}
\put(-80.00,90.00){\vector(-1,-1){40.00}}
\put(-40.00,90.00){\vector(-1,-1){40.00}}
\put(-40.00,90.00){\vector(0,-1){40.00}}
\put(-140.00,40.00){\line(+1,0){120.00}}
\put(-140.00,60.00){\line(+1,0){120.00}}
\put(-140.00,40.00){\line(0,+1){20.00}}
\put(-20.00,40.00){\line(0,+1){20.00}}
\end{picture}
\caption{Example of an Euclidean frame with dust $\{a\}$, root $\{b,c\}$ and kernel $\{d,e,f\}$.}
\end{center}
\end{figure}
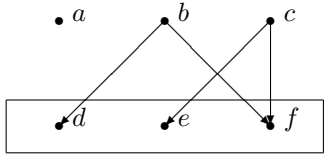
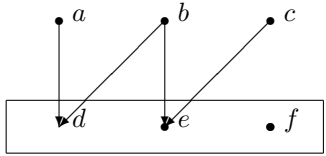
\begin{figure}[ht]
\begin{center}
\begin{picture}(-50,50)(-50,50)
\put(-40.00,50.00){\circle*{3}}
\put(-35.00,50.00){$f$}
\put(-80.00,50.00){\circle*{3}}
\put(-75.00,50.00){$e$}
\put(-40.00,50.00){\circle*{3}}
\put(-115.00,50.00){$d$}
\put(-40.00,90.00){\circle*{3}}
\put(-35.00,90.00){$c$}
\put(-80.00,90.00){\circle*{3}}
\put(-75.00,90.00){$b$}
\put(-120.00,90.00){\circle*{3}}
\put(-115.00,90.00){$a$}
\put(-120.00,90.00){\vector(0,-1){40.00}}
\put(-80.00,90.00){\vector(-1,-1){40.00}}
\put(-40.00,90.00){\vector(-1,-1){40.00}}
\put(-80.00,90.00){\vector(0,-1){40.00}}
\put(-140.00,40.00){\line(+1,0){120.00}}
\put(-140.00,60.00){\line(+1,0){120.00}}
\put(-140.00,40.00){\line(0,+1){20.00}}
\put(-20.00,40.00){\line(0,+1){20.00}}
\end{picture}
\caption{Example of an Euclidean frame with dust $\emptyset$, root $\{a,b,c\}$ and kernel $\{d,e,f\}$.}
\end{center}
\end{figure}
%
%
%
%
%
%
%
%
%
%
%
%
%
%
%
%
%
%
%
%
%
%
%
%
%
%
%
%
%
%
%
%
\\
\\
Obviously, for all Euclidean frames $(W,R)$, $\dust(W,R)$, $\root(W,R)$ and
\linebreak$
\kernel(W,R)$ constitute a partition of $W$.
%
%
%
%
%
%
%
%
%
%
%
%
%
%
%
%
%
%
%
%
%
%
%
%
%
%
%
%
%
%
%
%
%
%
\paragraph{Indexed families of Euclidean frames}
The {\it size of an indexed family $\{(W_{i},
$\linebreak$
R_{i}):\ i{\in}I\}$ of Euclidean frames}\/ is the cardinality of $I$.
An indexed family $\{(W_{i},R_{i}):\ i{\in}I\}$ of Euclidean frames is {\it non-empty}\/ if its size is positive.
\\
\\
A non-empty indexed family $\{(W_{i},R_{i}):\ i{\in}I\}$ of Euclidean frames is {\it disjoint}\/ if for all $i,j{\in}I$, if $i{\not=}j$ then $W_{i}{\cap}W_{j}{=}\emptyset$.
\\
\\
The {\it union of a disjoint indexed family $\{(W_{i},R_{i}):\ i{\in}I\}$ of Euclidean frames}\/ is the Euclidean frame $(W,R)$ such that $W{=}\bigcup\{W_{i}:\ i{\in}I\}$ and $R{=}\bigcup\{R_{i}:\ i{\in}I\}$.
%
%
%
%
%
%
%
%
%
%
\paragraph{Dust-finite indexed families of Euclidean frames}
A disjoint indexed family $\{(W_{i},R_{i}):\ i{\in}I\}$ of Euclidean frames is {\it dust-finite}\/ if for all $i{\in}I$, $\dust(W_{i},
$\linebreak$
R_{i})$ is finite.
For all $k{\in}\N$, a dust-finite indexed family $\{(W_{i},R_{i}):\ i{\in}I\}$ of Euclidean frames is {\it $k$-dust-bounded}\/ if for all $i{\in}I$, ${\parallel}\dust(W_{i},R_{i}){\parallel}{\leq}k$.
A dust-finite indexed family $\{(W_{i},R_{i}):\ i{\in}I\}$ of Euclidean frames is {\it dust-bounded}\/ if there exists $k{\in}\N$ such that $\{(W_{i},R_{i}):\ i{\in}I\}$ is $k$-dust-bounded.
%
%
%
%
%
%
%
%
%
%
\paragraph{Root-finite indexed families of Euclidean frames}
A disjoint indexed family $\{(W_{i},R_{i}):\ i{\in}I\}$ of Euclidean frames is {\it root-finite}\/ if for all $i{\in}I$, $\root(W_{i},
$\linebreak$
R_{i})$ is finite.
For all $k{\in}\N$, a root-finite indexed family $\{(W_{i},R_{i}):\ i{\in}I\}$ of Euclidean frames is {\it $k$-root-bounded}\/ if for all $i{\in}I$, ${\parallel}\root(W_{i},R_{i}){\parallel}{\leq}k$.
A root-finite indexed family $\{(W_{i},R_{i}):\ i{\in}I\}$ of Euclidean frames is {\it root-bounded}\/ if there exists $k{\in}\N$ such that $\{(W_{i},R_{i}):\ i{\in}I\}$ is $k$-root-bounded.
\paragraph{Kernel-finite indexed families of Euclidean frames}
A disjoint indexed family $\{(W_{i},R_{i}):\ i{\in}I\}$ of Euclidean frames is {\it kernel-finite}\/ if for all $i{\in}I$, $\kernel(W_{i},R_{i})$ is finite.
For all $k{\in}\N$, a kernel-finite indexed family $\{(W_{i},R_{i}):\ i{\in}I\}$ of Euclidean frames is {\it $k$-kernel-bounded}\/ if for all $i{\in}I$, ${\parallel}\kernel(W_{i},R_{i}){\parallel}{\leq}k$.
A kernel-finite indexed family $\{(W_{i},R_{i}):\ i{\in}I\}$ of Euclidean frames is {\it kernel-bounded}\/ if there exists $k{\in}\N$ such that $\{(W_{i},R_{i}):\ i{\in}I\}$ is $k$-kernel-bounded.
%
%
%
%
%
%
%
%
%
%
\paragraph{Galaxies}
For all sets $A,B$ and for all functions $\rho:\ A{\longrightarrow}\wp(B)$, if $A{\cap}B{=}\emptyset$ and $A{\cup}B{\not=}\emptyset$ then let ${\mathcal F}_{A,B}^{\rho}{=}(W_{A,B}^{\rho},R_{A,B}^{\rho})$ be the Euclidean frame (called {\it galaxy}) such that
\begin{itemize}
\item $W_{A,B}^{\rho}{=}A{\cup}B$,
\item $R_{A,B}^{\rho}{=}\bigcup\{\{s\}{\times}\rho(s):\ s{\in}A\}{\cup}(B{\times}B)$.
\end{itemize}
Obviously, for all galaxies ${\mathcal F}_{A,B}^{\rho}$, $\dust({\mathcal F}_{A,B}^{\rho}){=}\rho^{{-}1}(\emptyset)$, $\root({\mathcal F}_{A,B}^{\rho}){=}A{\setminus}\rho^{{-}1}(\emptyset)$ and $\kernel({\mathcal F}_{A,B}^{\rho}){=}B$.
\\
\\
The {\it upper part of the galaxy ${\mathcal F}_{A,B}^{\rho}$}\/ is the set $A$.
The {\it lower part of the galaxy ${\mathcal F}_{A,B}^{\rho}$}\/ is the set $B$.
\\
\\
A galaxy ${\mathcal F}_{A,B}^{\rho}$ is {\it finite}\/ if $A$ and $B$ are finite.
\\
\\
A galaxy ${\mathcal F}_{A,B}^{\rho}$ is {\it headed}\/ if there exists $s{\in}A$ such that $\rho(s){\not=}\emptyset$.
\\
\\
A galaxy ${\mathcal F}_{A,B}^{\rho}$ is {\it simple}\/ if $B{\not=}\emptyset$ and for all $s{\in}A$, either ${\parallel}\rho(s){\parallel}{\leq}1$, or ${\parallel}B{\setminus}\rho(s){\parallel}{\leq}
$\linebreak$
1$.
\\
\\
Obviously, for all galaxies ${\mathcal F}_{A,B}^{\rho}$, ${\mathcal F}_{A,B}^{\rho}$ is simple if and only if $B{\not=}\emptyset$ and for all $s{\in}A$, either $\rho(s){=}\emptyset$, or $\rho(s){=}B$, or there exists $t{\in}B$ such that $\rho(s){=}\{t\}$, or there exists $t{\in}B$ such that $\rho(s){=}B{\setminus}\{t\}$.
\begin{figure}[ht]
\begin{center}
\begin{picture}(-50,50)(-50,50)
\put(-40.00,50.00){\circle*{3}}
\put(-35.00,50.00){$f$}
\put(-80.00,50.00){\circle*{3}}
\put(-75.00,50.00){$e$}
\put(-120.00,50.00){\circle*{3}}
\put(-115.00,50.00){$d$}
\put(-40.00,90.00){\circle*{3}}
\put(-35.00,90.00){$c$}
\put(-80.00,90.00){\circle*{3}}
\put(-75.00,90.00){$b$}
\put(-120.00,90.00){\circle*{3}}
\put(-115.00,90.00){$a$}
\put(-80.00,90.00){\vector(-1,-1){40.00}}
\put(-80.00,90.00){\vector(0,-1){40.00}}
\put(-80.00,90.00){\vector(+1,-1){40.00}}
\put(-40.00,90.00){\vector(0,-1){40.00}}
\put(-40.00,90.00){\vector(0,-1){40.00}}
\put(-140.00,40.00){\line(+1,0){120.00}}
\put(-140.00,60.00){\line(+1,0){120.00}}
\put(-140.00,40.00){\line(0,+1){20.00}}
\put(-20.00,40.00){\line(0,+1){20.00}}
\end{picture}
\caption{Example of a simple galaxy with dust $\{a\}$, root $\{b,c\}$ and kernel $\{d,e,f\}$ and where $A{=}\{a,b,c\}$, $B{=}\{d,e,f\}$, $\rho(a){=}\emptyset$, $\rho(b){=}\{d,e,f\}$ and $\rho(c){=}\{f\}$.}
\end{center}
\end{figure}
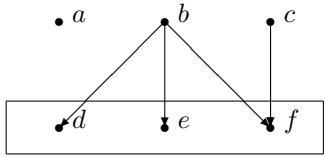
%
%
%
%
%
%
%
%
\begin{lemma}\label{lemma:every:frame:disjoint:union:galaxies}
Every Euclidean frame is the union of a disjoint indexed family of galaxies.
\end{lemma}
%
%
%
%
%
%
%
%
%
%
Let ${\mathcal K}_{2}$ be the class of all galaxies ${\mathcal F}^{\rho}_{A,B}$ such that ${\parallel}A{\parallel}{\geq}4$, ${\parallel}B{\parallel}{\geq}4$ and for all $s{\in}A$, ${\parallel}\rho(s){\parallel}{=}2$.
Let ${\mathcal L}_{2}$ be the class of all galaxies ${\mathcal F}^{\rho}_{A,B}$ such that ${\parallel}A{\parallel}{\geq}4$, ${\parallel}B{\parallel}{\geq}4$ and for all $s{\in}A$, ${\parallel}B{\setminus}\rho(s){\parallel}{=}2$.
\begin{figure}[ht]
\begin{center}
\begin{picture}(-50,50)(-50,50)
\put(-20.00,50.00){\circle*{3}}
\put(-15.00,50.00){$h$}
\put(-60.00,50.00){\circle*{3}}
\put(-55.00,50.00){$g$}
\put(-100.00,50.00){\circle*{3}}
\put(-95.00,50.00){$f$}
\put(-140.00,50.00){\circle*{3}}
\put(-135.00,50.00){$e$}
\put(-20.00,90.00){\circle*{3}}
\put(-15.00,90.00){$d$}
\put(-60.00,90.00){\circle*{3}}
\put(-55.00,90.00){$c$}
\put(-100.00,90.00){\circle*{3}}
\put(-95.00,90.00){$b$}
\put(-140.00,90.00){\circle*{3}}
\put(-135.00,90.00){$a$}
\put(-60.00,90.00){\vector(+1,-1){40.00}}
\put(-60.00,90.00){\vector(-1,-1){40.00}}
\put(-20.00,90.00){\vector(-1,-1){40.00}}
\put(-20.00,90.00){\vector(0,-1){40.00}}
\put(-140.00,90.00){\vector(0,-1){40.00}}
\put(-140.00,90.00){\vector(+1,-1){40.00}}
\put(-100.00,90.00){\vector(-1,-1){40.00}}
\put(-100.00,90.00){\vector(+1,-1){40.00}}
\put(-160.00,40.00){\line(+1,0){160.00}}
\put(-160.00,60.00){\line(+1,0){160.00}}
\put(-160.00,40.00){\line(0,+1){20.00}}
\put(-00.00,40.00){\line(0,+1){20.00}}
\end{picture}
\caption{Example of a galaxy in ${\mathcal K}_{2}$ and ${\mathcal L}_{2}$, with dust $\emptyset$, root $\{a,b,c,d\}$ and kernel $\{e,f,g,h\}$ and where $A{=}\{a,b,c,d\}$, $B{=}\{e,f,g,h\}$, $\rho(a){=}\{e,f\}$, $\rho(b){=}\{e,g\}$, $\rho(c){=}\{f,h\}$ and $\rho(d){=}\{g,h\}$.}
\end{center}
\end{figure}
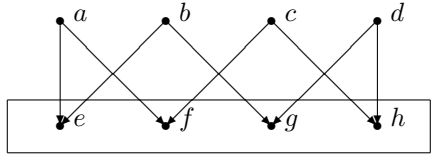
\paragraph{Flowers}
Let ${\mathcal F}_{0}^{0}{=}(W_{0}^{0},R_{0}^{0})$ be the Euclidean frame (called {\it isolated point}) such that
\begin{itemize}
\item $W_{0}^{0}{=}\{0\}$,
\item $R_{0}^{0}{=}\emptyset$.
\end{itemize}
%
%
%
%
%
%
%
%
%
%
%
%
%
%
%
%
%
%
%
%
Obviously, $\dust({\mathcal F}_{0}^{0}){=}\{0\}$, $\root({\mathcal F}_{0}^{0}){=}\emptyset$ and $\kernel({\mathcal F}_{0}^{0}){=}\emptyset$.
%
%
%
%
%
%
Moreover, the isolated point is a finite galaxy.
%
%
%
%
%
%
In other respect, ${\mathcal F}_{0}^{0}$ is not simple.
%
%
\\
\\
%
%
For all $m{\in}\mathbf{N}^{+}$ and for all $n{\in}\mathbf{N}^{-}$, let ${\mathcal F}_{m}^{n}{=}(W_{m}^{n},R_{m}^{n})$ be the Euclidean frame (called {\it flower}) such that if $n{\in}\N$ then
\begin{itemize}
\item $W_{m}^{n}{=}\lsem0,m{+}n\rsem$,
\item $R_{m}^{n}{=}(\{0\}{\times}\lsem1,m\rsem){\cup}(\lsem1,m{+}n\rsem{\times}\lsem1,m{+}n\rsem)$,
\end{itemize}
and if $n{=}{-}1$ then
\begin{itemize}
\item $W_{m}^{n}{=}\lsem1,m\rsem$,
\item $R_{m}^{n}{=}\lsem1,m\rsem{\times}\lsem1,m\rsem$.
\end{itemize}
Obviously, for all flowers ${\mathcal F}_{m}^{n}$, if $n{\in}\N$ then $\dust({\mathcal F}_{m}^{n}){=}\emptyset$, $\root({\mathcal F}_{m}^{n}){=}\{0\}$ and $\kernel({\mathcal F}_{m}^{n}){=}\lsem1,m{+}n\rsem$ and if $n{=}{-}1$ then $\dust({\mathcal F}_{m}^{n}){=}\emptyset$, $\root({\mathcal F}_{m}^{n}){=}\emptyset$ and
\linebreak$
\kernel({\mathcal F}_{m}^{n}){=}\lsem1,m\rsem$.
Moreover, every flower is a finite galaxy.
\begin{figure}[ht]
\begin{center}
\begin{picture}(-50,50)(-50,50)
\put(-120.00,50.00){\circle*{3}}
\put(-115.00,50.00){$3$}
\put(-160.00,50.00){\circle*{3}}
\put(-155.00,50.00){$2$}
\put(-200.00,50.00){\circle*{3}}
\put(-195.00,50.00){$1$}
\put(-160.00,90.00){\circle*{3}}
\put(-155.00,90.00){$0$}
\put(-160.00,90.00){\vector(-1,-1){40.00}}
\put(-160.00,90.00){\vector(0,-1){40.00}}
\put(-220.00,40.00){\line(+1,0){120.00}}
\put(-220.00,60.00){\line(+1,0){120.00}}
\put(-220.00,40.00){\line(0,+1){20.00}}
\put(-100.00,40.00){\line(0,+1){20.00}}
\put(40.00,50.00){\circle*{3}}
\put(45.00,50.00){$3$}
\put(00.00,50.00){\circle*{3}}
\put(+05.00,50.00){$2$}
\put(-40.00,50.00){\circle*{3}}
\put(-35.00,50.00){$1$}
\put(-60.00,40.00){\line(+1,0){120.00}}
\put(-60.00,60.00){\line(+1,0){120.00}}
\put(-60.00,40.00){\line(0,+1){20.00}}
\put(+60.00,40.00){\line(0,+1){20.00}}
\end{picture}
\caption{Flowers ${\mathcal F}_{2}^{1}$ and ${\mathcal F}_{3}^{-1}$.}
\end{center}
\end{figure}
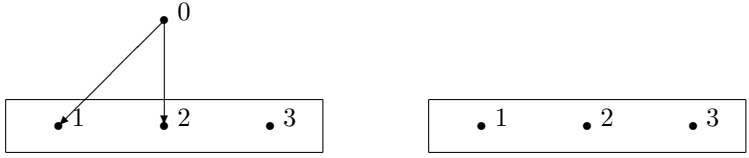
%
%
%
%
%
%
%
%
%
%
%
%
%
%
\begin{lemma}\label{lemma:about:flowers:simple:galaxies}
For all flowers ${\mathcal F}_{m}^{n}$, ${\mathcal F}_{m}^{n}$ is simple if and only if either $m{=}1$, or $n{=}{-}1$, or $n{=}0$, or $n=1$.
\end{lemma}
Notice that for all $m{\in}\mathbf{N}^{+}{\cup}\{0\}$ and for all $n{\in}\mathbf{N}^{-}$, if either $m{\not=}0$, or $n{=}0$ then either ${\mathcal F}_{m}^{n}$ is the isolated point, or ${\mathcal F}_{m}^{n}$ is a flower.
\section{Generated subframes, bounded morphic images and Ehrenfeucht-Fra\"\i ss\'e games}
In this section, we introduce the standard concepts of generated subframes, bounded morphic images and Ehrenfeucht-Fra\"\i ss\'e games.
\paragraph{Generated subframes}
Let $(W,R),(W^{\prime},R^{\prime})$ be frames.
\\
\\
$(W,R)$ is a {\it generated subframe of $(W^{\prime},R^{\prime})$}\/ if the following conditions are satisfied:
\begin{itemize}
\item $W{\subseteq}W^{\prime}$,
\item $R{\subseteq}R^{\prime}$,
\item for all $s{\in}W$ and for all $t^{\prime}{\in}W^{\prime}$, if $s{R^{\prime}}t^{\prime}$ then $t^{\prime}{\in}W$ and $s{R}t^{\prime}$.
\end{itemize}
For all $s{\in}W$, let $(W_{s},R_{s})$ be the least generated subframe of $(W,R)$ containing $s$.
For all $s{\in}W$, {\it $s$ generates $(W,R)$}\/ if the following conditions are satisfied:
\begin{itemize}
\item $W_{s}{=}W$,
\item $R_{s}{=}R$.
\end{itemize}
$(W,R)$ is {\it rooted}\/ if there exists $s{\in}W$ such that $s$ generates $(W,R)$.
See~\cite[Chapter~$3$]{Blackburn:deRijke:Venema:2001}.
\\
\\
Obviously, if $(W,R)$ is finite, rooted and Euclidean then either $(W,R)$ is isomorphic with the isolated point, or $(W,R)$ is isomorphic with a flower.
\begin{lemma}\label{lemma:generated:subframe:from:a:single:point:has:specific:shapes}
If $(W,R)$ is Euclidean then for all $s{\in}W$, either $R_{s}{=}W_{s}{\times}W_{s}$, or $W_{s}{=}\{s\}{\cup}C$ and $R_{s}{=}(\{s\}{\times}B){\cup}(C{\times}C)$ for some sets $B,C$ such that either $B{\not=}\emptyset$, or $C{=}\emptyset$, $B{\subseteq}C$ and $s{\not\in}C$.
\end{lemma}
\begin{lemma}\label{lemma:about:generated:subframe:and:disjoint:unions}
Let $\{(W_{i},R_{i}):\ i{\in}I\}$ be a disjoint indexed family of Euclidean frames.
If $(W,R)$ is the union of $\{(W_{i},R_{i}):\ i{\in}I\}$ then for all $i{\in}I$, $(W_{i},R_{i})$ is a generated subframe of $(W,R)$.
\end{lemma}
%
%
%
%
%
%
%
%
%
%
%
%
%
%
%
%
%
%
%
%
%
%
%
%
%
%
%
%
%
%
%
%
%
%
%
%
%
%
\paragraph{Bounded morphic images}
Let $q{\geq}3$.
\\
\\
Let $(W,R),(W^{\prime},R^{\prime})$ be frames.
\\
\\
A function $f:\ W^{\prime}{\longrightarrow}W$ is a {\it bounded morphism from $(W^{\prime},R^{\prime})$ to $(W,R)$}\/ if the following conditions are satisfied:
\begin{itemize}
\item for all $s^{\prime},t^{\prime}{\in}W^{\prime}$, if $s^{\prime}{R^{\prime}}t^{\prime}$ then $f(s^{\prime}){R}f(t^{\prime})$,
\item for all $s^{\prime}{\in}W^{\prime}$ and for all $t{\in}W$, if $f(s^{\prime}){R}t$ then there exists $t^{\prime}{\in}W^{\prime}$ such that $s^{\prime}{R^{\prime}}t^{\prime}$ and $f(t^{\prime})=t$.
\end{itemize}
$(W,R)$ is a {\it bounded morphic image of $(W^{\prime},R^{\prime})$}\/ if there exists a surjective bounded morphism from $(W^{\prime},R^{\prime})$ to $(W,R)$.
See~\cite[Chapter~$3$]{Blackburn:deRijke:Venema:2001}.
%
%
%
%
%
%
%
%
%
%
%
%
%
%
%
%
\begin{lemma}\label{lemma:about:bounded:morphisms:in:the:situation:of F:m:n:frames}
For all flowers ${\mathcal F}_{m}^{n},{\mathcal F}_{m^{\prime}}^{n^{\prime}}$, ${\mathcal F}_{m}^{n}$ is a bounded morphic image of ${\mathcal F}_{m^{\prime}}^{n^{\prime}}$ if and only if $(m,n){\ll}(m^{\prime},n^{\prime})$.
\end{lemma}
\begin{lemma}\label{lemma:about:k:s:new:element:bmi:of:W:R}
Let $k$ be an integer such that $k{\geq}4$.
Let $B,C$ be sets such that $B{\subseteq}C$, ${\parallel}C{\parallel}{>}k$, ${\parallel}B{\parallel}{>}1$ and ${\parallel}C{\setminus}B{\parallel}{>}1$.
Let $m{=}\min\{{\parallel}B{\parallel},k\}$ and $n{=}\min\{{\parallel}C{\setminus}B{\parallel},k\}$.
Let $s$ be a new element.
If $W{=}\{s\}{\cup}C$ and $R{=}(\{s\}{\times}B){\cup}(C{\times}C)$ then ${\mathcal F}_{m}^{n}$ is a bounded morphic image of $(W,R)$.
\end{lemma}
\begin{lemma}\label{lemma:about:partitions:q:greater:smaller}
If $R{=}W{\times}W$, $R^{\prime}{=}W^{\prime}{\times}W^{\prime}$ and either ${\parallel}W^{\prime}{\parallel}{\leq}q$ and ${\parallel}W^{\prime}{\parallel}{=}{\parallel}W{\parallel}$, or ${\parallel}W^{\prime}{\parallel}{>}q$ and ${\parallel}W{\parallel}{=}q$ then $(W,R)$ is a bounded morphic image of $(W^{\prime},R^{\prime})$.
\end{lemma}
\begin{lemma}\label{lemma:inclusion:Bprime:B:for:generated:subframes:is:bmi}
If $W{=}\{s\}{\cup}B$, $R{=}(\{s\}{\times}B){\cup}(B{\times}B)$, $W^{\prime}{=}\{s^{\prime}\}{\cup}B^{\prime}$ and $R^{\prime}{=}(\{s^{\prime}\}{\times}
$\linebreak$
B^{\prime}){\cup}(B^{\prime}{\times}B^{\prime})$ for some sets $B,B^{\prime}$ and for some elements $s,s^{\prime}$ such that $s{\not\in}B$, $B{\not=}\emptyset$, $s^{\prime}{\not\in}B^{\prime}$, $B^{\prime}{\not=}\emptyset$ and ${\parallel}B{\parallel}{\leq}{\parallel}B^{\prime}{\parallel}$ then $(W,R)$ is a bounded morphic image of $(W^{\prime},R^{\prime})$.
\end{lemma}
\begin{lemma}\label{lemma:inclusion:Bprime:B:for:generated:subframes:is:bmi:bis}
If $W{=}\{s\}{\cup}B$, $R{=}\{(s,u)\}{\cup}(B{\times}B)$, $W^{\prime}{=}\{s^{\prime}\}{\cup}B^{\prime}$ and $R^{\prime}{=}\{(s^{\prime},
$\linebreak$
u^{\prime})\}{\cup}(B^{\prime}{\times}B^{\prime})$ for some sets $B,B^{\prime}$ and for some elements $s,s^{\prime},u,u^{\prime}$ such that $s{\not\in}B$, $u{\in}B$, $B{\not=}\emptyset$, $s^{\prime}{\not\in}B^{\prime}$, $u^{\prime}{\in}B^{\prime}$, $B^{\prime}{\not=}\emptyset$, ${\parallel}B{\parallel}{\leq}{\parallel}B^{\prime}{\parallel}$ and if ${\parallel}B{\parallel}{\leq}{\parallel}2$ then ${\parallel}B{\parallel}{=}
$\linebreak$
{\parallel}B^{\prime}{\parallel}$ then $(W,R)$ is a bounded morphic image of $(W^{\prime},R^{\prime})$.
\end{lemma}
\begin{lemma}\label{lemma:inclusion:Bprime:B:for:generated:subframes:is:bmi:ter}
If $W{=}\{s\}{\cup}B$, $R{=}(\{s\}{\times}(B{\setminus}\{u\})){\cup}(B{\times}B)$, $W^{\prime}{=}\{s^{\prime}\}{\cup}B^{\prime}$ and $R^{\prime}{=}
$\linebreak$
(\{s^{\prime}\}{\times}(B^{\prime}{\setminus}\{u^{\prime}\})){\cup}(B^{\prime}{\times}B^{\prime})$ for some sets $B,B^{\prime}$ and for some elements $s,s^{\prime},u,u^{\prime}$ such that $s{\not\in}B$, $u{\in}B$, $B{\not=}\emptyset$, $s^{\prime}{\not\in}B^{\prime}$, $u^{\prime}{\in}B^{\prime}$, $B^{\prime}{\not=}\emptyset$, ${\parallel}B{\parallel}{\leq}{\parallel}B^{\prime}{\parallel}$ and if ${\parallel}B{\parallel}{\leq}{\parallel}2$ then ${\parallel}B{\parallel}{=}{\parallel}B^{\prime}{\parallel}$ then $(W,R)$ is a bounded morphic image of $(W^{\prime},R^{\prime})$.
\end{lemma}
\begin{lemma}\label{lemma:inclusion:Bprime:B:for:generated:subframes:is:bmi:quater}
If $W{=}B$, $R{=}B{\times}B$, $W^{\prime}{=}B^{\prime}$ and $R^{\prime}{=}B^{\prime}{\times}B^{\prime}$ for some sets $B,B^{\prime}$ such that $B{\not=}\emptyset$, $B^{\prime}{\not=}\emptyset$, ${\parallel}B{\parallel}{\leq}{\parallel}B^{\prime}{\parallel}$ and if ${\parallel}B{\parallel}{\leq}{\parallel}2$ then ${\parallel}B{\parallel}{=}{\parallel}B^{\prime}{\parallel}$ then $(W,R)$ is a bounded morphic image of $(W^{\prime},R^{\prime})$.
\end{lemma}
\begin{lemma}\label{lemma:about:frames:C:B:smaller:greater}
If $W{=}\{s\}{\cup}C$, $R{=}(\{s\}{\times}B){\cup}(C{\times}C)$, $W^{\prime}{=}\{s^{\prime}\}{\cup}C^{\prime}$ and $R^{\prime}{=}(\{s^{\prime}\}{\times}
$\linebreak$
B^{\prime}){\cup}(C^{\prime}{\times}C^{\prime})$ for some sets $B,C,B^{\prime},C^{\prime}$ and for some elements $s,s^{\prime}$ such that either $B{\not=}\emptyset$, or $C{=}\emptyset$, $B{\subseteq}C$, $s{\not\in}C$, either $B^{\prime}{\not=}\emptyset$, or $C^{\prime}{=}\emptyset$, $B^{\prime}{\subseteq}C^{\prime}$, $s^{\prime}{\not\in}C^{\prime}$, either ${\parallel}B^{\prime}{\parallel}{\leq}q$ and ${\parallel}B^{\prime}{\parallel}{=}{\parallel}B{\parallel}$, or ${\parallel}B^{\prime}{\parallel}{>}q$ and ${\parallel}B{\parallel}{=}q$ and either ${\parallel}C^{\prime}{\setminus}B^{\prime}{\parallel}{\leq}q$ and ${\parallel}C^{\prime}{\setminus}B^{\prime}{\parallel}{=}{\parallel}C{\setminus}B{\parallel}$, or ${\parallel}C^{\prime}{\setminus}B^{\prime}{\parallel}{>}q$ and ${\parallel}C{\setminus}B{\parallel}{=}q$ then $(W,R)$ is a bounded morphic image of $(W^{\prime},R^{\prime})$.
\end{lemma}
\begin{lemma}\label{lemma:about:what:happens:when:it:exists:bmi:flower:integer}
If $(W,R)$ is Euclidean then for all $s{\in}W$ and for all flowers ${\mathcal F}_{m}^{n}$, if $n{\in}\N$ and there exists a surjective bounded morphism $f$ from $(W_{s},R_{s})$ to ${\mathcal F}_{m}^{n}$ such that $f(s)$ generates ${\mathcal F}_{m}^{n}$ then $W_{s}{=}\{s\}{\cup}C$ and $R_{s}{=}(\{s\}{\times}B){\cup}(C{\times}C)$ for some sets $B,C$ such that $B{\subseteq}C$, $s{\not\in}C$, ${\parallel}B{\parallel}{\geq}m$ and ${\parallel}C{\setminus}B{\parallel}{\geq}n$.
\end{lemma}
\begin{lemma}\label{lemma:about:what:happens:when:it:exists:bmi:flower:minus:one}
If $(W,R)$ is Euclidean then for all $s{\in}W$ and for all flowers ${\mathcal F}_{m}^{n}$, if $n{=}{-}1$ and there exists a surjective bounded morphism $f$ from $(W_{s},R_{s})$ to ${\mathcal F}_{m}^{n}$ such that $f(s)$ generates ${\mathcal F}_{m}^{n}$ then either $R_{s}{=}W_{s}{\times}W_{s}$ and ${\parallel}W_{s}{\parallel}{\geq}m$, or $W_{s}{=}\{s\}{\cup}C$ and $R_{s}{=}(\{s\}{\times}B){\cup}(C{\times}C)$ for some sets $B,C$ such that $B{\subseteq}C$, $s{\not\in}C$ and ${\parallel}B{\parallel}{\geq}m$.
\end{lemma}
\paragraph{Ehrenfeucht-Fra\"\i ss\'e games}
Let $q{\geq}3$.
\\
\\
Let $(W,R)$ and $(W^{\prime},R^{\prime})$ be Euclidean frames.
\\
\\
The {\it Ehrenfeucht-Fra\"\i ss\'e game ${\mathcal G}_{q}((W,R),(W^{\prime},R^{\prime}))$}\/ is a game played by $2$~{\it players}\/ making $q$~{\it moves.}
For all $q^{\prime}{\in}\lsem1,q\rsem$, in its $q^{\prime}$-th move, the first player picks either a member $s_{q^{\prime}}$ of $(W,R)$, or a member $s_{q^{\prime}}^{\prime}$ of $(W^{\prime},R^{\prime})$.
For all $q^{\prime}{\in}\lsem1,q\rsem$, if the first player picked a member $s_{q^{\prime}}$ of $(W,R)$ in its $q^{\prime}$-th move then the second player picks a member $s_{q^{\prime}}^{\prime}$ of $(W^{\prime},R^{\prime})$ in its $q^{\prime}$-th move, whereas if the first player picked a member $s_{q^{\prime}}^{\prime}$ of $(W^{\prime},R^{\prime})$ in its $q^{\prime}$-th move then the second player picks a member $s_{q^{\prime}}$ of $(W,R)$ in its $q^{\prime}$-th move.
%
%
\\
\\
For all $q^{\prime}{\in}\lsem0,q\rsem$, for all $s_{1},\ldots,s_{q^{\prime}}{\in}W$ and for all $s^{\prime}_{1},\ldots,s^{\prime}_{q^{\prime}}{\in}W^{\prime}$, the couple $((s_{1},\ldots,s_{q^{\prime}}),(s^{\prime}_{1},\ldots,s^{\prime}_{q^{\prime}}))$ is a {\it position of the game.}
When $q^{\prime}{=}0$, this position is called the {\it empty position of the game.}
\\
\\
For all $q^{\prime}{\in}\lsem0,q\rsem$, for all $s_{1},\ldots,s_{q^{\prime}}{\in}W$ and for all $s^{\prime}_{1},\ldots,s^{\prime}_{q^{\prime}}{\in}W^{\prime}$, the second player {\it wins the position $((s_{1},\ldots,s_{q^{\prime}}),(s^{\prime}_{1},\ldots,s^{\prime}_{q^{\prime}}))$ of the game}\/ if for all $q^{\prime\prime},q^{\prime\prime\prime}{\in}
$\linebreak$
\lsem1,q^{\prime}\rsem$,
\begin{itemize}
\item $s_{q^{\prime\prime}}{=}s_{q^{\prime\prime\prime}}$ if and only if $s^{\prime}_{q^{\prime\prime}}{=}s^{\prime}_{q^{\prime\prime\prime}}$,
\item $s_{q^{\prime\prime}}{R}s_{q^{\prime\prime\prime}}$ if and only if $s^{\prime}_{q^{\prime\prime}}{R^{\prime}}s^{\prime}_{q^{\prime\prime\prime}}$.
\end{itemize}
%
%
A {\it strategy of the second player in the game}\/ is a deterministic procedure telling how the second player chooses its moves during the game.
%
%
%
%
A strategy of the second player in the game is {\it winning}\/ if the second player wins each position of the game that is reachable by applying the strategy.
See~\cite[Chapter~$2$]{Ebbinghaus:Flum:1999}.
\begin{lemma}\label{lemma:about:partitions:q:greater:smaller:games}
If $R{=}W{\times}W$, $R^{\prime}{=}W^{\prime}{\times}W^{\prime}$ and either ${\parallel}W^{\prime}{\parallel}{\leq}q$ and ${\parallel}W^{\prime}{\parallel}{=}{\parallel}W{\parallel}$, or ${\parallel}W^{\prime}{\parallel}{>}q$ and ${\parallel}W{\parallel}{=}q$ then the second player has a winning strategy in the Ehrenfeucht-Fra\"\i ss\'e game ${\mathcal G}_{q}((W,R),(W^{\prime},R^{\prime}))$.
\end{lemma}
\begin{lemma}\label{lemma:about:frames:C:B:smaller:greater:games}
If $W{=}\{s\}{\cup}C$, $R{=}(\{s\}{\times}B){\cup}(C{\times}C)$, $W^{\prime}{=}\{s^{\prime}\}{\cup}C^{\prime}$ and $R^{\prime}{=}(\{s^{\prime}\}{\times}
$\linebreak$
B^{\prime}){\cup}(C^{\prime}{\times}C^{\prime})$ for some sets $B,C,B^{\prime},C^{\prime}$ and for some elements $s,s^{\prime}$ such that either $B{\not=}\emptyset$, or $C{=}\emptyset$, $B{\subseteq}C$, $s{\not\in}C$, either $B^{\prime}{\not=}\emptyset$, or $C^{\prime}{=}\emptyset$, $B^{\prime}{\subseteq}C^{\prime}$, $s^{\prime}{\not\in}C^{\prime}$, either ${\parallel}B^{\prime}{\parallel}{\leq}q$ and ${\parallel}B^{\prime}{\parallel}{=}{\parallel}B{\parallel}$, or ${\parallel}B^{\prime}{\parallel}{>}q$ and ${\parallel}B{\parallel}{=}q$ and either ${\parallel}C^{\prime}{\setminus}B^{\prime}{\parallel}{\leq}q$ and ${\parallel}C^{\prime}{\setminus}B^{\prime}{\parallel}{=}{\parallel}C{\setminus}B{\parallel}$, or ${\parallel}C^{\prime}{\setminus}B^{\prime}{\parallel}{>}q$ and ${\parallel}C{\setminus}B{\parallel}{=}q$ then the second player has a winning strategy in the Ehrenfeucht-Fra\"\i ss\'e game ${\mathcal G}_{q}((W,R),(W^{\prime},R^{\prime}))$.
\end{lemma}
\section{Reductions}
In this section, we introduce $3$~types of reductions: the $\alpha$-reduction of a galaxy, the $\gamma$-reduction of a simple galaxy and the $\delta$-reduction of a disjoint indexed family of galaxies.
\\
\\
Let $q{\geq}3$.
\paragraph{$\alpha$-reductions}
The main effect of an $\alpha$-reduction on a galaxy is to reduce the cardinality of its upper part.
%
%
\\
\\
Let ${\mathcal F}_{A,B}^{\rho}$ be a galaxy.
\\
\\
Let $\sim$ be the equivalence relation on $A$ defined by $u{\sim}v$ if and only if $\rho(u){=}\rho(v)$.
For all $u{\in}A$, let $\lbrack u\rbrack_{\sim}$ be the equivalence class of $u$ modulo $\sim$.
\\
\\
For all $u{\in}A$, if ${\parallel}\lbrack u\rbrack_{\sim}{\parallel}{\leq}q$ then let $C_{u}$ be $\lbrack u\rbrack_{\sim}$ else let $C_{u}$ be an arbitrary subset of $\lbrack u\rbrack_{\sim}$ of cardinality $q$.
Let $A^{\prime}$ be $\bigcup\{C_{u}:\ u{\in}A\}$.
Obviously, from each equivalence class modulo $\sim$, $A^{\prime}$ contains at most $q$ elements.
Therefore, if $B$ is finite then $A^{\prime}$ contains at most $q{\times}2^{{\parallel}B{\parallel}}$ elements.
%
%
%
%
%
%
\\
\\
Let $B^{\prime}$ be $B$.
\\
\\
Let ${\mathcal F}_{A^{\prime},B^{\prime}}^{\rho^{\prime}}$ be the galaxy such that for all $u^{\prime}{\in}A^{\prime}$, $\rho^{\prime}(u^{\prime}){=}\rho(u^{\prime})$.
${\mathcal F}_{A^{\prime},B^{\prime}}^{\rho^{\prime}}$ is a {\it $q$-$\alpha$-reduction of ${\mathcal F}_{A,B}^{\rho}$.}
\begin{lemma}\label{lemma:the:alpha:reduction:is:always:dust:finite}
\begin{enumerate}
\item $A^{\prime}{\subseteq}A$ and if ${\parallel}A^{\prime}{\parallel}{\leq}2$ then $A^{\prime}{=}A$,
\item $B^{\prime}{=}B$,
\item for all $s^{\prime}{\in}A^{\prime}$, $\rho^{\prime}(s^{\prime}){=}\rho(s^{\prime})$,
\item ${\parallel}{\rho^{\prime}}^{{-}1}(\emptyset){\parallel}{\leq}q$ and ${\parallel}{\rho^{\prime}}^{{-}1}(B^{\prime}){\parallel}{\leq}q$,
\item for all $t^{\prime}{\in}B^{\prime}$, ${\parallel}{\rho^{\prime}}^{{-}1}(\{t^{\prime}\}){\parallel}{\leq}q$ and for all $t^{\prime}{\in}B^{\prime}$, ${\parallel}{\rho^{\prime}}^{{-}1}(B^{\prime}{\setminus}\{t^{\prime}\}){\parallel}{\leq}q$,
\item if ${\mathcal F}_{A,B}^{\rho}$ is simple then ${\mathcal F}_{A^{\prime},B^{\prime}}^{\rho^{\prime}}$ is simple, whereas if ${\mathcal F}_{A,B}^{\rho}$ is non-simple and $B$ is finite then ${\mathcal F}_{A^{\prime},B^{\prime}}^{\rho^{\prime}}$ is non-simple, $B^{\prime}$ is finite and ${\parallel}A^{\prime}{\parallel}{\leq}q{\times}2^{{\parallel}B{\parallel}}$,
\item ${\mathcal F}_{A^{\prime},B^{\prime}}^{\rho^{\prime}}$ is a bounded morphic image of ${\mathcal F}_{A,B}^{\rho}$,
\item the second player has a winning strategy in the Ehrenfeucht-Fra\"\i ss\'e game ${\mathcal G}_{q}({\mathcal F}_{A,B}^{\rho},{\mathcal F}_{A^{\prime},B^{\prime}}^{\rho^{\prime}})$.
\end{enumerate}
\end{lemma}
\begin{proof}
$\mathbf{(1)}$--$\mathbf{(6)}$~By construction of ${\mathcal F}_{A^{\prime},B^{\prime}}^{\rho^{\prime}}$.
\\
\\
$\mathbf{(7)}$~Let $g:\ A{\longrightarrow}A^{\prime}$ be a function such that for all $s{\in}A$, either $s{\in}A^{\prime}$ and $g(s){=}s$, or $s{\not\in}A^{\prime}$ and $g(s){\sim}s$.
Notice that the existence of such a function is a consequence of the definition of $A^{\prime}$.
Notice also that $g$ is surjective.
Moreover, for all $s{\in}A$, $g(s){\sim}s$.
Let $f:\ A{\cup}B{\longrightarrow}A^{\prime}{\cup}B^{\prime}$ be the function such that for all $s{\in}A$, $f(s){=}g(s)$ and for all $s{\in}B$, $f(s){=}s$.
Notice that $f$ is a bounded morphism from ${\mathcal F}_{A,B}^{\rho}$ to ${\mathcal F}_{A^{\prime},B^{\prime}}^{\rho^{\prime}}$.
Notice also that $f$ is surjective.
\\
\\
$\mathbf{(8)}$~Let $q^{\prime}{\in}\lsem1,q\rsem$ and $((s_{1},\ldots,s_{q^{\prime}{-}1}),(s^{\prime}_{1},\ldots,s^{\prime}_{q^{\prime}{-}1}))$ be a winning position of the Ehrenfeucht-Fra\"\i ss\'e game ${\mathcal G}_{q}({\mathcal F}_{A,B}^{\rho},{\mathcal F}_{A^{\prime},B^{\prime}}^{\rho^{\prime}})$.
Suppose the first player picked a member $s$ of ${\mathcal F}_{A,B}^{\rho}$ in its $q^{\prime}$th move.
Then
\begin{itemize}
\item if $s$ is in $A$ then
\begin{itemize}
\item if $s{=}s_{q^{\prime\prime}}$ for some $q^{\prime\prime}{\in}\lsem1,q^{\prime}{-}1\rsem$ then the second player picks $s^{\prime}_{q^{\prime\prime}}$ in $A^{\prime}$ in its $q^{\prime}$th move,
\item otherwise, the second player picks an element in $\lbrack s\rbrack_{\sim}{\cap}(A^{\prime}{\setminus}\{s^{\prime}_{1},\ldots,
$\linebreak$
s^{\prime}_{q^{\prime}{-}1}\})$ in its $q^{\prime}$th move,
\end{itemize}
\item if $s$ is in $B$ then the second player picks $s$ in $B^{\prime}$ in its $q^{\prime}$th move.
\end{itemize}
Suppose the first player picked a member $s^{\prime}$ of ${\mathcal F}_{A^{\prime},B^{\prime}}^{\rho^{\prime}}$ in its $q^{\prime}$th move.
Then
\begin{itemize}
\item if $s^{\prime}$ is in $A^{\prime}$ then
\begin{itemize}
\item if $s^{\prime}{=}s^{\prime}_{q^{\prime\prime}}$ for some $q^{\prime\prime}{\in}\lsem1,q^{\prime}{-}1\rsem$ then the second player picks $s_{q^{\prime\prime}}$ in $A$ in its $q^{\prime}$th move,
\item otherwise, the second player picks an element in $\lbrack s^{\prime}\rbrack_{\sim}{\cap}(A{\setminus}\{s_{1},\ldots,
$\linebreak$
s_{q^{\prime}{-}1}\})$ in its $q^{\prime}$th move,
\end{itemize}
\item if $s^{\prime}$ is in $B^{\prime}$ then the second player picks $s^{\prime}$ in $B$ in its $q^{\prime}$th move.
\end{itemize}
Obviously, $((s_{1},\ldots,s_{q^{\prime}{-}1},s),(s^{\prime}_{1},\ldots,s^{\prime}_{q^{\prime}{-}1},s^{\prime}))$ is a winning position of the Eh\-renfeucht-Fra\"\i ss\'e game ${\mathcal G}_{q}({\mathcal F}_{A,B}^{\rho},{\mathcal F}_{A^{\prime},B^{\prime}}^{\rho^{\prime}})$.
\medskip
\end{proof}
%
%
Let $\{{\mathcal F}_{A_{i},B_{i}}^{\rho_{i}}:\ i{\in}I\}$ be a disjoint indexed family of galaxies.
$\{{\mathcal F}_{A_{i}^{\prime},B_{i}^{\prime}}^{\rho_{i}^{\prime}}:\ i{\in}I\}$ is a {\it $q$-$\alpha$-reduction of $\{{\mathcal F}_{A_{i},B_{i}}^{\rho_{i}}:\ i{\in}I\}$}\/ if for all $i{\in}I$, ${\mathcal F}_{A_{i}^{\prime},B_{i}^{\prime}}^{\rho_{i}^{\prime}}$ is a $q$-$\alpha$-reduction of ${\mathcal F}_{A_{i},B_{i}}^{\rho_{i}}$.
\begin{proposition}\label{lemma:good:properties:of:alpha:reductions:indexed}
%
%
Let $I^{ns}$ be the set of all $i{\in}I$ such that ${\mathcal F}_{A_{i},B_{i}}^{\rho_{i}}$ is non-simple.
If $\{{\mathcal F}_{A_{i}^{\prime},B_{i}^{\prime}}^{\rho_{i}^{\prime}}:\ i{\in}I\}$ is a $q$-$\alpha$-reduction of $\{{\mathcal F}_{A_{i},B_{i}}^{\rho_{i}}:\ i{\in}I\}$, $(W,R)$ is the union of $\{{\mathcal F}_{A_{i},B_{i}}^{\rho_{i}}:\ i{\in}I\}$ and $(W^{\prime},R^{\prime})$ is the union of $\{{\mathcal F}_{A_{i}^{\prime},B_{i}^{\prime}}^{\rho_{i}^{\prime}}:\ i{\in}I\}$ then
\begin{enumerate}
\item $\{{\mathcal F}_{A_{i}^{\prime},B_{i}^{\prime}}^{\rho_{i}^{\prime}}:\ i{\in}I\}$ is dust-bounded,
\item if $\{{\mathcal F}_{A_{i},B_{i}}^{\rho_{i}}:\ i{\in}I^{ns}\}$ is kernel-bounded then $\{{\mathcal F}_{A_{i}^{\prime},B_{i}^{\prime}}^{\rho_{i}^{\prime}}:\ i{\in}I^{ns}\}$ is root-bounded and kernel-bounded,
\item $(W^{\prime},R^{\prime})$ is a bounded morphic image of $(W,R)$,
\item the second player has a winning strategy in the Ehrenfeucht-Fra\"\i ss\'e game ${\mathcal G}_{q}((W,R),(W^{\prime},R^{\prime}))$.
\end{enumerate}
\end{proposition}
\begin{proof}
By Lemma~\ref{lemma:the:alpha:reduction:is:always:dust:finite}.
\medskip
\end{proof}
\paragraph{$\gamma$-reductions}
%
%
The main effect of a $\gamma$-reduction on a simple galaxy is to reduce both the cardinality of its upper part and the cardinality of its lower part.
\\
\\
Let ${\mathcal F}_{A,B}^{\rho}$ be a simple galaxy satisfying the conditions $(4)$ and $(5)$ of Lemma~\ref{lemma:the:alpha:reduction:is:always:dust:finite}.
\\
\\
%
%
Let $\approx$ be the equivalence relation on $B$ defined by $u{\approx}v$ if and only if ${\parallel}\rho^{{-}1}(\{u\}){\parallel}
$\linebreak$
{=}{\parallel}\rho^{{-}1}(\{v\}){\parallel}$ and ${\parallel}\rho^{{-}1}(B{\setminus}\{u\}){\parallel}{=}{\parallel}\rho^{{-}1}(B{\setminus}\{v\}){\parallel}$.
For all $u{\in}B$, let $\lbrack u\rbrack_{\approx}$ be the equivalence class of $u$ modulo $\approx$.
\\
\\
For all $u{\in}B$, if ${\parallel}\lbrack u\rbrack_{\approx}{\parallel}{\leq}q$ then let $D_{u}$ be $\lbrack u\rbrack_{\approx}$ else let $D_{u}$ be an arbitrary subset of $\lbrack u\rbrack_{\approx}$ of cardinality $q$.
Let $B^{\prime}$ be $\bigcup\{D_{u}:\ u{\in}B\}$.
Obviously, from each equivalence class modulo $\approx$, $B^{\prime}$ contains at most $q$ elements.
Therefore, $B^{\prime}$ contains at most $q{\times}(q{+}1)^{2}$ elements.
\\
\\
%
%
%
%
Let $A^{\prime}$ be $\rho^{{-}1}(\emptyset){\cup}\rho^{{-}1}(B){\cup}\bigcup\{\rho^{{-}1}(\{u\}):\ u{\in}B^{\prime}\}{\cup}\bigcup\{\rho^{{-}1}(B{\setminus}\{u\}):\ u{\in}B^{\prime}\}$.
Since $B^{\prime}$ contains at most $q{\times}(q{+}1)^{2}$ elements, then $A^{\prime}$ contains at most $2{\times}q{\times}(q
$\linebreak$
{\times}(q{+}1)^{2}+1)$ elements.
%
%
%
%
%
%
\\
\\
Let ${\mathcal F}_{A^{\prime},B^{\prime}}^{\rho^{\prime}}$ be the galaxy such that for all $u^{\prime}{\in}A^{\prime}$, $\rho^{\prime}(u^{\prime}){=}\rho(u^{\prime}){\cap}B^{\prime}$.
${\mathcal F}_{A^{\prime},B^{\prime}}^{\rho^{\prime}}$ is a {\it $q$-$\gamma$-reduction of ${\mathcal F}_{A,B}^{\rho}$.}
\begin{lemma}\label{lemma:gamma:reduction:prebounded:dust:quasi:root:bounded:OK}
\begin{enumerate}
\item $A^{\prime}{\subseteq}A$,
\item $B^{\prime}{\subseteq}B$ and if ${\parallel}B^{\prime}{\parallel}{\leq}2$ then $B^{\prime}{=}B$,
\item for all $s^{\prime}{\in}A^{\prime}$, $\rho^{\prime}(s^{\prime}){\subseteq}\rho(s^{\prime})$,
\item ${\parallel}{\rho^{\prime}}^{{-}1}(\emptyset){\parallel}{\leq}q$ and ${\parallel}{\rho^{\prime}}^{{-}1}(B^{\prime}){\parallel}{\leq}q$,
\item for all $t^{\prime}{\in}B^{\prime}$, ${\parallel}{\rho^{\prime}}^{{-}1}(\{t^{\prime}\}){\parallel}{\leq}q$ and for all $t^{\prime}{\in}B^{\prime}$, ${\parallel}{\rho^{\prime}}^{{-}1}(B^{\prime}{\setminus}\{t^{\prime}\}){\parallel}{\leq}q$,
\item ${\mathcal F}_{A^{\prime},B^{\prime}}^{\rho^{\prime}}$ is simple, ${\parallel}B^{\prime}{\parallel}{\leq}q{\times}(q{+}1)^{2}$ and ${\parallel}A^{\prime}{\parallel}{\leq}2{\times}q{\times}(q{\times}(q{+}1)^{2}+1)$,
\item for all $s^{\prime}{\in}A^{\prime}{\cup}B^{\prime}$, there exists $s{\in}A{\cup}B$ such that the least generated subframe of ${\mathcal F}_{A^{\prime},B^{\prime}}^{\rho^{\prime}}$ containing $s^{\prime}$ is a bounded morphic image of the least generated subframe of ${\mathcal F}_{A,B}^{\rho}$ containing $s$,
%
%
\item the second player has a winning strategy in the Ehrenfeucht-Fra\"\i ss\'e game ${\mathcal G}_{q}({\mathcal F}_{A,B}^{\rho},{\mathcal F}_{A^{\prime},B^{\prime}}^{\rho^{\prime}})$.
\end{enumerate}
\end{lemma}
\begin{proof}
$\mathbf{(1)}$--$\mathbf{(6)}$~By construction of ${\mathcal F}_{A^{\prime},B^{\prime}}^{\rho^{\prime}}$.
\\
\\
$\mathbf{(7)}$~Let $s^{\prime}{\in}A^{\prime}{\cup}B^{\prime}$.
Hence, either $\mathbf{(i)}$~$s^{\prime}{\in}\rho^{-1}(\emptyset)$, or $\mathbf{(ii)}$~$s^{\prime}{\in}{\cup}\rho^{{-}1}(B)$, or $\mathbf{(iii)}$~there exists $u{\in}B^{\prime}$ such that $s^{\prime}{\in}\rho^{{-}1}(\{u\})$, or $\mathbf{(iv)}$~there exists $u{\in}B^{\prime}$ such that $s^{\prime}{\in}\rho^{{-}1}(B
$\linebreak$
{\setminus}\{u\})$, or $\mathbf{(v)}$~$s^{\prime}{\in}B^{\prime}$.
We consider the following $5$~cases.
\\
\\
{\bf Case $\mathbf{(i)}$:}
Thus, $s^{\prime}{\in}A^{\prime}$, $s^{\prime}{\in}A$, $\rho^{\prime}(s^{\prime}){=}\emptyset$ and $\rho(s^{\prime}){=}\emptyset$.
Consequently, the least generated subframe of ${\mathcal F}_{A^{\prime},B^{\prime}}^{\rho^{\prime}}$ containing $s^{\prime}$ is isomorphic to the isolated point and the least generated subframe of ${\mathcal F}_{A,B}^{\rho}$ containing $s^{\prime}$ is isomorphic to the isolated point.
Hence, the former generated subframe is a bounded morphic image of the latter generated subframe.
\\
\\
{\bf Case $\mathbf{(ii)}$:}
Thus, $s^{\prime}{\in}A^{\prime}$, $s^{\prime}{\in}A$, $\rho^{\prime}(s^{\prime}){=}B^{\prime}$ and $\rho(s^{\prime}){=}B$.
Consequently, the least generated subframe of ${\mathcal F}_{A^{\prime},B^{\prime}}^{\rho^{\prime}}$ containing $s^{\prime}$ is of the form $(\{s^{\prime}\}{\cup}B^{\prime},(\{s^{\prime}\}{\times}B^{\prime}){\cup}
$\linebreak$
(B^{\prime}{\times}B^{\prime}))$ and the least generated subframe of ${\mathcal F}_{A,B}^{\rho}$ containing $s^{\prime}$ is of the form $(\{s^{\prime}\}{\cup}B,(\{s^{\prime}\}{\times}B){\cup}(B{\times}B))$.
Since $B^{\prime}{\subseteq}B$, then ${\parallel}B^{\prime}{\parallel}{\leq}{\parallel}B{\parallel}$.
Since $B^{\prime}{\not=}\emptyset$ and $B{\not=}\emptyset$, then by Lemma~\ref{lemma:inclusion:Bprime:B:for:generated:subframes:is:bmi}, the former generated subframe is a bounded morphic image of the latter generated subframe.
\\
\\
{\bf Case $\mathbf{(iii)}$:}
Hence, $s^{\prime}{\in}A^{\prime}$, $u{\in}B^{\prime}$, $s^{\prime}{\in}A$, $u{\in}B$, $\rho^{\prime}(s^{\prime}){=}\{u\}$ and $\rho(s^{\prime}){=}\{u\}$.
Thus, the least generated subframe of ${\mathcal F}_{A^{\prime},B^{\prime}}^{\rho^{\prime}}$ containing $s^{\prime}$ is of the form $(\{s^{\prime}\}{\cup}B^{\prime},\{(s^{\prime},
$\linebreak$
u)\}{\cup}(B^{\prime}{\times}B^{\prime}))$ and the least generated subframe of ${\mathcal F}_{A,B}^{\rho}$ containing $s^{\prime}$ is of the form $(\{s^{\prime}\}{\cup}B,\{(s^{\prime},u)\}{\cup}(B{\times}B))$.
Since $B^{\prime}{\subseteq}B$, then ${\parallel}B^{\prime}{\parallel}{\leq}{\parallel}B{\parallel}$.
Since $B^{\prime}{\not=}\emptyset$ and $B{\not=}\emptyset$, then by Lemma~\ref{lemma:inclusion:Bprime:B:for:generated:subframes:is:bmi:bis}, the former generated subframe is a bounded morphic image of the latter generated subframe.
\\
\\
{\bf Case $\mathbf{(iv)}$:}
Consequently, $s^{\prime}{\in}A^{\prime}$, $u{\in}B^{\prime}$, $s^{\prime}{\in}A$, $u{\in}B$, $\rho^{\prime}(s^{\prime}){=}B^{\prime}{\setminus}\{u\}$ and $\rho(s^{\prime}){=}B
$\linebreak$
{\setminus}\{u\}$.
Hence, the least generated subframe of ${\mathcal F}_{A^{\prime},B^{\prime}}^{\rho^{\prime}}$ containing $s^{\prime}$ is of the form $(\{s^{\prime}\}{\cup}B^{\prime},(\{s^{\prime}\}{\times}(B^{\prime}{\setminus}\{u\})){\cup}(B^{\prime}{\times}B^{\prime}))$ and the least generated subframe of ${\mathcal F}_{A,B}^{\rho}$ containing $s^{\prime}$ is of the form $(\{s^{\prime}\}{\cup}B,(\{s^{\prime}\}{\times}(B{\setminus}\{u\})){\cup}(B{\times}B))$.
Since $B^{\prime}{\subseteq}B$, then ${\parallel}B^{\prime}{\parallel}{\leq}{\parallel}B{\parallel}$.
Since $B^{\prime}{\not=}\emptyset$ and $B{\not=}\emptyset$, then by Lemma~\ref{lemma:inclusion:Bprime:B:for:generated:subframes:is:bmi:ter}, the former generated subframe is a bounded morphic image of the latter generated subframe.
\\
\\
{\bf Case $\mathbf{(v)}$:}
Thus, $s^{\prime}{\in}B^{\prime}$ and $s^{\prime}{\in}B$.
Consequently, the least generated subframe of ${\mathcal F}_{A^{\prime},B^{\prime}}^{\rho^{\prime}}$ containing $s^{\prime}$ is of the form $(B^{\prime},B^{\prime}{\times}B^{\prime})$ and the least generated subframe of ${\mathcal F}_{A,B}^{\rho}$ containing $s^{\prime}$ is of the form $(B,B{\times}B)$.
Since $B^{\prime}{\subseteq}B$, then ${\parallel}B^{\prime}{\parallel}{\leq}{\parallel}B{\parallel}$.
Since $B^{\prime}{\not=}\emptyset$ and $B{\not=}\emptyset$, then by Lemma~\ref{lemma:inclusion:Bprime:B:for:generated:subframes:is:bmi:quater}, the former generated subframe is a bounded morphic image of the latter generated subframe.
\\
\\
$\mathbf{(8)}$~Let us assume that some well-order on $A{\cup}B$ has been fixed.
Let $q^{\prime}{\in}\lsem1,q\rsem$ and $((s_{1},\ldots,s_{q^{\prime}{-}1}),(s^{\prime}_{1},\ldots,s^{\prime}_{q^{\prime}{-}1}))$ be a winning position of the Ehrenfeucht-Fra\"\i ss\'e game ${\mathcal G}_{q}({\mathcal F}_{A,B}^{\rho},{\mathcal F}_{A^{\prime},B^{\prime}}^{\rho^{\prime}})$.
Suppose the first player picked a member $s$ of ${\mathcal F}_{A,B}^{\rho}$ in its $q^{\prime}$th move.
Then
\begin{itemize}
\item if $s{=}s_{q^{\prime\prime}}$ for some $q^{\prime\prime}{\in}\lsem1,q^{\prime}{-}1\rsem$ and $s{\in}B$ then the second player picks $s^{\prime}_{q^{\prime\prime}}$ in $B^{\prime}$ in its $q^{\prime}$th move,
\item if $s{\not=}s_{q^{\prime\prime}}$ for every $q^{\prime\prime}{\in}\lsem1,q^{\prime}{-}1\rsem$, $s{\in}B$ and ${\parallel}\lbrack s\rbrack_{\approx}{\parallel}{\leq}q$ then the second player picks $s$ in $B^{\prime}$ in its $q^{\prime}$th move,
\item if $s{\not=}s_{q^{\prime\prime}}$ for every $q^{\prime\prime}{\in}\lsem1,q^{\prime}{-}1\rsem$, $s{\in}B$ and ${\parallel}\lbrack s\rbrack_{\approx}{\parallel}{>}q$ then the second player picks the least $s^{\prime}$ in $B^{\prime}$ such that
\begin{itemize}
\item $s{\approx}s^{\prime}$,
\item $s^{\prime}{\not=}s^{\prime}_{q^{\prime\prime}}$ for every $q^{\prime\prime}{\in}\lsem1,q^{\prime}{-}1\rsem$,
\item for all $q^{\prime\prime}{\in}\lsem1,q^{\prime}{-}1\rsem$, if $s_{q^{\prime\prime}}{\in}A$, $\rho(s_{q^{\prime\prime}}){\not=}\emptyset$ and $\rho(s_{q^{\prime\prime}}){\not=}B$ then
\begin{itemize}
\item $\rho(s_{q^{\prime\prime}}){=}\{s\}$ if and only if $\rho^{\prime}(s^{\prime}_{q^{\prime\prime}}){=}\{s^{\prime}\}$,
\item $\rho(s_{q^{\prime\prime}}){=}B{\setminus}\{s\}$ if and only if $\rho^{\prime}(s^{\prime}_{q^{\prime\prime}}){=}B^{\prime}{\setminus}\{s^{\prime}\}$,
\end{itemize}
\end{itemize}
in its $q^{\prime}$th move,
\item if $s{=}s_{q^{\prime\prime}}$ for some $q^{\prime\prime}{\in}\lsem1,q^{\prime}{-}1\rsem$ and $s{\in}A$ then the second player picks $s^{\prime}_{q^{\prime\prime}}$ in $A^{\prime}$ in its $q^{\prime}$th move,
\item if $s{\not=}s_{q^{\prime\prime}}$ for every $q^{\prime\prime}{\in}\lsem1,q^{\prime}{-}1\rsem$, $s{\in}A$ and either $\rho(s){=}\emptyset$, or $\rho(s){=}B$ then the second player picks $s$ in $A^{\prime}$ in its $q^{\prime}$th move,
\item if $s{\not=}s_{q^{\prime\prime}}$ for every $q^{\prime\prime}{\in}\lsem1,q^{\prime}{-}1\rsem$, $s{\in}A$ and either $\rho(s){=}\{t\}$, or $\rho(s){=}B{\setminus}\{t\}$ for some $t{\in}B$ such that ${\parallel}\lbrack t\rbrack_{\approx}{\parallel}{\leq}q$ then the second player picks $s$ in $A^{\prime}$ in its $q^{\prime}$th move,
\item if $s{\not=}s_{q^{\prime\prime}}$ for every $q^{\prime\prime}{\in}\lsem1,q^{\prime}{-}1\rsem$, $s{\in}A$ and $\rho(s){=}\{t\}$ for some $t{\in}B$ such that ${\parallel}\lbrack t\rbrack_{\approx}{\parallel}{>}q$ and $t{=}s_{q^{\prime\prime\prime}}$ for some $q^{\prime\prime\prime}{\in}\lsem1,q^{\prime}{-}1\rsem$ then the second player picks the least element in ${\rho^{\prime}}^{{-}1}(\{s^{\prime}_{q^{\prime\prime\prime}}\}){\setminus}\{s^{\prime}_{1},\ldots,s^{\prime}_{q^{\prime}{-}1}\}$,
\item if $s{\not=}s_{q^{\prime\prime}}$ for every $q^{\prime\prime}{\in}\lsem1,q^{\prime}{-}1\rsem$, $s{\in}A$ and $\rho(s){=}B{\setminus}\{t\}$ for some $t{\in}B$ such that ${\parallel}\lbrack t\rbrack_{\approx}{\parallel}{>}q$ and $t{=}s_{q^{\prime\prime\prime}}$ for some $q^{\prime\prime\prime}{\in}\lsem1,q^{\prime}{-}1\rsem$ then the second player picks the least element in ${\rho^{\prime}}^{{-}1}(B^{\prime}{\setminus}\{s^{\prime}_{q^{\prime\prime\prime}}\}){\setminus}\{s^{\prime}_{1},\ldots,s^{\prime}_{q^{\prime}{-}1}\}$,
\item if $s{\not=}s_{q^{\prime\prime}}$ for every $q^{\prime\prime}{\in}\lsem1,q^{\prime}{-}1\rsem$, $s{\in}A$ and $\rho(s){=}\{t\}$ for some $t{\in}B$ such that ${\parallel}\lbrack t\rbrack_{\approx}{\parallel}{>}q$ and $t{\not=}s_{q^{\prime\prime\prime}}$ for every $q^{\prime\prime\prime}{\in}\lsem1,q^{\prime}{-}1\rsem$ then $t^{\prime}$ being the least element in $B^{\prime}$ such that
\begin{itemize}
\item $t{\approx}t^{\prime}$,
\item $t^{\prime}{\not=}s^{\prime}_{q^{\prime\prime}}$ for every $q^{\prime\prime}{\in}\lsem1,q^{\prime}{-}1\rsem$,
\item for all $q^{\prime\prime}{\in}\lsem1,q^{\prime}{-}1\rsem$, if $s_{q^{\prime\prime}}{\in}A$, $\rho(s_{q^{\prime\prime}}){\not=}\emptyset$ and $\rho(s_{q^{\prime\prime}}){\not=}B$ then
\begin{itemize}
\item $\rho(s_{q^{\prime\prime}}){=}\{t\}$ if and only if $\rho^{\prime}(s^{\prime}_{q^{\prime\prime}}){=}\{t^{\prime}\}$,
\item $\rho(s_{q^{\prime\prime}}){=}B{\setminus}\{t\}$ if and only if $\rho^{\prime}(s^{\prime}_{q^{\prime\prime}}){=}B^{\prime}{\setminus}\{t^{\prime}\}$,
\end{itemize}
\end{itemize}
the second player picks the least element in ${\rho^{\prime}}^{{-}1}(\{t^{\prime}\}){\setminus}\{s^{\prime}_{1},\ldots,s^{\prime}_{q^{\prime}{-}1}\}$,
\item if $s{\not=}s_{q^{\prime\prime}}$ for every $q^{\prime\prime}{\in}\lsem1,q^{\prime}{-}1\rsem$, $s{\in}A$ and $\rho(s){=}B{\setminus}\{t\}$ for some $t{\in}B$ such that ${\parallel}\lbrack t\rbrack_{\approx}{\parallel}{>}q$ and $t{\not=}s_{q^{\prime\prime\prime}}$ for some $q^{\prime\prime\prime}{\in}\lsem1,q^{\prime}{-}1\rsem$ then $t^{\prime}$ being the least element in $B^{\prime}$ such that
\begin{itemize}
\item $t{\approx}t^{\prime}$,
\item $t^{\prime}{\not=}s^{\prime}_{q^{\prime\prime}}$ for every $q^{\prime\prime}{\in}\lsem1,q^{\prime}{-}1\rsem$,
\item for all $q^{\prime\prime}{\in}\lsem1,q^{\prime}{-}1\rsem$, if $s_{q^{\prime\prime}}{\in}A$, $\rho(s_{q^{\prime\prime}}){\not=}\emptyset$ and $\rho(s_{q^{\prime\prime}}){\not=}B$ then
\begin{itemize}
\item $\rho(s_{q^{\prime\prime}}){=}\{t\}$ if and only if $\rho^{\prime}(s^{\prime}_{q^{\prime\prime}}){=}\{t^{\prime}\}$,
\item $\rho(s_{q^{\prime\prime}}){=}B{\setminus}\{t\}$ if and only if $\rho^{\prime}(s^{\prime}_{q^{\prime\prime}}){=}B^{\prime}{\setminus}\{t^{\prime}\}$,
\end{itemize}
\end{itemize}
the second player picks the least element in ${\rho^{\prime}}^{{-}1}(B^{\prime}{\setminus}\{t^{\prime}\}){\setminus}\{s^{\prime}_{1},\ldots,s^{\prime}_{q^{\prime}{-}1}\}$.
\end{itemize}
Suppose the first player picked a member $s^{\prime}$ of ${\mathcal F}_{A^{\prime},B^{\prime}}^{\rho^{\prime}}$ in its $q^{\prime}$th move.
Then
\begin{itemize}
\item if $s^{\prime}{=}s^{\prime}_{q^{\prime\prime}}$ for some $q^{\prime\prime}{\in}\lsem1,q^{\prime}{-}1\rsem$ and $s^{\prime}{\in}B^{\prime}$ then the second player picks $s_{q^{\prime\prime}}$ in $B$ in its $q^{\prime}$th move,
\item if $s^{\prime}{\not=}s^{\prime}_{q^{\prime\prime}}$ for every $q^{\prime\prime}{\in}\lsem1,q^{\prime}{-}1\rsem$, $s^{\prime}{\in}B^{\prime}$ and ${\parallel}\lbrack s^{\prime}\rbrack_{\approx}{\parallel}{\leq}q$ then the second player picks $s^{\prime}$ in $B$ in its $q^{\prime}$th move,
\item if $s^{\prime}{\not=}s^{\prime}_{q^{\prime\prime}}$ for every $q^{\prime\prime}{\in}\lsem1,q^{\prime}{-}1\rsem$, $s^{\prime}{\in}B^{\prime}$ and ${\parallel}\lbrack s^{\prime}\rbrack_{\approx}{\parallel}{>}q$ then the second player picks the least $s$ in $B$ such that
\begin{itemize}
\item $s^{\prime}{\approx}s$,
\item $s{\not=}s_{q^{\prime\prime}}$ for every $q^{\prime\prime}{\in}\lsem1,q^{\prime}{-}1\rsem$,
\item for all $q^{\prime\prime}{\in}\lsem1,q^{\prime}{-}1\rsem$, if $s^{\prime}_{q^{\prime\prime}}{\in}A^{\prime}$, $\rho^{\prime}(s^{\prime}_{q^{\prime\prime}}){\not=}\emptyset$ and $\rho^{\prime}(s^{\prime}_{q^{\prime\prime}}){\not=}B^{\prime}$ then
\begin{itemize}
\item $\rho^{\prime}(s^{\prime}_{q^{\prime\prime}}){=}\{s^{\prime}\}$ if and only if $\rho(s_{q^{\prime\prime}}){=}\{s\}$,
\item $\rho^{\prime}(s^{\prime}_{q^{\prime\prime}}){=}B^{\prime}{\setminus}\{s^{\prime}\}$ if and only if $\rho(s_{q^{\prime\prime}}){=}B{\setminus}\{s\}$,
\end{itemize}
\end{itemize}
in its $q^{\prime}$th move,
\item if $s^{\prime}{=}s^{\prime}_{q^{\prime\prime}}$ for some $q^{\prime\prime}{\in}\lsem1,q^{\prime}{-}1\rsem$ and $s^{\prime}{\in}A^{\prime}$ then the second player picks $s_{q^{\prime\prime}}$ in $A$ in its $q^{\prime}$th move,
\item if $s^{\prime}{\not=}s^{\prime}_{q^{\prime\prime}}$ for every $q^{\prime\prime}{\in}\lsem1,q^{\prime}{-}1\rsem$, $s^{\prime}{\in}A^{\prime}$ and either $\rho^{\prime}(s^{\prime}){=}\emptyset$, or $\rho^{\prime}(s^{\prime}){=}B^{\prime}$ then the second player picks $s^{\prime}$ in $A$ in its $q^{\prime}$th move,
\item if $s^{\prime}{\not=}s^{\prime}_{q^{\prime\prime}}$ for every $q^{\prime\prime}{\in}\lsem1,q^{\prime}{-}1\rsem$, $s^{\prime}{\in}A^{\prime}$ and either $\rho^{\prime}(s^{\prime}){=}\{t^{\prime}\}$, or $\rho^{\prime}(s^{\prime}){=}B^{\prime}
$\linebreak$
{\setminus}\{t^{\prime}\}$ for some $t^{\prime}{\in}B^{\prime}$ such that ${\parallel}\lbrack t^{\prime}\rbrack_{\approx}{\parallel}{\leq}q$ then the second player picks $s^{\prime}$ in $A$ in its $q^{\prime}$th move,
\item if $s^{\prime}{\not=}s^{\prime}_{q^{\prime\prime}}$ for every $q^{\prime\prime}{\in}\lsem1,q^{\prime}{-}1\rsem$, $s^{\prime}{\in}A^{\prime}$ and $\rho(s^{\prime}){=}\{t^{\prime}\}$ for some $t^{\prime}{\in}B^{\prime}$ such that ${\parallel}\lbrack t^{\prime}\rbrack_{\approx}{\parallel}{>}q$ and $t^{\prime}{=}s^{\prime}_{q^{\prime\prime\prime}}$ for some $q^{\prime\prime\prime}{\in}\lsem1,q^{\prime}{-}1\rsem$ then the second player picks the least element in ${\rho}^{{-}1}(\{s_{q^{\prime\prime\prime}}\}){\setminus}\{s_{1},\ldots,s_{q^{\prime}{-}1}\}$,
\item if $s^{\prime}{\not=}s^{\prime}_{q^{\prime\prime}}$ for every $q^{\prime\prime}{\in}\lsem1,q^{\prime}{-}1\rsem$, $s^{\prime}{\in}A^{\prime}$ and $\rho^{\prime}(s^{\prime}){=}B^{\prime}{\setminus}\{t^{\prime}\}$ for some $t^{\prime}{\in}B^{\prime}$ such that ${\parallel}\lbrack t^{\prime}\rbrack_{\approx}{\parallel}{>}q$ and $t^{\prime}{=}s^{\prime}_{q^{\prime\prime\prime}}$ for some $q^{\prime\prime\prime}{\in}\lsem1,q^{\prime}{-}1\rsem$ then the second player picks the least element in ${\rho}^{{-}1}(B{\setminus}\{s_{q^{\prime\prime\prime}}\}){\setminus}\{s_{1},\ldots,s_{q^{\prime}{-}1}\}$,
\item if $s^{\prime}{\not=}s^{\prime}_{q^{\prime\prime}}$ for every $q^{\prime\prime}{\in}\lsem1,q^{\prime}{-}1\rsem$, $s^{\prime}{\in}A^{\prime}$ and $\rho^{\prime}(s^{\prime}){=}\{t^{\prime}\}$ for some $t^{\prime}{\in}B^{\prime}$ such that ${\parallel}\lbrack t^{\prime}\rbrack_{\approx}{\parallel}{>}q$ and $t^{\prime}{\not=}s^{\prime}_{q^{\prime\prime\prime}}$ for every $q^{\prime\prime\prime}{\in}\lsem1,q^{\prime}{-}1\rsem$ then $t$ being the least element in $B$ such that
\begin{itemize}
\item $t^{\prime}{\approx}t$,
\item $t{\not=}s_{q^{\prime\prime}}$ for every $q^{\prime\prime}{\in}\lsem1,q^{\prime}{-}1\rsem$,
\item for all $q^{\prime\prime}{\in}\lsem1,q^{\prime}{-}1\rsem$, if $s^{\prime}_{q^{\prime\prime}}{\in}A^{\prime}$, $\rho^{\prime}(s^{\prime}_{q^{\prime\prime}}){\not=}\emptyset$ and $\rho^{\prime}(s^{\prime}_{q^{\prime\prime}}){\not=}B^{\prime}$ then
\begin{itemize}
\item $\rho^{\prime}(s^{\prime}_{q^{\prime\prime}}){=}\{t^{\prime}\}$ if and only if $\rho(s_{q^{\prime\prime}}){=}\{t\}$,
\item $\rho^{\prime}(s^{\prime}_{q^{\prime\prime}}){=}B^{\prime}{\setminus}\{t^{\prime}\}$ if and only if $\rho(s_{q^{\prime\prime}}){=}B{\setminus}\{t\}$,
\end{itemize}
\end{itemize}
the second player picks the least element in ${\rho}^{{-}1}(\{t\}){\setminus}\{s_{1},\ldots,s_{q^{\prime}{-}1}\}$,
\item if $s^{\prime}{\not=}s^{\prime}_{q^{\prime\prime}}$ for every $q^{\prime\prime}{\in}\lsem1,q^{\prime}{-}1\rsem$, $s^{\prime}{\in}A^{\prime}$ and $\rho^{\prime}(s^{\prime}){=}B^{\prime}{\setminus}\{t^{\prime}\}$ for some $t^{\prime}{\in}B^{\prime}$ such that ${\parallel}\lbrack t^{\prime}\rbrack_{\approx}{\parallel}{>}q$ and $t^{\prime}{\not=}s^{\prime}_{q^{\prime\prime\prime}}$ for some $q^{\prime\prime\prime}{\in}\lsem1,q^{\prime}{-}1\rsem$ then $t$ being the least element in $B$ such that
\begin{itemize}
\item $t^{\prime}{\approx}t$,
\item $t{\not=}s_{q^{\prime\prime}}$ for every $q^{\prime\prime}{\in}\lsem1,q^{\prime}{-}1\rsem$,
\item for all $q^{\prime\prime}{\in}\lsem1,q^{\prime}{-}1\rsem$, if $s^{\prime}_{q^{\prime\prime}}{\in}A^{\prime}$, $\rho^{\prime}(s^{\prime}_{q^{\prime\prime}}){\not=}\emptyset$ and $\rho^{\prime}(s^{\prime}_{q^{\prime\prime}}){\not=}B^{\prime}$ then
\begin{itemize}
\item $\rho^{\prime}(s^{\prime}_{q^{\prime\prime}}){=}\{t^{\prime}\}$ if and only if $\rho(s_{q^{\prime\prime}}){=}\{t\}$,
\item $\rho^{\prime}(s^{\prime}_{q^{\prime\prime}}){=}B^{\prime}{\setminus}\{t^{\prime}\}$ if and only if $\rho(s_{q^{\prime\prime}}){=}B{\setminus}\{t\}$,
\end{itemize}
\end{itemize}
the second player picks the least element in ${\rho}^{{-}1}(B{\setminus}\{t\}){\setminus}\{s_{1},\ldots,s_{q^{\prime}{-}1}\}$.
\end{itemize}
Obviously, $((s_{1},\ldots,s_{q^{\prime}{-}1},s),(s^{\prime}_{1},\ldots,s^{\prime}_{q^{\prime}{-}1},s^{\prime}))$ is a winning position of the Eh\-renfeucht-Fra\"\i ss\'e game ${\mathcal G}_{q}({\mathcal F}_{A,B}^{\rho},{\mathcal F}_{A^{\prime},B^{\prime}}^{\rho^{\prime}})$.
\medskip
\end{proof}
%
%
Let $\{{\mathcal F}_{A_{i},B_{i}}^{\rho_{i}}:\ i{\in}I\}$ be a disjoint family of galaxies such that for all $i{\in}A$, either ${\mathcal F}_{A_{i},B_{i}}^{\rho_{i}}$ is a simple galaxy satisfying the conditions $(4)$ and $(5)$ of Lemma~\ref{lemma:the:alpha:reduction:is:always:dust:finite}, or ${\mathcal F}_{A_{i},B_{i}}^{\rho_{i}}$ is a non-simple galaxy.
$\{{\mathcal F}_{A_{i}^{\prime},B_{i}^{\prime}}^{\rho_{i}^{\prime}}:\ i{\in}I\}$ is a {\it $q$-$\gamma$-reduction of $\{{\mathcal F}_{A_{i},B_{i}}^{\rho_{i}}:\ i{\in}I\}$}\/ if for all $i{\in}I$, either ${\mathcal F}_{A_{i},B_{i}}^{\rho_{i}}$ is a simple galaxy satisfying the conditions $(4)$ and $(5)$ of Lemma~\ref{lemma:the:alpha:reduction:is:always:dust:finite} and ${\mathcal F}_{A_{i}^{\prime},B_{i}^{\prime}}^{\rho_{i}^{\prime}}$ is a $q$-$\gamma$-reduction of ${\mathcal F}_{A_{i},B_{i}}^{\rho_{i}}$, or ${\mathcal F}_{A_{i},B_{i}}^{\rho_{i}}$ is a non-simple galaxy and ${\mathcal F}_{A_{i}^{\prime},B_{i}^{\prime}}^{\rho_{i}^{\prime}}$ is equal to ${\mathcal F}_{A_{i},B_{i}}^{\rho_{i}}$.
\begin{proposition}\label{lemma:good:properties:of:gamma:reductions:indexed}
%
%
Let $I^{ns}$ be the set of all $i{\in}I$ such that ${\mathcal F}_{A_{i},B_{i}}^{\rho_{i}}$ is non-simple.
If $\{{\mathcal F}_{A_{i}^{\prime},B_{i}^{\prime}}^{\rho_{i}^{\prime}}:\ i{\in}I\}$ is a $q$-$\gamma$-reduction of $\{{\mathcal F}_{A_{i},B_{i}}^{\rho_{i}}:\ i{\in}I\}$, $(W,R)$ is the union of $\{{\mathcal F}_{A_{i},B_{i}}^{\rho_{i}}:\ i{\in}I\}$ and $(W^{\prime},R^{\prime})$ is the union of $\{{\mathcal F}_{A_{i}^{\prime},B_{i}^{\prime}}^{\rho_{i}^{\prime}}:\ i{\in}I\}$ then
\begin{enumerate}
\item if $\{{\mathcal F}_{A_{i},B_{i}}^{\rho_{i}}:\ i{\in}I\}$ is dust-bounded then $\{{\mathcal F}_{A_{i}^{\prime},B_{i}^{\prime}}^{\rho_{i}^{\prime}}:\ i{\in}I\}$ is dust-bounded,
\item if $\{{\mathcal F}_{A_{i},B_{i}}^{\rho_{i}}:\ i{\in}I^{ns}\}$ is root-bounded and kernel-bounded then $\{{\mathcal F}_{A_{i}^{\prime},B_{i}^{\prime}}^{\rho_{i}^{\prime}}:\ i{\in}I\}$ is root-bounded and kernel-bounded,
\item for all $s^{\prime}{\in}W^{\prime}$, there exists $s{\in}W$ such that $(W^{\prime}_{s^{\prime}},R^{\prime}_{s^{\prime}})$ is a bounded morphic image of $(W_{s},R_{s})$,
\item the second player has a winning strategy in the Ehrenfeucht-Fra\"\i ss\'e game ${\mathcal G}_{q}((W,R),(W^{\prime},R^{\prime}))$.
\end{enumerate}
\end{proposition}
\begin{proof}
By Lemma~\ref{lemma:gamma:reduction:prebounded:dust:quasi:root:bounded:OK}.
\medskip
\end{proof}
\paragraph{$\delta$-reductions}
%
%
The main effect of a $\delta$-reduction on a disjoint indexed family of galaxies is to reduce its size.
\\
\\
%
%
%
%
Let $\{{\mathcal F}_{A_{i},B_{i}}^{\rho_{i}}:\ i{\in}I\}$ be a disjoint indexed family of galaxies.
\\
\\
Let $\simeq$ be the equivalence relation on $I$ defined by $i{\simeq}j$ if and only if ${\mathcal F}_{A_{i},B_{i}}^{\rho_{i}}$ is isomorphic with ${\mathcal F}_{A_{j},B_{j}}^{\rho_{j}}$.
For all $i{\in}I$, let $\lbrack i\rbrack_{\simeq}$ be the equivalence class of $i$ modulo $\simeq$.
\\
\\
For all $i{\in}I$, if ${\parallel}\lbrack i\rbrack_{\simeq}{\parallel}{\leq}q$ then let $X_{i}$ be $\lbrack i\rbrack_{\simeq}$ else let $X_{i}$ be an arbitrary subset of $\lbrack i\rbrack_{\simeq}$ of cardinality $q$.
Let $J$ be $\bigcup\{X_{i}:\ i{\in}I\}$.
Obviously, from each equivalence class modulo $\simeq$, $J$ contains at most $q$ elements.
Therefore, if $\simeq$ possesses $N$ equivalence classes then $J$ contains at most $q{\times}N$ elements.
\\
\\
$\{{\mathcal F}_{A_{j},B_{j}}^{\rho_{j}}:\ j{\in}J\}$ is a {\it $q$-$\delta$-reduction of $\{{\mathcal F}_{A_{i},B_{i}}^{\rho_{i}}:\ i{\in}I\}$.}
\begin{proposition}\label{lemma:good:properties:of:delta:reductions:indexed}
%
%
%
%
If $(W,R)$ is the union of $\{{\mathcal F}_{A_{i},B_{i}}^{\rho_{i}}:\ i{\in}I\}$ and $(W^{\prime},R^{\prime})$ is the union of $\{{\mathcal F}_{A_{j},B_{j}}^{\rho_{j}}:\ j{\in}J\}$ then
\begin{enumerate}
\item If $\{{\mathcal F}_{A_{i},B_{i}}^{\rho_{i}}:\ i{\in}I\}$ is dust-boun\-ded (respectively root-boun\-ded, kernel-boun\-ded) then $\{{\mathcal F}_{A_{j},B_{j}}^{\rho_{j}}:\ j{\in}J\}$ is dust-boun\-ded (respectively root-boun\-ded, kernel-boun\-ded,
%
%
%
%
\item for all $Q{\geq}1$, if for all $i{\in}I$, ${\parallel}A_{i}{\parallel}{\leq}Q$ and ${\parallel}B_{i}{\parallel}{\leq}Q$ then ${\parallel}J{\parallel}{\leq}q{\times}(Q{+}1)^{2}{\times}
$\linebreak$
2^{Q^{2}}$,
\item $(W^{\prime},R^{\prime})$ is a bounded morphic image of $(W,R)$,
\item the second player has a winning strategy in the Ehrenfeucht-Fra\"\i ss\'e game ${\mathcal G}_{q}((W,R),(W^{\prime},R^{\prime}))$.
\end{enumerate}
\end{proposition}
\begin{proof}
$\mathbf{(1)}$ and $\mathbf{(2)}$~By construction of $\{{\mathcal F}_{A_{j},B_{j}}^{\rho_{j}}:\ j{\in}J\}$.
\\
\\
$\mathbf{(3)}$~Let $h:\ I{\longrightarrow}J$ be a function such that for all $i{\in}I$, either $i{\in}J$ and $h(i){=}i$, or $i{\not\in}J$ and $h(i){\simeq}i$.
Notice that the existence of such a function is a consequence of the definition of $J$.
Notice also that $h$ is surjective.
Moreover, for all $i{\in}I$, $h(i){\simeq}i$.
For all $i{\in}I$, let $g_{i}$ be an isomorphism from ${\mathcal F}_{A_{i},B_{i}}^{\rho_{i}}$ to ${\mathcal F}_{A_{h(i)},B_{h(i)}}^{\rho_{h(i)}}$.
Let $f:\ W{\longrightarrow}W^{\prime}$ be the function such that for all $i{\in}I$ and for all $s{\in}A_{i}{\cup}B_{i}$, $f(s){=}g_{i}(s)$.
Notice that $f$ is a bounded morphism from $(W,R)$ to $(W^{\prime},R^{\prime})$.
Moreover, $f$ is surjective.
\\
\\
$\mathbf{(4)}$~For all $i{\in}I$ and for all $j{\in}J$, if $i{\simeq}j$ then let $f_{i,j}$ be an isomorphism from ${\mathcal F}_{A_{i},B_{i}}^{\rho_{i}}$ to ${\mathcal F}_{A_{j},B_{j}}^{\rho_{j}}$.
Let $q^{\prime}{\in}\lsem1,q\rsem$ and $((s_{1},\ldots,s_{q^{\prime}{-}1}),(s^{\prime}_{1},\ldots,s^{\prime}_{q^{\prime}{-}1}))$ be a winning position of the Ehrenfeucht-Fra\"\i ss\'e game ${\mathcal G}_{q}((W,R),(W^{\prime},R^{\prime}))$.
Suppose the first player picked a member $s$ of $(W,R)$ in its $q^{\prime}$th move.
Let $i{\in}I$ be such that $s$ is a member of ${\mathcal F}_{A_{i},B_{i}}^{\rho_{i}}$, $i_{1},\ldots,i_{q^{\prime}{-}1}{\in}I$ be such that $s_{1}$ is a member of ${\mathcal F}_{A_{i_{1}},B_{i_{1}}}^{\rho_{i_{1}}}$, $\ldots$, $s_{q^{\prime}{-}1}$ is a member of ${\mathcal F}_{A_{i_{q^{\prime}{-}1}},B_{i_{q^{\prime}{-}1}}}^{\rho_{i_{q^{\prime}{-}1}}}$ and $j_{1},\ldots,j_{q^{\prime}{-}1}{\in}J$ be such that $s^{\prime}_{1}$ is a member of ${\mathcal F}_{A_{j_{1}},B_{j_{1}}}^{\rho_{j_{1}}}$, $\ldots$, $s^{\prime}_{q^{\prime}{-}1}$ is a member of ${\mathcal F}_{A_{j_{q^{\prime}{-}1}},B_{j_{q^{\prime}{-}1}}}^{\rho_{j_{q^{\prime}{-}1}}}$.
Then
\begin{itemize}
\item if $i{=}i_{q^{\prime\prime}}$ for some $q^{\prime\prime}{\in}\lsem1,q^{\prime}{-}1\rsem$ then the second player picks $f_{i_{q^{\prime\prime}},j_{q^{\prime\prime}}}(s)$ in its $q^{\prime}$th move,
\item otherwise, $j$ being an element in $\lbrack i\rbrack_{\simeq}{\cap}(J{\setminus}\{j_{1},\ldots,j_{q^{\prime}{-}1}\})$, the second play\-er picks $f_{i,j}(s)$ in its $q^{\prime}$th move,
\end{itemize}
Suppose the first player picked a member $s^{\prime}$ of $(W^{\prime},R^{\prime})$ in its $q^{\prime}$th move.
Let $j{\in}J$ be such that $s^{\prime}$ is a member of ${\mathcal F}_{A_{j},B_{j}}^{\rho_{j}}$, $i_{1},\ldots,i_{q^{\prime}{-}1}{\in}I$ be such that $s_{1}$ is a member of ${\mathcal F}_{A_{i_{1}},B_{i_{1}}}^{\rho_{i_{1}}}$, $\ldots$, $s_{q^{\prime}{-}1}$ is a member of ${\mathcal F}_{A_{i_{q^{\prime}{-}1}},B_{i_{q^{\prime}{-}1}}}^{\rho_{i_{q^{\prime}{-}1}}}$ and $j_{1},\ldots,j_{q^{\prime}{-}1}{\in}J$ be such that $s^{\prime}_{1}$ is a member of ${\mathcal F}_{A_{j_{1}},B_{j_{1}}}^{\rho_{j_{1}}}$, $\ldots$, $s^{\prime}_{q^{\prime}{-}1}$ is a member of ${\mathcal F}_{A_{j_{q^{\prime}{-}1}},B_{j_{q^{\prime}{-}1}}}^{\rho_{j_{q^{\prime}{-}1}}}$.
Then
\begin{itemize}
\item if $j{=}j_{q^{\prime\prime}}$ for some $q^{\prime\prime}{\in}\lsem1,q^{\prime}{-}1\rsem$ then the second player picks $f^{{-}1}_{i_{q^{\prime\prime}},j_{q^{\prime\prime}}}(s^{\prime})$ in its $q^{\prime}$th move,
\item otherwise, $i$ being an element in $\lbrack j\rbrack_{\simeq}{\cap}(I{\setminus}\{i_{1},\ldots,i_{q^{\prime}{-}1}\})$, the second play\-er picks $f^{{-}1}_{i,j}(s)$ in its $q^{\prime}$th move,
\end{itemize}
Obviously, $((s_{1},\ldots,s_{q^{\prime}{-}1},s),(s^{\prime}_{1},\ldots,s^{\prime}_{q^{\prime}{-}1},s^{\prime}))$ is a winning position of the Eh\-renfeucht-Fra\"\i ss\'e game ${\mathcal G}_{q}((W,R),(W^{\prime},R^{\prime}))$.
\medskip
\end{proof}
\section{Propositional Modal Logic}
In this section, we introduce tools and techniques in Propositional Modal Logic that we will need later in our proofs about modal definability.
\paragraph{Modal syntax}
Let us consider a countably infinite set $\PVAR$ of {\it propositional variables}\/ (denoted $\mathbf{p}$, $\mathbf{q}$, $\ldots$).
\\
\\
The {\it modal formulas}\/ (denoted $\varphi$, $\psi$, $\ldots$) are inductively defined as follows:
\begin{itemize}
\item $\varphi,\psi::=\mathbf{p}\mid\bot\mid\neg\varphi\mid(\varphi\vee\psi)\mid\Box\varphi$.
\end{itemize}
We adopt the standard rules for omission of the parentheses.
\\
\\
We define the other Boolean constructs as usual.
The modal formula $\Diamond\varphi$ is obtained as the well-known abbreviation
\begin{itemize}
\item $\Diamond\varphi::=\neg\Box\neg\varphi$.
\end{itemize}
We will also use the abbreviations
\begin{itemize}
\item $\lbrack U\rbrack\varphi::=\varphi\wedge\Box\Box\varphi$,
\item $\langle U\rangle\varphi::=\neg\lbrack U\rbrack\neg\varphi$.
\end{itemize}
The {\it set of all propositional variables occurring in the modal formula $\varphi$}\/ (denoted $\vvar(\varphi)$) is defined as usual.
\\
\\
For all modal formulas $\varphi$, let $\size(\varphi)$ be the number of symbols occurring in $\varphi$.
\paragraph{Modal semantics}
Let $(W,R)$ be a frame.
\\
\\
A {\it valuation on $(W,R)$}\/ is a function $V:\ \PVAR{\longrightarrow}{\cal P}(W)$.
\\
\\
The {\it satisfiability of a modal formula $\varphi$ at $s{\in}W$ with respect to a valuation $V$ on $(W,R)$}\/ (denoted $(W,R),V,s{\models}\varphi$) is inductively defined as follows:
\begin{itemize}
\item $(W,R),V,s{\models}\mathbf{p}$ if and only if $s{\in}V(\mathbf{p})$,
\item $(W,R),V,s{\not\models}\bot$,
\item $(W,R),V,s{\models}\neg\varphi$ if and only if $(W,R),V,s{\not\models}\varphi$,
\item $(W,R),V,s{\models}\varphi\vee\psi$ if and only if $(W,R),V,s{\models}\varphi$ or $(W,R),V,s{\models}\psi$,
\item $(W,R),V,s{\models}\Box\varphi$ if and only if for all $t{\in}W$, if $s{R}t$ then $(W,R),V,t{\models}\varphi$.
\end{itemize}
As a result,
\begin{itemize}
\item $(W,R),V,s{\models}\Diamond\varphi$ if and only if there exists $t{\in}W$ such that $s{R}t$ and $(W,R),
$\linebreak$
V,t{\models}\varphi$.
\end{itemize}
Moreover, if $(W,R)$ is Euclidean then
\begin{itemize}
\item $(W,R),V,s{\models}\lbrack U\rbrack\varphi$ if and only if for all $t{\in}W_{s}$, $(W,R),V,t{\models}\varphi$,
\item $(W,R),V,s{\models}\langle U\rangle\varphi$ if and only if there exists $t{\in}W_{s}$ such that $(W,R),V,t{\models}
$\linebreak$
\varphi$.
\end{itemize}
A modal formula $\varphi$ is {\it true with respect to a valuation $V$ on $(W,R)$}\/ (denoted $(W,R),V{\models}\varphi$) if $\varphi$ is satisfied at all $s{\in}W$ with respect to $V$.
A set $\Sigma$ of modal formulas is {\it true with respect to a valuation $V$ on $(W,R)$}\/ (denoted $(W,R),V{\models}\Sigma$) if for all modal formulas $\varphi$ in $\Sigma$, $(W,R),V{\models}\varphi$.
A modal formula $\varphi$ is {\it valid in $(W,R)$}\/ (denoted $(W,R){\models}\varphi$) if $\varphi$ is true with respect to all valuations on $(W,R)$.
A set $\Sigma$ of modal formulas is {\it valid in $(W,R)$}\/ (denoted $(W,R){\models}\Sigma$) if for all modal formulas $\varphi$ in $\Sigma$, $(W,R){\models}\varphi$.
%
%
\begin{lemma}\label{lemma:validity:in:a:given:frame:is:in:coNP}
The following decision problem is in $\coNP$:
\begin{description}
\item[input:] a modal formula $\varphi$ and a finite frame $(W^{\prime\prime},R^{\prime\prime})$,
\item[output:] determine whether $(W^{\prime\prime},R^{\prime\prime}){\models}\varphi$.
\end{description}
\end{lemma}
\begin{proof}
See~\cite{Arnold:Crubille:1988,Clarke:et:al:1986}.
\medskip
\end{proof}
\begin{lemma}\label{lemma:about:upper:bound:on:the:kernel:1}
Let $\varphi$ be a modal formula.
If there exists $m{\in}\mathbf{N}^{+}$ such that ${\mathcal F}_{m}^{2}{\not\models}\varphi$ then there exists $m{\in}\lsem1,2^{{\parallel}\vvar(\varphi){\parallel}}\rsem$ such that ${\mathcal F}_{m}^{2}{\not\models}\varphi$.
\end{lemma}
\begin{lemma}\label{lemma:about:upper:bound:on:the:kernel:2}
Let $\varphi$ be a modal formula.
If there exists $n{\in}\mathbf{N}^{-}$ such that ${\mathcal F}_{2}^{n}{\not\models}\varphi$ then there exists $n{\in}\lsem{-}1,2^{{\parallel}\vvar(\varphi){\parallel}}\rsem$ such that ${\mathcal F}_{2}^{n}{\not\models}\varphi$.
\end{lemma}
A frame $(W^{\prime},R^{\prime})$ is {\it modally weaker than $(W,R)$}\/ (denoted $(W^{\prime},R^{\prime}){\preceq}(W,R)$) if for all modal formulas $\varphi$, if $(W^{\prime},R^{\prime}){\models}\varphi$ then $(W,R){\models}\varphi$.
\begin{lemma}[Generated Subframe Lemma]\label{Generated:Subframe:Lemma}
For all frames $(W^{\prime},R^{\prime})$, if $(W,
$\linebreak$
R)$ is a generated subframe of $(W^{\prime},R^{\prime})$ then $(W^{\prime},R^{\prime}){\preceq}(W,R)$.
\end{lemma}
\begin{proof}
See~\cite[Chapter~$3$]{Blackburn:deRijke:Venema:2001}.
\medskip
\end{proof}
\begin{lemma}[Bounded Morphism Lemma]\label{Bounded:Morphism:Lemma}
For all frames $(W^{\prime},R^{\prime})$, if $(W,
$\linebreak$
R)$ is a bounded morphic image of $(W^{\prime},R^{\prime})$ then $(W^{\prime},R^{\prime}){\preceq}(W,R)$.
\end{lemma}
\begin{proof}
See~\cite[Chapter~$3$]{Blackburn:deRijke:Venema:2001}.
\medskip
\end{proof}
\begin{lemma}\label{lemma:sigma:about:K2:galaxies:and:generated:by:s:subframes}
Let ${\mathcal F}_{A,B}^{\rho}$ be a headed galaxy.
For all modal formulas $\varphi$, if ${\mathcal F}_{A,B}^{\rho}{\not\models}\varphi$ then there exists $s{\in}A$ such that $(W_{s},R_{s}){\not\models}\varphi$ where $(W_{s},R_{s})$ is the least generated subframe of ${\mathcal F}_{A,B}^{\rho}$ containing $s$.
\end{lemma}
%
%
%
%
%
%
%
%
\begin{lemma}\label{lemma:about:finite:subset:when:K5:shape:of:the:frame}
If $W{=}\{s\}{\cup}C$ and $R{=}(\{s\}{\times}B){\cup}(C{\times}C)$ for some sets $B,C$ and for some element $s$ such that $B{\not=}\emptyset$, $C{\setminus}B{\not=}\emptyset$, $B$ is finite, $B{\subseteq}C$ and $s{\not\in}C$ then for all modal formulas $\varphi$, if $(W,R){\not\models}\varphi$ then there exists a flower ${\mathcal F}_{m}^{n}$ such that $m{=}{\parallel}B{\parallel}$ and ${\mathcal F}_{m}^{n}{\not\models}\varphi$.
\end{lemma}
\begin{lemma}\label{lemma:about:finite:subset:when:K5:shape:of:the:frame:bis}
If $W{=}\{s\}{\cup}C$ and $R{=}(\{s\}{\times}B){\cup}(C{\times}C)$ for some sets $B,C$ and for some element $s$ such that $B{\not=}\emptyset$, $C{\setminus}B{\not=}\emptyset$, $C{\setminus}B$ is finite, $B{\subseteq}C$ and $s{\not\in}C$ then for all modal formulas $\varphi$, if $(W,R){\not\models}\varphi$ then there exists a flower ${\mathcal F}_{m}^{n}$ such that $n{=}{\parallel}C{\setminus}B{\parallel}$ and ${\mathcal F}_{m}^{n}{\not\models}\varphi$.
\end{lemma}
Let $q{\geq}3$.
\\
\\
Let $\{{\mathcal F}_{A_{i},B_{i}}^{\rho_{i}}:\ i{\in}I\}$ be a disjoint indexed family of galaxies.
\begin{lemma}\label{lemma:alpha:reduction:modal:logic}
%
%
If $\{{\mathcal F}_{A_{i}^{\prime},B_{i}^{\prime}}^{\rho_{i}^{\prime}}:\ i{\in}I\}$ is a $q$-$\alpha$-reduction of $\{{\mathcal F}_{A_{i},B_{i}}^{\rho_{i}}:\ i{\in}I\}$, $(W,R)$ is the union of $\{{\mathcal F}_{A_{i},B_{i}}^{\rho_{i}}:\ i{\in}I\}$ and $(W^{\prime},R^{\prime})$ is the union of $\{{\mathcal F}_{A_{i}^{\prime},B_{i}^{\prime}}^{\rho_{i}^{\prime}}:\ i{\in}I\}$ then $(W,R){\preceq}(W^{\prime},R^{\prime})$.
\end{lemma}
\begin{proof}
By Proposition~\ref{lemma:good:properties:of:alpha:reductions:indexed} and Lemma~\ref{Bounded:Morphism:Lemma}.
\medskip
\end{proof}
\begin{lemma}\label{lemma:gamma:reduction:modal:logic}
%
%
If $\{{\mathcal F}_{A_{i}^{\prime},B_{i}^{\prime}}^{\rho_{i}^{\prime}}:\ i{\in}I\}$ is a $q$-$\gamma$-reduction of $\{{\mathcal F}_{A_{i},B_{i}}^{\rho_{i}}:\ i{\in}I\}$, $(W,R)$ is the union of $\{{\mathcal F}_{A_{i},B_{i}}^{\rho_{i}}:\ i{\in}I\}$ and $(W^{\prime},R^{\prime})$ is the union of $\{{\mathcal F}_{A_{i}^{\prime},B_{i}^{\prime}}^{\rho_{i}^{\prime}}:\ i{\in}I\}$ then
$(W,R){\preceq}(W^{\prime},R^{\prime})$.
\end{lemma}
\begin{proof}
By Proposition~\ref{lemma:good:properties:of:gamma:reductions:indexed} and Lemma~\ref{Bounded:Morphism:Lemma}.
\medskip
\end{proof}
\begin{lemma}\label{lemma:delta:reduction:modal:logic}
%
%
If $\{{\mathcal F}_{A_{j},B_{j}}^{\rho_{j}}:\ j{\in}J\}$ is a $q$-$\delta$-reduction of $\{{\mathcal F}_{A_{i},B_{i}}^{\rho_{i}}:\ i{\in}I\}$, $(W,R)$ is the union of $\{{\mathcal F}_{A_{i},B_{i}}^{\rho_{i}}:\ i{\in}I\}$ and $(W^{\prime},R^{\prime})$ is the union of $\{{\mathcal F}_{A_{j},B_{j}}^{\rho_{j}}:\ j{\in}J\}$ then $(W,R){\preceq}(W^{\prime},R^{\prime})$.
\end{lemma}
\begin{proof}
By Proposition~\ref{lemma:good:properties:of:delta:reductions:indexed} and Lemma~\ref{Bounded:Morphism:Lemma}.
\medskip
\end{proof}
\paragraph{Jankov-Fine formulas}
Now, let us adapt to the isolated point and flowers, the standard concept of Jankov-Fine formulas.
\\
\\
Let $m{\in}\mathbf{N}^{+}{\cup}\{0\}$ and $n{\in}\mathbf{N}^{-}$ be such that either $m{\not=}0$, or $n{=}0$.
\\
\\
Let $\chi_{m}^{n}$ be the modal formula defined as follows:
\begin{itemize}
\item the negation of modal formula
\begin{itemize}
\item $\Box\bot$,
\end{itemize}
when $m{=}0$,
\item the negation of the conjunction of the modal formulas
\begin{itemize}
\item $p_{0}$,
\item $\lbrack U\rbrack\bigvee\{p_{k}:\ 0{\leq}k{\leq}m{+}n\}$,
\item $\lbrack U\rbrack(p_{k}\rightarrow\neg p_{l})$ for each $k,l{\in}\N$ such that $0{\leq}k{<}l{\leq}m{+}n$,
\item $\langle U\rangle p_{k}$ for each $k{\in}\N$ such that $0{\leq}k{\leq}m{+}n$,
\item $\lbrack U\rbrack(p_{0}\rightarrow\Box\neg p_{0})$,
\item $\lbrack U\rbrack(p_{0}\rightarrow\Diamond p_{k})$ for each $k{\in}\N$ such that $1{\leq}k{\leq}m$,
\item $\lbrack U\rbrack(p_{0}\rightarrow\Box\neg p_{k})$ for each $k{\in}\N$ such that $m{+}1{\leq}k{\leq}m{+}n$,
\item $\lbrack U\rbrack(p_{k}\rightarrow\Box\neg p_{0})$ for each $k{\in}\N$ such that $1{\leq}k{\leq}m{+}n$,
\item $\lbrack U\rbrack(p_{k}\rightarrow\Diamond p_{l})$ for each $k,l{\in}\N$ such that $1{\leq}k,l{\leq}m{+}n$,
\end{itemize}
when $m{\not=}0$ and $n{\in}\N$,
\item the negation of the conjunction of the modal formulas
\begin{itemize}
\item $p_{1}$,
\item $\lbrack U\rbrack\bigvee\{p_{k}:\ 1{\leq}k{\leq}m\}$,
\item $\lbrack U\rbrack(p_{k}\rightarrow\neg p_{l})$ for each $k,l{\in}\N$ such that $1{\leq}k{<}l{\leq}m$,
\item $\langle U\rangle p_{k}$ for each $k{\in}\N$ such that $1{\leq}k{\leq}m$,
\item $\lbrack U\rbrack(p_{k}\rightarrow\Diamond p_{l})$ for each $k,l{\in}\N$ such that $1{\leq}k,l{\leq}m$,
\end{itemize}
when $m{\not=}0$ and $n{=}{-}1$.
\end{itemize}
The modal formula $\chi_{m}^{n}$ is called {\it Jankov-Fine formula of the frame ${\mathcal F}_{m}^{n}$.}
\begin{lemma}\label{lemma:Jankov:Fine:ONE}
For all Euclidean frames $(W,R)$ and for all $s{\in}W$, the following conditions are equivalent:
\begin{itemize}
\item there exists a valuation $V$ on $(W,R)$ such that $(W,R),V,s{\not\models}\chi_{m}^{n}$,
\item there exists a surjective bounded morphism $f$ from $(W_{s},R_{s})$ to ${\mathcal F}_{m}^{n}$ such that $f(s)$ generates ${\mathcal F}_{m}^{n}$.
\end{itemize}
\end{lemma}
\begin{proof}
See~\cite[Chapter~$3$]{Blackburn:deRijke:Venema:2001}.
\medskip
\end{proof}
\begin{lemma}\label{lemma:Jankov:Fine:TWO}
${\mathcal F}_{m}^{n}$ does not validate its Jankov-Fine formula.
\end{lemma}
\begin{proof}
By Lemma~\ref{lemma:Jankov:Fine:ONE}.
\medskip
\end{proof}
\paragraph{Normal modal logics}
A {\it normal modal logic}\/ is a set $\L$ of formulas such that
\begin{itemize}
\item $\L$ is closed under the rule of uniform substitution,
\item $\L$ contains all propositional tautologies,
\item for all modal formulas $\varphi,\psi$, $\Box(\varphi\rightarrow\psi)\rightarrow(\Box\varphi\rightarrow\Box\psi){\in}\L$,
\item for all modal formulas $\varphi,\psi$, if $\varphi{\in}\L$ and $\varphi\rightarrow\psi{\in}\L$ then $\psi{\in}\L$,
\item for all modal formulas $\varphi$, if $\varphi{\in}\L$ then $\Box\varphi{\in}\L$.
\end{itemize}
From now on, {\bf when we write ``modal logic'', we mean ``normal modal logic''.}
\\
\\
A modal logic $\L$ is {\it consistent}\/ if $\bot{\not\in}\L$.
\\
\\
A consistent modal logic $\L$ is {\it Euclidean}\/ if for all modal formulas $\varphi$, $\Diamond\varphi\rightarrow\Box\Diamond\varphi{\in}\L$~---~i.e. $\L$ is an extension of the modal logic $\K5$.
\\
\\
The extensions of the modal logic $\K5$ have been studied by Nagle~\cite{Nagle:1981} and Nagle and Thomason~\cite{Nagle:Thomason:1985}.
They possess remarkable properties: they have the poly-size model property, they are finitely axiomatizable, etc.
\\
\\
For all modal formulas $\varphi$, let $\K5{\oplus}\varphi$ be the least extension of $\K5$ containing $\varphi$.
\\
\\
For all Euclidean modal logics $\L$, let $\Fr(\L)$ be the class of all frames $(W,R)$ such that $(W,R){\models}\L$.
\\
\\
For all Euclidean modal logics $\L$, let
\begin{itemize}
\item $\mathtt{S}_{\L}{=}\{(m,n):\ m{\in}\mathbf{N}^{+},\ n{\in}\mathbf{N}^{-}$ and ${\mathcal F}_{m}^{n}{\models}\L\}$.
\end{itemize}
\begin{lemma}\label{lemma:S:L:is:closed:subset:of:N:plus:times:N:minus}
For all Euclidean modal logics $\L$, $\mathtt{S}_{\L}$ is a closed subset of $\mathbf{N}^{+}{\times}\mathbf{N}^{-}$.
\end{lemma}
\begin{proof}
By Lemmas~\ref{lemma:about:bounded:morphisms:in:the:situation:of F:m:n:frames} and~\ref{Bounded:Morphism:Lemma}.
\medskip
\end{proof}
\begin{lemma}\label{lemma:alpha}
Let ${\mathcal F}^{\rho}_{A,B}$ be a galaxy in ${\mathcal K}_{2}$.
If $\{2\}{\times}\mathbf{N}^{-}{\subseteq}\mathtt{S}_{\L}$ then ${\mathcal F}^{\rho}_{A,B}{\models}\L$.
\end{lemma}
\begin{proof}
For the sake of the contradiction, suppose $\{2\}{\times}\mathbf{N}^{-}{\subseteq}\mathtt{S}_{\L}$ and ${\mathcal F}^{\rho}_{A,B}{\not\models}\L$.
Hence, there exists a modal formula $\varphi{\in}\L$ such that ${\mathcal F}^{\rho}_{A,B}{\not\models}\varphi$.
Since ${\mathcal F}^{\rho}_{A,B}$ is a galaxy in ${\mathcal K}_{2}$, then ${\mathcal F}^{\rho}_{A,B}$ is headed and by Lemma~\ref{lemma:sigma:about:K2:galaxies:and:generated:by:s:subframes}, there exists $s{\in}A$ such that $(W_{s},R_{s}){\not\models}\varphi$ where $(W_{s},R_{s})$ is the least generated subframe of ${\mathcal F}_{A,B}^{\rho}$ containing $s$.
Moreover, obviously, $W_{s}{=}\{s\}{\cup}C^{\prime}$ and $R{=}(\{s\}{\times}B^{\prime}){\cup}(C^{\prime}{\times}C^{\prime})$ for some sets $B^{\prime},C^{\prime}$ such that ${\parallel}B^{\prime}{\parallel}{=}2$, ${\parallel}C^{\prime}{\setminus}B^{\prime}{\parallel}{\geq}2$, $B^{\prime}{\subseteq}C^{\prime}$ and $s{\not\in}C^{\prime}$.
Thus, by Lemma~\ref{lemma:about:finite:subset:when:K5:shape:of:the:frame}, there exists a flower ${\mathcal F}_{m}^{n}$ such that $m{=}2$ and ${\mathcal F}_{m}^{n}{\not\models}\varphi$.
Since $\{2\}{\times}\mathbf{N}^{-}{\subseteq}\mathtt{S}_{\L}$, then $(m,n){\in}\mathtt{S}_{\L}$.
Consequently, ${\mathcal F}_{m}^{n}{\models}\L$.
Since $\varphi{\in}\L$, then ${\mathcal F}_{m}^{n}{\models}\varphi$: a contradiction.
\medskip
\end{proof}
\begin{lemma}\label{lemma:beta}
Let ${\mathcal F}^{\rho}_{A,B}$ be a galaxy in ${\mathcal L}_{2}$.
If $\mathbf{N}^{+}{\times}\{2\}{\subseteq}\mathtt{S}_{\L}$ then ${\mathcal F}^{\rho}_{A,B}{\models}\L$.
\end{lemma}
\begin{proof}
For the sake of the contradiction, suppose $\mathbf{N}^{+}{\times}\{2\}{\subseteq}\mathtt{S}_{\L}$ and ${\mathcal F}^{\rho}_{A,B}{\not\models}\L$.
Hence, there exists a modal formula $\varphi{\in}\L$ such that ${\mathcal F}^{\rho}_{A,B}{\not\models}\varphi$.
Since ${\mathcal F}^{\rho}_{A,B}$ is a galaxy in ${\mathcal L}_{2}$, then ${\mathcal F}^{\rho}_{A,B}$ is headed and by Lemma~\ref{lemma:sigma:about:K2:galaxies:and:generated:by:s:subframes}, there exists $s{\in}A$ such that $(W_{s},R_{s}){\not\models}\varphi$ where $(W_{s},R_{s})$ is the least generated subframe of ${\mathcal F}_{A,B}^{\rho}$ containing $s$.
Moreover, obviously, $W_{s}{=}\{s\}{\cup}C^{\prime}$ and $R{=}(\{s\}{\times}B^{\prime}){\cup}(C^{\prime}{\times}C^{\prime})$ for some sets $B^{\prime},C^{\prime}$ such that ${\parallel}B^{\prime}{\parallel}{\geq}2$, ${\parallel}C^{\prime}{\setminus}B^{\prime}{\parallel}{=}2$, $B^{\prime}{\subseteq}C^{\prime}$ and $s{\not\in}C^{\prime}$.
Thus, by Lemma~\ref{lemma:about:finite:subset:when:K5:shape:of:the:frame:bis}, there exists a flower ${\mathcal F}_{m}^{n}$ such that $n{=}2$ and ${\mathcal F}_{m}^{n}{\not\models}\varphi$.
Since $\mathbf{N}^{+}{\times}\{2\}{\subseteq}\mathtt{S}_{\L}$, then $(m,n){\in}\mathtt{S}_{\L}$.
Consequently, ${\mathcal F}_{m}^{n}{\models}\L$.
Since $\varphi{\in}\L$, then ${\mathcal F}_{m}^{n}{\models}\varphi$: a contradiction.
\medskip
\end{proof}
\section{Model Theory of Classical First-Order Logic}\label{section:first:order:syntax:and:semantics}
In this section, we introduce tools and techniques in Model Theory of Classical First-Order Logic that we will need later in our proofs about modal definability.
\paragraph{First-order syntax}
Let us consider a countably infinite set $\IVAR$ of {\it individual variables}\/ (denoted $\mathbf{x}$, $\mathbf{y}$, $\ldots$).
\\
\\
The {\it first-order formulas}\/ (denoted $A$, $B$, $\ldots$) are inductively defined as follows:
\begin{itemize}
\item $A,B::=\mathbf{R}(\mathbf{x},\mathbf{y})\mid\mathbf{x}{=}\mathbf{y}\mid\neg A\mid(A\vee B)\mid\forall\mathbf{x}A$.
\end{itemize}
We adopt the standard rules for omission of the parentheses.
\\
\\
We define the other Boolean constructs as usual.
The first-order formula $\exists\mathbf{x}A$ is obtained as the well-known abbreviation
\begin{itemize}
\item $\exists\mathbf{x}A::=\neg\forall\mathbf{x}\neg A$.
\end{itemize}
We will also use the abbreviations
\begin{itemize}
\item $\exists^{{=}1}\mathbf{x}A(\mathbf{x})::=\exists\mathbf{x}(A(\mathbf{x})\wedge\forall\mathbf{y}(A(\mathbf{y})\rightarrow\mathbf{y}{=}\mathbf{x}))$,
\item $\exists^{{=}2}\mathbf{x}A(\mathbf{x})::=\exists\mathbf{x_{1}}\exists\mathbf{x_{2}}(A(\mathbf{x_{1}})\wedge A(\mathbf{x_{2}})\wedge\mathbf{x_{1}}{\not=}\mathbf{x_{2}}\wedge\forall\mathbf{y}(A(\mathbf{y})\rightarrow\mathbf{y}{=}\mathbf{x_{1}}\vee\mathbf{y}{=}\mathbf{x_{2}}))$,
\item $\{\mathbf{x_{1}},\mathbf{x_{2}}\}{\cap}\{\mathbf{y_{1}},\mathbf{y_{2}}\}{=}\emptyset::=\mathbf{x_{1}}{\not=}\mathbf{y_{1}}\wedge\mathbf{x_{1}}{\not=}\mathbf{y_{2}}\wedge\mathbf{x_{2}}{\not=}\mathbf{y_{1}}\wedge\mathbf{x_{2}}{\not=}\mathbf{y_{2}}$,
\item $\mathbf{R}(\mathbf{z}){=}\{\mathbf{x_{i}},\mathbf{y_{j}}\}::=\forall\mathbf{t}(\mathbf{R}(\mathbf{z},\mathbf{t})\leftrightarrow\mathbf{t}{=}\mathbf{x_{i}}\vee\mathbf{t}{=}\mathbf{y_{j}})$,
\item $\mathbf{R}(\mathbf{z}){=}\{\mathbf{x_{i}},\mathbf{y_{j}}\}^{\mathbf{c}}::=\forall\mathbf{t}(\mathbf{R}(\mathbf{z},\mathbf{t})\leftrightarrow\mathbf{R}(\mathbf{t},\mathbf{t})\wedge\mathbf{t}{\not=}\mathbf{x_{i}}\wedge\mathbf{t}{\not=}\mathbf{y_{j}})$,
\item $\{\mathbf{x_{1}},\mathbf{x_{2}}\}{=}\{\mathbf{y_{1}},\mathbf{y_{2}}\}::=(\mathbf{x_{1}}{=}\mathbf{y_{1}}\wedge\mathbf{x_{2}}{=}\mathbf{y_{2}})\vee(\mathbf{x_{1}}{=}\mathbf{y_{2}}\wedge\mathbf{x_{2}}{=}\mathbf{y_{1}})$.
\end{itemize}
The {\it set of all free individual variables occurring in the first-order formula $A$}\/ (denoted $\fiv(A)$) is defined as usual.
The first-order formula $A$ is a {\it sentence}\/ if $\fiv(A){=}\emptyset$.
\\
\\
For all first-order formulas $A$, let $\size(A)$ be the number of symbols occurring in $A$.
\\
\\
The {\it quantifier depth of the first-order formula $A$}\/ (denoted $\qd(A)$) is defined as usual.
For all first-order formulas $A$, let $\qdd(A){=}\max\{\qd(A),3\}$.
\\
\\
Obviously, for all first-order formulas $A$, $\qd(A){\leq}\qdd(A)$ and $\qdd(A){\geq}3$.
\\
\\
For all integers $k$ such that $k{\geq}4$ and for all first-order formulas $A$, let $\Pi_{k}^{A}{=}\{(m,
$\linebreak$
n):\ m{\in}\mathbf{N}^{+}$ and $n{\in}\mathbf{N}^{-}$ are such that if $m{\geq}2$ and $n{\geq}2$ then $m{+}n{\leq}k$, if $m{=}1$ then $n{\leq}\qdd(A)$ and if either $n{=}{-}1$, or $n{=}0$, or $n{=}1$ then $m{\leq}\qdd(A)\}$.
\begin{lemma}\label{lemma:ll:finitely:many:PI:PI:pairs}
For all integers $k$ such that $k{\geq}4$ and for all first-order formulas $A$, $\Pi_{k}^{A}{\subseteq}\lsem1,\max\{k,\qdd(A)\}\rsem{\times}\lsem{-}1,\max\{k,\qdd(A)\}\rsem$.
\end{lemma}
\paragraph{First-order semantics}
Let $(W,R)$ be a frame.
\\
\\
An {\it assignment on $(W,R)$}\/ is a function $g:\ \IVAR{\longrightarrow}W$.
The {\it update of an assignment $g$ on $(W,R)$ with respect to $s{\in}W$ and an individual variable $\mathbf{x}$}\/ is the assignment $g_{s}^{\mathbf{x}}$ on $(W,R)$ such that for all individual variables $\mathbf{y}$,
\begin{itemize}
\item $g_{s}^{\mathbf{x}}(\mathbf{y}){=}s$ when $\mathbf{y}{=}\mathbf{x}$,
\item $g_{s}^{\mathbf{x}}(\mathbf{y}){=}g(\mathbf{y})$ when $\mathbf{y}{\not=}\mathbf{x}$.
\end{itemize}
The {\it satisfiability of a first-order formula $A$ with respect to an assignment $g$ on $(W,R)$}\/ (denoted $(W,R),g\models A$) is inductively defined as follows:
\begin{itemize}
\item $(W,R),g{\models}\mathbf{R}(\mathbf{x},\mathbf{y})$ if and only if $g(\mathbf{x}){R}g(\mathbf{y})$,
\item $(W,R),g{\models}\mathbf{x}{=}\mathbf{y}$ if and only if $g(\mathbf{x}){=}g(\mathbf{y})$,
%
%
%
%
\item $(W,R),g{\models}\neg A$ if and only if $(W,R),g{\not\models}A$,
\item $(W,R),g{\models}A\vee B$ if and only if $(W,R),g{\models}A$ or $(W,R),g{\models}B$,
\item $(W,R),g{\models}\forall\mathbf{x}A$ if and only if for all $s{\in}W$, $(W,R),g_{s}^{\mathbf{x}}{\models}A$.
\end{itemize}
As a result,
\begin{itemize}
\item $(W,R),g{\models}\exists\mathbf{x}A$ if and only if there exists $s{\in}W$ such that $(W,R),g_{s}^{\mathbf{x}}{\models}A$.
\end{itemize}
Moreover,
\begin{itemize}
\item $(W,R),g{\models}\exists^{{=}1}\mathbf{x}A(\mathbf{x})$ if and only if there exists exactly one $s{\in}W$ such that $(W,R),g_{s}^{\mathbf{x}}{\models}A$,
\item $(W,R),g{\models}\exists^{{=}2}\mathbf{x}A(\mathbf{x})$ if and only if there exists exactly two $s{\in}W$ such that $(W,R),g_{s}^{\mathbf{x}}{\models}A$,
\item $(W,R),g{\models}\{\mathbf{x_{1}},\mathbf{x_{2}}\}{\cap}\{\mathbf{y_{1}},\mathbf{y_{2}}\}{=}\emptyset$ if and only if $\{g(\mathbf{x_{1}}),g(\mathbf{x_{2}})\}{\cap}\{g(\mathbf{y_{1}}),
$\linebreak$
g(\mathbf{y_{2}})\}{=}\emptyset$,
\item $(W,R),g{\models}\mathbf{R}(\mathbf{z}){=}\{\mathbf{x_{i}},\mathbf{y_{j}}\}$ if and only if $R(g(\mathbf{z})){=}\{g(\mathbf{x_{i}}),g(\mathbf{y_{j}})\}$,
\item $(W,R),g{\models}\mathbf{R}(\mathbf{z}){=}\{\mathbf{x_{i}},\mathbf{y_{j}}\}^{\mathbf{c}}$ if and only if $R(g(\mathbf{z})){=}\{t{\in}W:\ t{R}t\}{\setminus}\{g(\mathbf{x_{i}}),
$\linebreak$
g(\mathbf{y_{j}})\}$
%
%
\item $(W,R),g{\models}\{\mathbf{x_{1}},\mathbf{x_{2}}\}{=}\{\mathbf{y_{1}},\mathbf{y_{2}}\}$ if and only if $\{g(\mathbf{x_{1}}),g(\mathbf{x_{2}})\}{=}\{g(\mathbf{y_{1}}),
$\linebreak$
g(\mathbf{y_{2}})\}$.
\end{itemize}
From now on, for all first-order formulas $A(\mathbf{x_{1}},\ldots,\mathbf{x_{m}})$ with free individual variables $\mathbf{x_{1}},\ldots,\mathbf{x_{m}}$ and for all $s_{1},\ldots,s_{m}\in W$, {\bf when we write ``$(W,R){\models}A(\mathbf{x_{1}},
$\linebreak$
\ldots,\mathbf{x_{m}})\ \lbrack s_{1},\ldots,s_{m}\rbrack$'', we shall mean ``for all assignments $g$ on $(W,R)$, if $g(\mathbf{x_{1}}){=}s_{1}$, $\ldots$, $g(\mathbf{x_{m}}){=}s_{m}$ then $(W,R),g{\models}A(\mathbf{x_{1}},\ldots,\mathbf{x_{m}})$''.}
\\
\\
A first-order formula $A$ is {\it valid in $(W,R)$}\/ (denoted $(W,R){\models}A$) if $A$ is satisfied with respect to all assignments on $(W,R)$.
A set $\Sigma$ of first-order formulas is {\it valid in $(W,R)$}\/ (denoted $(W,R){\models}\Sigma$) if for all first-order formulas $A$ in $\Sigma$, $(W,R){\models}A$.
%
%
\begin{lemma}\label{lemma:validity:in:a:given:frame:is:in:coNP:fol}
The following decision problem is in $\PSPACE$:
\begin{description}
\item[input:] a sentence $A$ and a finite frame $(W^{\prime\prime},R^{\prime\prime})$,
\item[output:] determine whether $(W^{\prime\prime},R^{\prime\prime}){\models}A$.
\end{description}
\end{lemma}
\begin{proof}
See~\cite{Vardi:1982}.
\medskip
\end{proof}
%
%
The {\it theory of a set $T$ of sentences}\/ is the set of all sentences satisfied on every frame satisfying $T$'s sentences.
A set of sentences is a {\it theory}\/ if it is the theory of a set of sentences.
\\
\\
A theory $T$ is {\it hereditarily undecidable}\/ if for all theories $T^{\prime}$, if $T{\supseteq}T^{\prime}$ then $T^{\prime}$ is undecidable.
\\
\\
For all classes ${\mathcal C}$ of frames, let $\Th({\mathcal C})$ be the set (called {\it theory of ${\mathcal C}$}) of all sentences which are valid in every ${\mathcal C}$'s frame.
\begin{lemma}\label{lemma:gamma}
If $\{2\}{\times}\mathbf{N}^{-}{\subseteq}\mathtt{S}_{\L}$ then $\Th(\Fr(\L))$ is contained in the first-order
\linebreak
theory of ${\mathcal K}_{2}$.
\end{lemma}
\begin{proof}
For the sake of the contradiction, suppose $\{2\}{\times}\mathbf{N}^{-}{\subseteq}\mathtt{S}_{\L}$ and $\Th(\Fr(\L))$ is not contained in the first-order theory of ${\mathcal K}_{2}$.
Hence, there exists a sentence $A$ such that $A$ is valid in every $\Fr(\L)$'s frame and $A$ is not valid in some ${\mathcal K}_{2}$'s galaxy.
Thus, let ${\mathcal F}_{A,B}^{\rho}$ be a ${\mathcal K}_{2}$'s galaxy such that ${\mathcal F}_{A,B}^{\rho}{\not\models}A$.
Since $\{2\}{\times}\mathbf{N}^{-}{\subseteq}\mathtt{S}_{\L}$, then by Lemma~\ref{lemma:alpha}, ${\mathcal F}_{A,B}^{\rho}{\models}\L$.
Since $A$ is valid in every $\Fr(\L)$'s frame, then ${\mathcal F}_{A,B}^{\rho}{\models}A$: a contradiction.
\medskip
\end{proof}
\begin{lemma}\label{lemma:delta}
If $\mathbf{N}^{+}{\times}\{2\}{\subseteq}\mathtt{S}_{\L}$ then $\Th(\Fr(\L))$ is contained in the first-order
\linebreak
theory of ${\mathcal L}_{2}$.
\end{lemma}
\begin{proof}
For the sake of the contradiction, suppose $\mathbf{N}^{+}{\times}\{2\}{\subseteq}\mathtt{S}_{\L}$ and $\Th(\Fr(\L))$ is not contained in the first-order theory of ${\mathcal L}_{2}$.
Hence, there exists a sentence $A$ such that $A$ is valid in every $\Fr(\L)$'s frame and $A$ is not valid in some ${\mathcal L}_{2}$'s galaxy.
Thus, let ${\mathcal F}_{A,B}^{\rho}$ be a ${\mathcal L}_{2}$'s galaxy such that ${\mathcal F}_{A,B}^{\rho}{\not\models}A$.
Since $\mathbf{N}^{+}{\times}\{2\}{\subseteq}\mathtt{S}_{\L}$, then by Lemma~\ref{lemma:beta}, ${\mathcal F}_{A,B}^{\rho}{\models}\L$.
Since $A$ is valid in every $\Fr(\L)$'s frame, then ${\mathcal F}_{A,B}^{\rho}{\models}A$: a contradiction.
\medskip
\end{proof}
For all $q{\in}\N$, a frame $(W^{\prime},R^{\prime})$ is {\it elementary $q$-equivalent to $(W,R)$}\/ (denoted $(W^{\prime},R^{\prime}){\equiv_{q}}(W,R)$) if for all sentences $A$, if $\qd(A){\leq}q$ then $(W^{\prime},R^{\prime}){\models}A$ if and only if $(W,R){\models}A$.
\begin{lemma}\label{ehrenfeucht:theorem}
For all $q{\in}\N$ and for all frames $(W^{\prime},R^{\prime})$, if the second player has a winning strategy in the Ehrenfeucht-Fra\"\i ss\'e game ${\mathcal G}_{q}((W,R),(W^{\prime},R^{\prime}))$ then $(W^{\prime},R^{\prime}){\equiv_{q}}(W,R)$.
\end{lemma}
\begin{proof}
See~\cite[Theorem~$2.2.8$ at Page~$18$]{Ebbinghaus:Flum:1999}.
\medskip
\end{proof}
Let $q{\geq}3$.
\\
\\
Let $\{{\mathcal F}_{A_{i},B_{i}}^{\rho_{i}}:\ i{\in}I\}$ be a disjoint indexed family of galaxies.
\begin{lemma}\label{lemma:alpha:reduction:fol}
%
%
If $\{{\mathcal F}_{A_{i}^{\prime},B_{i}^{\prime}}^{\rho_{i}^{\prime}}:\ i{\in}I\}$ is a $q$-$\alpha$-reduction of $\{{\mathcal F}_{A_{i},B_{i}}^{\rho_{i}}:\ i{\in}I\}$, $(W,R)$ is the union of $\{{\mathcal F}_{A_{i},B_{i}}^{\rho_{i}}:\ i{\in}I\}$ and $(W^{\prime},R^{\prime})$ is the union of $\{{\mathcal F}_{A_{i}^{\prime},B_{i}^{\prime}}^{\rho_{i}^{\prime}}:\ i{\in}I\}$ then $(W,R){\equiv_{q}}(W^{\prime},R^{\prime})$.
\end{lemma}
\begin{proof}
By Proposition~\ref{lemma:good:properties:of:alpha:reductions:indexed} and Lemma~\ref{ehrenfeucht:theorem}.
\medskip
\end{proof}
%
%
%
%
%
%
%
%
\begin{lemma}\label{lemma:gamma:reduction:fol}
%
%
If $\{{\mathcal F}_{A_{i}^{\prime},B_{i}^{\prime}}^{\rho_{i}^{\prime}}:\ i{\in}I\}$ is a $q$-$\gamma$-reduction of $\{{\mathcal F}_{A_{i},B_{i}}^{\rho_{i}}:\ i{\in}I\}$, $(W,R)$ is the union of $\{{\mathcal F}_{A_{i},B_{i}}^{\rho_{i}}:\ i{\in}I\}$ and $(W^{\prime},R^{\prime})$ is the union of $\{{\mathcal F}_{A_{i}^{\prime},B_{i}^{\prime}}^{\rho_{i}^{\prime}}:\ i{\in}I\}$ then $(W,R){\equiv_{q}}(W^{\prime},R^{\prime})$.
\end{lemma}
\begin{proof}
By Proposition~\ref{lemma:good:properties:of:gamma:reductions:indexed} and Lemma~\ref{ehrenfeucht:theorem}.
\medskip
\end{proof}
\begin{lemma}\label{lemma:delta:reduction:fol}
%
%
If $\{{\mathcal F}_{A_{j},B_{j}}^{\rho_{j}}:\ j{\in}J\}$ is a $q$-$\delta$-reduction of $\{{\mathcal F}_{A_{i},B_{i}}^{\rho_{i}}:\ i{\in}I\}$, $(W,R)$ is the union of $\{{\mathcal F}_{A_{i},B_{i}}^{\rho_{i}}:\ i{\in}I\}$ and $(W^{\prime},R^{\prime})$ is the union of $\{{\mathcal F}_{A_{j},B_{j}}^{\rho_{j}}:\ j{\in}J\}$ then $(W,R){\equiv_{q}}(W^{\prime},R^{\prime})$.
\end{lemma}
\begin{proof}
By Proposition~\ref{lemma:good:properties:of:delta:reductions:indexed} and Lemma~\ref{ehrenfeucht:theorem}.
\medskip
\end{proof}
\paragraph{Rooted translation}
The {\it rooted translation of a first-order formula $A$ with respect to an individual variable $\mathbf{x}$ not occurring in $A$}\/ (denoted $\tau(\mathbf{x},A)$) is inductively defined as follows:
\begin{itemize}
\item $\tau(\mathbf{x},\mathbf{R}(\mathbf{y},\mathbf{z}))$ is $\mathbf{R}(\mathbf{y},\mathbf{z})$,
\item $\tau(\mathbf{x},\mathbf{y}=\mathbf{z})$ is $\mathbf{y}=\mathbf{z}$,
\item $\tau(\mathbf{x},\neg A)$ is $\neg\tau(\mathbf{x},A)$,
\item $\tau(\mathbf{x},A\vee B)$ is $\tau(\mathbf{x},A)\vee\tau(\mathbf{x},B)$,
%
%
%
%
\item $\tau(\mathbf{x},\forall\mathbf{y}A)$ is $\forall\mathbf{y}(\mathbf{x}=\mathbf{y}\vee\exists\mathbf{z}(\mathbf{R}(\mathbf{x},\mathbf{z})\wedge\mathbf{R}(\mathbf{z},\mathbf{y}))\rightarrow\tau(\mathbf{x},A))$.
\end{itemize}
%
%
%
%
%
%
%
%
%
%
%
%
%
%
%
%
\begin{lemma}\label{lemma:about:the:rooted:translation}
Let $(W,R)$ be an Euclidean frame.
For all sentences $A$ and for all individual variables $\mathbf{x}$ not occurring in $A$, the following conditions are equivalent:
\begin{itemize}
\item $(W,R){\models}\forall\mathbf{x}\tau(\mathbf{x},A)$,
\item for all $s{\in}W$, $(W_{s},R_{s}){\models}A$.
\end{itemize}
\end{lemma}
\section{Definability problems}\label{section:definability:problems}
In this section, we introduce $3$~definability problems: first-order definability problem, modal definability problem and correspondence problem.
We remind what is known concerning their computability.
%
%
%
%
\\
\\
%
%
A modal formula $\varphi$ is {\it first-order definable in a class ${\mathcal C}$ of frames}\/ if there exists a sentence $A$ such that for all frames ${\mathcal F}$ in ${\mathcal C}$, ${\mathcal F}{\models}\varphi$ if and only if ${\mathcal F}{\models}A$.
As is well-known, in the class of all frames, $\Box\Diamond p\rightarrow\Diamond\Box p$ is not first-order definable~\cite[Chapter~$3$]{Blackburn:deRijke:Venema:2001}.
The {\it first-order definability problem in a class ${\mathcal C}$ of frames}\/ is the following decision problem: determine whether a given modal formula is first-order definable in ${\mathcal C}$.
\\
\\
%
%
%
%
%
%
%
%
A sentence $A$ is {\it modally definable in a class ${\mathcal C}$ of frames}\/ if there exists a modal formula $\varphi$ such that for all frames ${\mathcal F}$ in ${\mathcal C}$, ${\mathcal F}{\models}\varphi$ if and only if ${\mathcal F}{\models}A$.
As is well-known, in the class of all frames, $\exists\mathbf{x}\mathbf{R}(\mathbf{x},\mathbf{x})$ is not modally definable~\cite[Chapter~$3$]{Blackburn:deRijke:Venema:2001}.
The {\it modal definability problem in a class ${\mathcal C}$ of frames}\/ is the following decision problem: determine whether a given sentence is modally definable in ${\mathcal C}$.
%
%
%
%
%
%
%
%
\\
\\
%
%
A modal formula $\varphi$ and a sentence $A$ {\it correspond in a class ${\mathcal C}$ of frames}\/ if for all frames ${\mathcal F}$ in ${\mathcal C}$, ${\mathcal F}{\models}\varphi$ if and only if ${\mathcal F}{\models}A$.
As is well-known, in the class of all frames, $\Diamond p\rightarrow\Box\Diamond p$ corresponds with $\forall\mathbf{x}\forall\mathbf{y}\forall\mathbf{z}(\mathbf{R}(\mathbf{x},\mathbf{y})\wedge\mathbf{R}(\mathbf{x},\mathbf{z})\rightarrow\mathbf{R}(\mathbf{y},\mathbf{z})\wedge\mathbf{R}(\mathbf{z},\mathbf{y}))$~\cite[Chapter~$3$]{Blackburn:deRijke:Venema:2001}.
The {\it correspondence problem in a class ${\mathcal C}$ of frames}\/ is the following decision problem: determine whether a given modal formula and a given sentence correspond in ${\mathcal C}$.
%
%
%
%
\\
\\
Owing to the significance of Correspondence Theory~\cite{vanBenthem:1983,vanBenthem:1984}, it is natural to ask whether the first-order definability problem, the modal definability problem and the correspondence problem are decidable.
%
%
%
%
The answers to these questions have been firstly obtained by Chagrova in her doctoral thesis~\cite{Chagrova:1989} and then further developed in~\cite{Chagrov:Chagrova:1995,Chagrov:Chagrova:2006,Chagrov:Chagrova:2007,Chagrova:1991}: the first-order definability problem, the modal definability problem and the correspondence problem are undecidable.
%
%
%
%
%
%
%
%
%
%
%
%
\\
\\
Chagrova's results concern the class of all frames.
They have been obtained by reductions from accessibility problems in Minsky machines~\cite{Chagrov:Chagrova:2006}.
Although Chagrov and Chagrova~\cite{Chagrov:Chagrova:1995} have also considered the first-order definability problem with respect to the class of all irreflexive transitive frames, it is not obvious how similar results can be obtained with respect to restricted classes of frames such as the class of all transitive frames, the class of all symmetric frames, etc.
\\
\\
With respect to the class of frames determined by $\S5$, it follows from $\S5$~Reduction Theorem in~\cite[Page~$51$]{Hughes:Cresswell:1968} and Lemma~$9.7$ in~\cite[Page~$99$]{vanBenthem:1983} that every modal formula is first-order definable.
Moreover, as proved in~\cite{Balbiani:Tinchev:2005,Balbiani:Tinchev:2006}, modal definability with respect to the class of frames determined by $\S5$ is $\PSPACE$-complete.
%
%
%
%
With respect to the classes of frames determined by $\KD45$ and $\K45$, similar results hold as well~\cite{Georgiev:2017a,Georgiev:2017b}.
\\
\\
Therefore, a first question naturally comes: is there an extension of $\K5$ such that the class of frames it determines gives rise both to a trivial first-order definability problem and an undecidable modal definability problem?
This question has been positively answered: with respect to the class of frames determined by $\K5$, first-order definability is trivial, whereas modal definability is undecidable~\cite{Balbiani:Georgiev:Tinchev:2018}.
\\
\\
The truth is that the classes of frames determined by all extensions of $\K5$ give rise to a trivial first-order definability problem.
Therefore, a second question naturally comes: characterize those extensions of $\K5$ such that the classes of frames they determine give rise to an undecidable modal definability problem.
Our answer will be given in Theorem~\ref{proposition:finale}: for all Euclidean modal logics $\L$, the problem of deciding the modal definability of sentences with respect to $\Fr(\L)$ is undecidable when $\mathtt{S}_{\L}{\setminus}((\{1\}{\times}\mathbf{N}^{-}){\cup}(\mathbf{N}^{+}{\times}\{{-}1,0,1\}))$ is infinite and decidable in $\EXPSPACE$ otherwise.
%
%
%
%
\\
\\
The following result joins up the modal definability character of a sentence and the behaviour of this sentence with respect to generated subframes and bounded morphic images.
We will need it later in our proofs about modal definability.
\begin{proposition}\label{proposition:characterizing:modal:definability}
Let $\L$ be an Euclidean modal logic.
For all sentences $A$, $A$ is modally definable in $\Fr(\L)$ if and only if
\begin{itemize}
\item for all Euclidean frames $(W,R)$, if $(W,R){\models}\L$ then $(W,R){\models}A$ if and only if for all $s{\in}W$, $(W_{s},R_{s}){\models}A$,
\item for all flowers ${\mathcal F}_{m}^{n},{\mathcal F}_{m^{\prime}}^{n^{\prime}}$, if ${\mathcal F}_{m}^{n}{\models}\L$, ${\mathcal F}_{m}^{n}{\models}A$ and ${\mathcal F}_{m^{\prime}}^{n^{\prime}}$ is a bounded morphic image of ${\mathcal F}_{m}^{n}$ then ${\mathcal F}_{m^{\prime}}^{n^{\prime}}{\models}A$.
\end{itemize}
\end{proposition}
\begin{proof}
Let $A$ be a sentence.
\\
\\
$(\Rightarrow)$
By Lemmas~\ref{Generated:Subframe:Lemma} and~\ref{Bounded:Morphism:Lemma}.
\\
\\
$(\Leftarrow)$
Suppose
\begin{itemize}
\item for all Euclidean frames $(W,R)$, if $(W,R){\models}\L$ then $(W,R){\models}A$ if and only if for all $s{\in}W$, $(W_{s},R_{s}){\models}A$,
\item for all flowers ${\mathcal F}_{m}^{n},{\mathcal F}_{m^{\prime}}^{n^{\prime}}$, if ${\mathcal F}_{m}^{n}{\models}\L$, ${\mathcal F}_{m}^{n}{\models}A$ and ${\mathcal F}_{m^{\prime}}^{n^{\prime}}$ is a bounded morphic image of ${\mathcal F}_{m}^{n}$ then ${\mathcal F}_{m^{\prime}}^{n^{\prime}}{\models}A$.
\end{itemize}
We consider the following $2$~cases.
\\
\\
{\bf Case where $\Fr(\L){\models}A$:}
Hence, the modal formula $\top$ and the sentence $A$ correspond in $\Fr(\L)$.
Thus, $A$ is modally definable in $\Fr(\L)$.
\\
\\
{\bf Case where $\Fr(\L){\not\models}A$:}
Consequently, let $(W^{\prime},R^{\prime})$ be an Euclidean frame such that $(W^{\prime},R^{\prime}){\models}\L$ and $(W^{\prime},R^{\prime}){\not\models}A$.
Since for all Euclidean frames $(W,R)$, if $(W,R){\models}\L$ then $(W,R){\models}A$ if and only if for all $s{\in}W$, $(W_{s},R_{s}){\models}A$, then let $s^{\prime}{\in}W^{\prime}$ be such that $(W^{\prime}_{s^{\prime}},R^{\prime}_{s^{\prime}}){\not\models}A$.
Since $(W^{\prime},R^{\prime}){\models}\L$, then by Lemma~\ref{Generated:Subframe:Lemma}, $(W^{\prime}_{s^{\prime}},R^{\prime}_{s^{\prime}}){\models}\L$.
Moreover, by Lemma~\ref{lemma:generated:subframe:from:a:single:point:has:specific:shapes}, either $R^{\prime}_{s^{\prime}}{=}W^{\prime}_{s^{\prime}}{\times}W^{\prime}_{s^{\prime}}$, or $W^{\prime}_{s^{\prime}}{=}\{s^{\prime}\}{\cup}C^{\prime}$ and $R^{\prime}_{s^{\prime}}{=}(\{s^{\prime}\}{\times}B^{\prime}){\cup}(C^{\prime}{\times}C^{\prime})$ for some sets $B^{\prime},C^{\prime}$ such that either $B^{\prime}{\not=}\emptyset$, or $C^{\prime}{=}\emptyset$, $B^{\prime}{\subseteq}C^{\prime}$ and $s^{\prime}{\not\in}C^{\prime}$.
Let ${\mathcal C}$ be the class of all rooted Euclidean frames $(W,R)$ such that $(W,R){\models}\L$ and $(W,R){\not\models}A$.
Since $(W^{\prime}_{s^{\prime}},R^{\prime}_{s^{\prime}}){\not\models}A$ and $(W^{\prime}_{s^{\prime}},R^{\prime}_{s^{\prime}}){\models}\L$, then $(W^{\prime}_{s^{\prime}},R^{\prime}_{s^{\prime}})$ is in ${\mathcal C}$.
Let $q{=}\qdd(A)$.
Let
\begin{itemize}
\item $\mathtt{S}_{\L}^{A}{=}\{(m,n):\ m{\in}\lsem0,q\rsem$, $n{\in}\lsem{-}1,q\rsem$, either $m{\not=}0$, or $n{=}0$, ${\mathcal F}_{m}^{n}{\models}\L$ and ${\mathcal F}_{m}^{n}{\not\models}A\}$.
\end{itemize}
Obviously, $\mathtt{S}_{\L}^{A}{\subseteq}\mathtt{S}_{\L}{\cup}\{(0,0\}$.
\\
\\
We claim that $\mathtt{S}_{\L}^{A}$ is non-empty.
We consider the following $2$~cases.
\\
\\
{\bf Case where $R^{\prime}_{s^{\prime}}{=}W^{\prime}_{s^{\prime}}{\times}W^{\prime}_{s^{\prime}}$:}
Let $(W^{\prime\prime},R^{\prime\prime})$ be an Euclidean frame obtained from $(W^{\prime}_{s^{\prime}},R^{\prime}_{s^{\prime}})$ as follows:
\begin{itemize}
\item if ${\parallel}W^{\prime}_{s^{\prime}}{\parallel}{\leq}q$ then let $W^{\prime\prime}$ be $W^{\prime}_{s^{\prime}}$ else let $W^{\prime\prime}$ be an arbitrary subset of $W^{\prime}_{s^{\prime}}$ of cardinality $q$,
\item let $R^{\prime\prime}$ be $W^{\prime\prime}{\times}W^{\prime\prime}$.
\end{itemize}
Obviously, either ${\parallel}W^{\prime}_{s^{\prime}}{\parallel}{\leq}q$ and ${\parallel}W^{\prime\prime}{\parallel}{=}{\parallel}W^{\prime}_{s^{\prime}}{\parallel}$, or ${\parallel}W^{\prime}_{s^{\prime}}{\parallel}{>}q$ and ${\parallel}W^{\prime\prime}{\parallel}{=}q$.
Hen\-ce, by Lemma~\ref{lemma:about:partitions:q:greater:smaller}, $(W^{\prime\prime},R^{\prime\prime})$ is a bounded morphic image of $(W^{\prime}_{s^{\prime}},R^{\prime}_{s^{\prime}})$.
Moreover, by Lemma~\ref{lemma:about:partitions:q:greater:smaller:games}, the second player has a winning strategy in the Ehrenfeucht-Fra\"\i ss\'e game ${\mathcal G}_{q}((W^{\prime}_{s^{\prime}},R^{\prime}_{s^{\prime}}),(W^{\prime\prime},R^{\prime\prime}))$.
Thus, by Lemmas~\ref{Bounded:Morphism:Lemma} and~\ref{ehrenfeucht:theorem}, $(W^{\prime}_{s^{\prime}},R^{\prime}_{s^{\prime}})
$\linebreak$
{\preceq}(W^{\prime\prime},R^{\prime\prime})$ and $(W^{\prime\prime},R^{\prime\prime}){\equiv_{q}}(W^{\prime}_{s^{\prime}},R^{\prime}_{s^{\prime}})$.
Since $(W^{\prime}_{s^{\prime}},R^{\prime}_{s^{\prime}}){\not\models}A$ and $(W^{\prime}_{s^{\prime}},R^{\prime}_{s^{\prime}}){\models}\L$, then $(W^{\prime\prime},R^{\prime\prime}){\not\models}A$ and $(W^{\prime\prime},R^{\prime\prime}){\models}\L$.
Let $m{=}{\parallel}W^{\prime\prime}{\parallel}$ and $n{=}{-}1$.
Obviously, $m{\in}\lsem0,q\rsem$, $n{\in}\lsem{-}1,q\rsem$ and either $m{\not=}0$, or $n{=}0$.
Moreover, $(W^{\prime\prime},R^{\prime\prime})$ is isomorphic with ${\mathcal F}_{m}^{n}$.
Since $(W^{\prime\prime},R^{\prime\prime}){\not\models}A$ and $(W^{\prime\prime},R^{\prime\prime}){\models}\L$, then ${\mathcal F}_{m}^{n}{\not\models}A$ and ${\mathcal F}_{m}^{n}{\models}\L$.
Since $m{\in}\lsem0,q\rsem$, $n{\in}\lsem{-}1,q\rsem$ and either $m{\not=}0$, or $n{=}0$, then $(m,n)$ is in $\mathtt{S}_{\L}^{A}$.
%
%
\\
\\
{\bf Case where $W^{\prime}_{s^{\prime}}{=}\{s^{\prime}\}{\cup}C^{\prime}$ and $R^{\prime}_{s^{\prime}}{=}(\{s^{\prime}\}{\times}B^{\prime}){\cup}(C^{\prime}{\times}C^{\prime})$ for some sets $B^{\prime},
$\linebreak$
C^{\prime}$ such that either $B^{\prime}{\not=}\emptyset$, or $C^{\prime}{=}\emptyset$, $B^{\prime}{\subseteq}C^{\prime}$ and $s^{\prime}{\not\in}C^{\prime}$:}
Let $(W^{\prime\prime},R^{\prime\prime})$ be an Euclidean frame obtained from $(W^{\prime}_{s^{\prime}},R^{\prime}_{s^{\prime}})$ as follows:
\begin{itemize}
\item if ${\parallel}B^{\prime}{\parallel}{\leq}q$ then let $B^{\prime\prime}$ be $B^{\prime}$ else let $B^{\prime\prime}$ be an arbitrary subset of $B^{\prime}$ of cardinality $q$,
\item if ${\parallel}C^{\prime}{\setminus}B^{\prime}{\parallel}{\leq}q$ then let $D^{\prime\prime}$ be $C^{\prime}{\setminus}B^{\prime}$ else let $D^{\prime\prime}$ be an arbitrary subset of $C^{\prime}{\setminus}B^{\prime}$ of cardinality $q$,
\item let $C^{\prime\prime}$ be $B^{\prime\prime}{\cup}D^{\prime\prime}$,
\item let $W^{\prime\prime}$ be $\{s^{\prime}\}{\cup}C^{\prime\prime}$,
\item let $R^{\prime\prime}$ be $(\{s^{\prime}\}{\times}B^{\prime\prime}){\cup}(C^{\prime\prime}{\times}C^{\prime\prime})$.
\end{itemize}
Obviously, either ${\parallel}B^{\prime}{\parallel}{\leq}q$ and ${\parallel}B^{\prime}{\parallel}{=}{\parallel}B^{\prime\prime}{\parallel}$, or ${\parallel}B^{\prime}{\parallel}{>}q$ and ${\parallel}B^{\prime\prime}{\parallel}{=}q$ and either ${\parallel}C^{\prime}{\setminus}B^{\prime}{\parallel}{\leq}q$ and ${\parallel}C^{\prime}{\setminus}B^{\prime}{\parallel}{=}{\parallel}C^{\prime\prime}{\setminus}B^{\prime\prime}{\parallel}$, or ${\parallel}C^{\prime}{\setminus}B^{\prime}{\parallel}{>}q$ and ${\parallel}C^{\prime\prime}{\setminus}B^{\prime\prime}{\parallel}{=}q$.
Consequently, by Lemma~\ref{lemma:about:frames:C:B:smaller:greater}, $(W^{\prime\prime},R^{\prime\prime})$ is a bounded morphic image of $(W^{\prime}_{s^{\prime}},R^{\prime}_{s^{\prime}})$.
Moreover, by Lemma~\ref{lemma:about:frames:C:B:smaller:greater:games}, the second player has a winning strategy in the Ehrenfeucht-Fra\"\i ss\'e game ${\mathcal G}_{q}((W^{\prime}_{s^{\prime}},R^{\prime}_{s^{\prime}}),(W^{\prime\prime},R^{\prime\prime}))$.
Hence, by Lemmas~\ref{Bounded:Morphism:Lemma}
\linebreak
and~\ref{ehrenfeucht:theorem}, $(W^{\prime}_{s^{\prime}},R^{\prime}_{s^{\prime}}){\preceq}(W^{\prime\prime},R^{\prime\prime})$ and $(W^{\prime\prime},R^{\prime\prime}){\equiv_{q}}(W^{\prime}_{s^{\prime}},R^{\prime}_{s^{\prime}})$.
Since $(W^{\prime}_{s^{\prime}},R^{\prime}_{s^{\prime}}){\not\models}A$ and $(W^{\prime}_{s^{\prime}},R^{\prime}_{s^{\prime}}){\models}\L$, then $(W^{\prime\prime},R^{\prime\prime}){\not\models}A$ and $(W^{\prime\prime},R^{\prime\prime}){\models}\L$.
Let $m{=}{\parallel}B^{\prime\prime}{\parallel}$ and $n{=}{\parallel}D^{\prime\prime}{\parallel}$.
Obviously, $m{\in}\lsem0,q\rsem$, $n{\in}\lsem{-}1,q\rsem$ and either $m{\not=}0$, or $n{=}0$.
Moreover, $(W^{\prime\prime},R^{\prime\prime})$ is isomorphic with ${\mathcal F}_{m}^{n}$.
Since $(W^{\prime\prime},R^{\prime\prime}){\not\models}A$ and $(W^{\prime\prime},R^{\prime\prime}){\models}\L$, then ${\mathcal F}_{m}^{n}{\not\models}A$ and ${\mathcal F}_{m}^{n}{\models}\L$.
Since $m{\in}\lsem0,q\rsem$, $n{\in}\lsem{-}1,q\rsem$ and either $m{\not=}0$, or $n{=}0$, then $(m,n)$ is in $\mathtt{S}_{\L}^{A}$.
%
%
\\
\\
Thus, $\mathtt{S}_{\L}^{A}$ is non-empty.
\\
\\
We claim that $\mathtt{S}_{\L}^{A}$ is finite.
Obviously, $\mathtt{S}_{\L}^{A}{\subseteq}\lsem0,q\rsem{\times}\lsem{-}1,q\rsem$.
Consequently, $\mathtt{S}_{\L}^{A}$ is finite.
\\
\\
Let
\begin{itemize}
\item $\varphi_{\L}^{A}{=}\bigwedge\{\chi_{m}^{n}:\ (m,n){\in}\mathtt{S}_{\L}^{A}\}$.
\end{itemize}
%
%
%
%
We claim that the modal formula $\varphi_{\L}^{A}$ and the sentence $A$ correspond in $\Fr(\L)$.
If not, let $(W^{\prime},R^{\prime})$ be an Euclidean frame such that $(W^{\prime},R^{\prime}){\models}\L$ and either $(W^{\prime},R^{\prime}){\models}\varphi_{\L}^{A}$ and $(W^{\prime},R^{\prime}){\not\models}A$, or $(W^{\prime},R^{\prime}){\not\models}\varphi_{\L}^{A}$ and $(W^{\prime},R^{\prime}){\models}A$.
In the former case, since for all Euclidean frames $(W,R)$, if $(W,R){\models}\L$ then $(W,R){\models}A$ if and only if for all $s{\in}W$, $(W_{s},R_{s}){\models}A$, then let $s^{\prime}{\in}W^{\prime}$ be such that $(W^{\prime}_{s^{\prime}},R^{\prime}_{s^{\prime}}){\not\models}A$.
Since $(W^{\prime},R^{\prime}){\models}\L$ and $(W^{\prime},R^{\prime}){\models}\varphi_{\L}^{A}$, then by Lemma~\ref{Generated:Subframe:Lemma}, $(W^{\prime}_{s^{\prime}},R^{\prime}_{s^{\prime}}){\models}\L$ and $(W^{\prime}_{s^{\prime}},R^{\prime}_{s^{\prime}}){\models}\varphi_{\L}^{A}$.
Since $(W^{\prime}_{s^{\prime}},R^{\prime}_{s^{\prime}}){\not\models}A$, then let $(W^{\prime\prime},R^{\prime\prime})$ be the Euclidean frame and $(m,n)$ be the couple of integers obtained from $(W^{\prime}_{s^{\prime}},R^{\prime}_{s^{\prime}})$ as above in the argument showing that $\mathtt{S}_{\L}^{A}$ is non-empty.
Remind from this argument that $(W^{\prime}_{s^{\prime}},R^{\prime}_{s^{\prime}}){\preceq}(W^{\prime\prime},R^{\prime\prime})$ and $(W^{\prime\prime},R^{\prime\prime}){\equiv_{q}}(W^{\prime}_{s^{\prime}},R^{\prime}_{s^{\prime}})$.
Moreover, $m{\in}\lsem0,q\rsem$, $n{\in}\lsem{-}1,q\rsem$ and either $m{\not=}0$, or $n{=}0$.
In other respect, $(W^{\prime\prime},R^{\prime\prime})$ is isomorphic with ${\mathcal F}_{m}^{n}$.
Since $(W^{\prime}_{s^{\prime}},R^{\prime}_{s^{\prime}}){\not\models}A$, $(W^{\prime}_{s^{\prime}},R^{\prime}_{s^{\prime}}){\models}\L$ and $(W^{\prime}_{s^{\prime}},R^{\prime}_{s^{\prime}}){\models}\varphi_{\L}^{A}$, then ${\mathcal F}_{m}^{n}{\not\models}A$, ${\mathcal F}_{m}^{n}{\models}\L$ and ${\mathcal F}_{m}^{n}{\models}\varphi_{\L}^{A}$.
Since $m{\in}\lsem0,q\rsem$, $n{\in}\lsem{-}1,q\rsem$ and either $m{\not=}0$, or $n{=}0$, then $(m,n){\in}\mathtt{S}_{\L}^{A}$.
Since ${\mathcal F}_{m}^{n}{\models}\varphi_{\L}^{A}$, then ${\mathcal F}_{m}^{n}{\models}\chi_{m}^{n}$: a contradiction with Lemma~\ref{lemma:Jankov:Fine:TWO}.
In the latter case, there exists a valuation $V^{\prime}$ on $(W^{\prime},R^{\prime})$ and there exists $s^{\prime}{\in}W^{\prime}$ such that $(W^{\prime},R^{\prime}),V^{\prime},s^{\prime}{\not\models}\varphi_{\L}^{A}$.
Moreover, since for all Euclidean frames $(W,R)$, if $(W,R){\models}\L$ then $(W,R){\models}A$ if and only if for all $s{\in}W$, $(W_{s},R_{s}){\models}A$, then $(W^{\prime}_{s^{\prime}},R^{\prime}_{s^{\prime}}){\models}A$.
Hence, there exists $(m,n){\in}\mathtt{S}_{\L}^{A}$ such that $(W^{\prime},R^{\prime}),V^{\prime},s^{\prime}{\not\models}\chi_{m}^{n}$.
Thus, by Lemma~\ref{lemma:Jankov:Fine:ONE}, there exists a surjective bounded morphism $f$ from $(W^{\prime}_{s^{\prime}},R^{\prime}_{s^{\prime}})$ to ${\mathcal F}_{m}^{n}$ such that $f(s^{\prime})$ generates ${\mathcal F}_{m}^{n}$.
We consider the following $3$~cases.
\\
\\
{\bf Case where $m{=}0$:}
Since either $m{\not=}0$, or $n{=}0$, then $n{=}0$.
Since $(m,n){\in}\mathtt{S}_{\L}^{A}$ and $m{=}0$, then ${\mathcal F}_{0}^{0}{\not\models}A$.
Moreover, since $(W^{\prime},R^{\prime}),V^{\prime},s^{\prime}{\not\models}\chi_{m}^{n}$ and $m{=}0$, then $(W^{\prime},R^{\prime}),V^{\prime},s^{\prime}{\models}\Box\bot$.
Consequently, $(W^{\prime}_{s^{\prime}},R^{\prime}_{s^{\prime}})$ is isomorphic with ${\mathcal F}_{0}^{0}$.
Since $(W^{\prime}_{s^{\prime}},R^{\prime}_{s^{\prime}}){\models}A$, then ${\mathcal F}_{0}^{0}{\models}A$: a contradiction.
\\
\\
{\bf Case where $m{\not=}0$ and $n{\in}\N$:}
Since $f$ is a surjective bounded morphism from $(W^{\prime}_{s^{\prime}},R^{\prime}_{s^{\prime}})$ to ${\mathcal F}_{m}^{n}$ such that $f(s^{\prime})$ generates ${\mathcal F}_{m}^{n}$, then by Lemma~\ref{lemma:about:what:happens:when:it:exists:bmi:flower:integer}, $W^{\prime}_{s^{\prime}}{=}\{s^{\prime}\}{\cup}C^{\prime}$ and $R^{\prime}_{s^{\prime}}{=}(\{s^{\prime}\}{\times}B^{\prime}){\cup}(C^{\prime}{\times}C^{\prime})$ for some sets $B^{\prime},C^{\prime}$ such that $B^{\prime}{\subseteq}
$\linebreak$
C^{\prime}$, $s^{\prime}{\not\in}C^{\prime}$, ${\parallel}B^{\prime}{\parallel}{\geq}m$ and ${\parallel}C^{\prime}{\setminus}B^{\prime}{\parallel}{\geq}n$.
Since $(m,n){\in}\mathtt{S}_{\L}^{A}$, then ${\mathcal F}_{m}^{n}{\not\models}A$, $m{\leq}q$ and $n{\leq}q$.
Let $m^{\prime}{=}\min\{{\parallel}B^{\prime}{\parallel},q\}$ and $n^{\prime}{=}\min\{{\parallel}C^{\prime}{\setminus}B^{\prime}{\parallel},q\}$.
Obviously, either ${\parallel}B^{\prime}{\parallel}{\leq}q$ and ${\parallel}B^{\prime}{\parallel}{=}m^{\prime}$, or ${\parallel}B^{\prime}{\parallel}{>}q$ and $m^{\prime}{=}q$ and either ${\parallel}C^{\prime}{\setminus}B^{\prime}{\parallel}{\leq}q$ and ${\parallel}C^{\prime}{\setminus}B^{\prime}{\parallel}
$\linebreak$
{=}n^{\prime}$, or ${\parallel}C^{\prime}{\setminus}B^{\prime}{\parallel}{>}q$ and $n^{\prime}{=}q$.
Hence, by Lemma~\ref{lemma:about:frames:C:B:smaller:greater:games}, the second player has a winning strategy in the Ehren\-feucht-Fra\"\i ss\'e game ${\mathcal G}_{q}((W_{s^{\prime}}^{\prime},R_{s^{\prime}}^{\prime}),{\mathcal F}_{m^{\prime}}^{n^{\prime}})$.
Thus, by Lemma~\ref{ehrenfeucht:theorem}, $(W_{s^{\prime}}^{\prime},R_{s^{\prime}}^{\prime}){\equiv_{q}}{\mathcal F}_{m^{\prime}}^{n^{\prime}}$.
Moreover, since ${\parallel}B^{\prime}{\parallel}{\geq}m$, ${\parallel}C^{\prime}{\setminus}B^{\prime}{\parallel}{\geq}n$, $m{\leq}q$ and $n{\leq}q$, then $m{\leq}m^{\prime}$ and $n{\leq}n^{\prime}$.
Consequently, $(m,n){\ll}(m^{\prime},n^{\prime})$.
Hence, by Lemma~\ref{lemma:about:bounded:morphisms:in:the:situation:of F:m:n:frames}, ${\mathcal F}_{m}^{n}$ is a bounded morphic image of ${\mathcal F}_{m^{\prime}}^{n^{\prime}}$.
Since $(W^{\prime}_{s^{\prime}},R^{\prime}_{s^{\prime}}){\models}A$ and $(W^{\prime}_{s^{\prime}},R^{\prime}_{s^{\prime}}){\equiv_{q}}{\mathcal F}_{m^{\prime}}^{n^{\prime}}$, then ${\mathcal F}_{m^{\prime}}^{n^{\prime}}{\models}A$.
Since ${\mathcal F}_{m}^{n}$ is a bounded morphic image of ${\mathcal F}_{m^{\prime}}^{n^{\prime}}$, then ${\mathcal F}_{m}^{n}{\models}A$: a contradiction.
\\
\\
{\bf Case where $m{\not=}0$ and $n{=}{-}1$:}
Since $f$ is a surjective bounded morphism from $(W^{\prime}_{s^{\prime}},R^{\prime}_{s^{\prime}})$ to ${\mathcal F}_{m}^{n}$ such that $f(s^{\prime})$ generates ${\mathcal F}_{m}^{n}$, then by Lemma~\ref{lemma:about:what:happens:when:it:exists:bmi:flower:minus:one}, either $R^{\prime}_{s^{\prime}}{=}W^{\prime}_{s^{\prime}}{\times}W^{\prime}_{s^{\prime}}$ and ${\parallel}W^{\prime}_{s^{\prime}}{\parallel}{\geq}m$, or $W^{\prime}_{s^{\prime}}{=}\{s^{\prime}\}{\cup}C^{\prime}$ and $R^{\prime}_{s^{\prime}}{=}(\{s^{\prime}\}{\times}B^{\prime}){\cup}(C^{\prime}{\times}C^{\prime})$ for some sets $B^{\prime},C^{\prime}$ such that $B^{\prime}{\subseteq}C^{\prime}$, $s^{\prime}{\not\in}C^{\prime}$ and ${\parallel}B^{\prime}{\parallel}{\geq}m$.
In the former case, since $(m,n){\in}\mathtt{S}_{\L}^{A}$ and $n{=}{-}1$, then ${\mathcal F}_{m}^{{-}1}{\not\models}A$ and $m{\leq}q$.
Let $m^{\prime}{=}\min\{{\parallel}W^{\prime}_{s^{\prime}}{\parallel},q\}$.
Obviously, either ${\parallel}W^{\prime}_{s^{\prime}}{\parallel}{\leq}q$ and ${\parallel}W^{\prime}_{s^{\prime}}{\parallel}{=}m^{\prime}$, or ${\parallel}W^{\prime}_{s^{\prime}}{\parallel}{>}q$ and $m^{\prime}{=}q$.
Thus, by Lemma~\ref{lemma:about:partitions:q:greater:smaller:games}, the second player has a winning strategy in the Ehren\-feucht-Fra\"\i ss\'e game ${\mathcal G}_{q}((W^{\prime}_{s^{\prime}},R^{\prime}_{s^{\prime}}),{\mathcal F}_{m^{\prime}}^{{-}1})$.
Consequently, by Lemma~\ref{ehrenfeucht:theorem}, $(W^{\prime}_{s^{\prime}},R^{\prime}_{s^{\prime}}){\equiv_{q}}{\mathcal F}_{m^{\prime}}^{{-}1}$.
Moreover, since ${\parallel}W^{\prime}_{s^{\prime}}{\parallel}{\geq}m$ and $m{\leq}q$, then $m{\leq}m^{\prime}$.
Hence, $(m,{-}1){\ll}(m^{\prime},{-}1)$.
Thus, by Lemma~\ref{lemma:about:bounded:morphisms:in:the:situation:of F:m:n:frames}, ${\mathcal F}_{m}^{{-}1}$ is a bounded morphic image of ${\mathcal F}_{m^{\prime}}^{{-}1}$.
Since $(W^{\prime}_{s^{\prime}},R^{\prime}_{s^{\prime}})
$\linebreak$
{\models}A$ and $(W^{\prime}_{s^{\prime}},R^{\prime}_{s^{\prime}}){\equiv_{q}}{\mathcal F}_{m^{\prime}}^{{-}1}$, then ${\mathcal F}_{m^{\prime}}^{{-}1}{\models}A$.
Since ${\mathcal F}_{m}^{{-}1}$ is a bounded morphic image of ${\mathcal F}_{m^{\prime}}^{{-}1}$, then ${\mathcal F}_{m}^{{-}1}{\models}A$: a contradiction.
In the latter case, since $(m,n){\in}\mathtt{S}_{\L}^{A}$ and $n{=}{-}1$, then ${\mathcal F}_{m}^{{-}1}{\not\models}A$ and $m{\leq}q$.
Let $m^{\prime}{=}\min\{{\parallel}C^{\prime}{\parallel},q\}$.
Obviously, either ${\parallel}C^{\prime}{\parallel}{\leq}q$ and ${\parallel}C^{\prime}{\parallel}{=}m^{\prime}$, or ${\parallel}C^{\prime}{\parallel}{>}q$ and $m^{\prime}{=}q$.
Consequently, by Lemma~\ref{lemma:about:frames:C:B:smaller:greater:games}, the second player has a winning strategy in the Ehren\-feucht-Fra\"\i ss\'e game ${\mathcal G}_{q}((C^{\prime},C^{\prime}
$\linebreak$
{\times}C^{\prime}),{\mathcal F}_{m^{\prime}}^{{-}1})$.
Hence, by Lemma~\ref{ehrenfeucht:theorem}, $(C^{\prime},C^{\prime}{\times}C^{\prime}){\equiv_{q}}{\mathcal F}_{m^{\prime}}^{{-}1}$.
Moreover, since $B^{\prime}{\subseteq}C^{\prime}$, ${\parallel}B^{\prime}{\parallel}{\geq}m$ and $m{\leq}q$, then $m{\leq}m^{\prime}$.
Thus, $(m,{-}1){\ll}(m^{\prime},{-}1)$.
Consequently, by Lemma~\ref{lemma:about:bounded:morphisms:in:the:situation:of F:m:n:frames}, ${\mathcal F}_{m}^{{-}1}$ is a bounded morphic image of ${\mathcal F}_{m^{\prime}}^{{-}1}$.
Obviously, $(C^{\prime},C^{\prime}{\times}C^{\prime})$ is a bounded morphic image of $(W^{\prime}_{s^{\prime}},R^{\prime}_{s^{\prime}})$.
Since $(W^{\prime}_{s^{\prime}},R^{\prime}_{s^{\prime}}){\models}A$ and $(C^{\prime},C^{\prime}{\times}C^{\prime}){\equiv_{q}}
$\linebreak$
{\mathcal F}_{m^{\prime}}^{{-}1}$, then ${\mathcal F}_{m^{\prime}}^{{-}1}{\models}A$.
Since ${\mathcal F}_{m}^{{-}1}$ is a bounded morphic image of ${\mathcal F}_{m^{\prime}}^{{-}1}$, then ${\mathcal F}_{m}^{{-}1}{\models}A$: a contradiction.
\\
\\
Hence, the modal formula $\varphi_{\L}^{A}$ and the sentence $A$ correspond in $\Fr(\L)$.
\\
\\
Thus, $A$ is modally definable in $\Fr(\L)$.
\medskip
\end{proof}
\section{Relative interpretations}\label{section:about:relative:interpretations}
In this section, we adapt to our purpose a theorem about relative interpretations.
See~\cite[Page~$272$]{Ershov:1980} for a more general presentation.
We will use it later for the proofs that the theories of the classes of frames determined by some Euclidean modal logics are undecidable.
%
%
\\
\\
Let ${\mathcal C}$ and ${\mathcal C}^{\prime}$ be classes of frames.
\\
\\
We say that ${\mathcal C}$ is {\it relatively elementary definable in ${\mathcal C}^{\prime}$}\/ if there exists first-order formulas $\mathbf{U}(\mathbf{x_{1}},\mathbf{x_{2}})$, $\mathbf{E}(\mathbf{x_{1}},\mathbf{x_{2}},\mathbf{x_{3}},\mathbf{x_{4}})$ and $\mathbf{A}(\mathbf{\mathbf{x_{1}},\mathbf{x_{2}},x_{3}},\mathbf{x_{4}})$ such that for all frames $(W,R)$ in ${\mathcal C}$, there exists a frame $(W^{\prime},R^{\prime})$ in ${\mathcal C}^{\prime}$ such that
\begin{itemize}
\item $X$ being the set of all couples $(s_{1},s_{2})$ such that $s_{1},s_{2}{\in}W^{\prime}$ are such that $(W^{\prime},R^{\prime}){\models}\mathbf{U}(\mathbf{x_{1}},\mathbf{x_{2}})\ \lbrack s_{1},s_{2}\rbrack$,
\item $\eta$ being the binary relation on $X$ such that for all $(s_{1},s_{2}),(t_{1},t_{2})$ in $X$, $(s_{1},s_{2})\eta(t_{1},t_{2})$ if and only if $(W^{\prime},R^{\prime}){\models}\mathbf{E}(\mathbf{x_{1}},\mathbf{x_{2}},\mathbf{y_{1}},\mathbf{y_{2}})\ \lbrack s_{1},s_{2},t_{1},t_{2}\rbrack$,
\item $S$ being the binary relation on $X$ such that for all $(s_{1},s_{2}),(t_{1},t_{2})$ in $X$, $(s_{1},s_{2})S(t_{1},t_{2})$ if and only if $(W^{\prime},R^{\prime}){\models}\mathbf{A}(\mathbf{x_{1}},\mathbf{x_{2}},\mathbf{y_{1}},\mathbf{y_{2}})\ \lbrack s_{1},s_{2},t_{1},t_{2}\rbrack$,
\end{itemize}
$X$ is non-empty, $\eta$ is a congruence on $(X,S)$ and $(X,S)/\eta$ is isomorphic to $(W,R)$.
\begin{proposition}[Relative elementary definability]\label{Syntactic:Restriction:Lemma:hereditarily:undecidable}
%
%
If ${\mathcal C}$ is relatively
\linebreak
elementary definable in ${\mathcal C}^{\prime}$ and $\Th({\mathcal C})$ is hereditarily undecidable then $\Th({\mathcal C}^{\prime})$ is hereditarily undecidable.
\end{proposition}
\begin{proof}
See~\cite[Page~$272$]{Ershov:1980}.
\medskip
\end{proof}
\begin{lemma}\label{lemma:relatively:elementary:definable:K2}
The class of all irreflexive symmetric frames with at least $2$~elements is relatively elementary definable in ${\mathcal K}_{2}$.
\end{lemma}
\begin{proof}
Let $(W,R)$ be an arbitrary irreflexive symmetric frame with at least $2$ elements.
Let $E$ be $\{\{s_{1},s_{2}\}:\ s_{1},s_{2}{\in}W$ and $s_{1}{R}s_{2}\}$.
Notice that for all $s_{1},s_{2}{\in}W$, if $\{s_{1},s_{2}\}{\in}E$ then $s_{1}{\not=}s_{2}$.
Without loss of generality, we can assume that $E$ and $W{\times}\{(3,4)\}$ are disjoint.
Let $A$ be $E{\cup}(W{\times}\{3,4\})$.
Let $B$ be $W{\times}\{1,2\}$.
Let ${\mathcal F}^{\rho}_{A,B}$ be the galaxy such that
\begin{itemize}
\item for all $\{s_{1},s_{2}\}{\in}E$, $\rho(\{s_{1},s_{2}\}){=}\{(s_{1},1),(s_{2},1)\}$,
\item for all $s{\in}W$, $\rho((s,3)){=}\{(s,1),(s,2)\}$,
\item for all $s{\in}W$, $\rho((s,4)){=}\{(s,1),(s,2)\}$.
\end{itemize}
%
%
Obviously, ${\mathcal F}^{\rho}_{A,B}$ is in ${\mathcal K}_{2}$.
\\
\\
Let $X$ be $\{(a_{1},a_{2}):\ a_{1},a_{2}{\in}W^{\rho}_{A,B}$ and ${\mathcal F}^{\rho}_{A,B}{\models}\mathbf{U}(\mathbf{x_{1}},\mathbf{x_{2}})\ \lbrack a_{1},a_{2}\rbrack\}$ where
\begin{itemize}
\item $\mathbf{U}(\mathbf{x_{1}},\mathbf{x_{2}})::=\mathbf{x_{1}}{\not=}\mathbf{x_{2}}\wedge\mathbf{R}(\mathbf{x_{1}},\mathbf{x_{2}})\wedge\mathbf{R}(\mathbf{x_{2}},\mathbf{x_{1}})\wedge\exists^{{=}2}\mathbf{y}(\mathbf{R}(\mathbf{y}){=}\{\mathbf{x_{1}},\mathbf{x_{2}}\}\wedge\neg\exists\mathbf{z}\mathbf{R}(\mathbf{z},\mathbf{y}))$.
\end{itemize}
Obviously, for all $s{\in}W$, $((s,1),(s,2)),((s,2),(s,1)){\in}X$.
Hence, $X$ is non-empty.
\\
\\
Obviously, for all $(a_{1},a_{2}){\in}X$, there exists $s{\in}W$ such that $\{a_{1},a_{2}\}{=}\{(s,1),(s,
$\linebreak$
2)\}$.
\\
\\
Let $\bowtie$ be the binary relation on $X$ such that for all $(a_{1},a_{2}),(b_{1},b_{2}){\in}X$, $(a_{1},a_{2}){\bowtie}
$\linebreak$
(b_{1},b_{2})$ if and only if ${\mathcal F}^{\rho}_{A,B}{\models}\mathbf{E}(\mathbf{x_{1}},\mathbf{x_{2}},\mathbf{y_{1}},\mathbf{y_{2}})\ \lbrack a_{1},a_{2},b_{1},b_{2}\rbrack$ where
\begin{itemize}
\item $\mathbf{E}(\mathbf{x_{1}},\mathbf{x_{2}},\mathbf{y_{1}},\mathbf{y_{2}})::=\mathbf{U}(\mathbf{x_{1}},\mathbf{x_{2}})\wedge\mathbf{U}(\mathbf{y_{1}},\mathbf{y_{2}})\wedge\{\mathbf{x_{1}},\mathbf{x_{2}}\}{=}\{\mathbf{y_{1}},\mathbf{y_{2}}\}$.
\end{itemize}
Obviously, $\bowtie$ is an equivalence relation on $X$.
For all $(a_{1},a_{2}){\in}X$, the equivalence class of $(a_{1},a_{2})$ modulo $\bowtie$ will be denoted $\lbrack(a_{1},a_{2})\rbrack$.
\\
\\
Since for all $(a_{1},a_{2}){\in}X$, there exists $s{\in}W$ such that $\{a_{1},a_{2}\}{=}\{(s,1),(s,2)\}$, then for all $(a_{1},a_{2}),(b_{1},b_{2}){\in}X$, $(a_{1},a_{2}){\bowtie}(b_{1},b_{2})$ if and only if there exists $s{\in}W$ such that $\{a_{1},a_{2}\}{=}\{(s,1),(s,2)\}$ and $\{b_{1},b_{2}\}{=}\{(s,1),(s,2)\}$.
\\
\\
Let $S$ be the binary relation on $X$ such that for all $(a_{1},a_{2}),(b_{1},b_{2}){\in}X$, $(a_{1},a_{2}){S}
$\linebreak$
(b_{1},b_{2})$ if and only if ${\mathcal F}^{\rho}_{A,B}{\models}\mathbf{A}(\mathbf{x_{1}},\mathbf{x_{2}},\mathbf{y_{1}},\mathbf{y_{2}})\ \lbrack a_{1},a_{2},b_{1},b_{2}\rbrack$ where
\begin{itemize}
\item $\mathbf{A}(\mathbf{x_{1}},\mathbf{x_{2}},\mathbf{y_{1}},\mathbf{y_{2}})::=\mathbf{U}(\mathbf{x_{1}},\mathbf{x_{2}})\wedge\mathbf{U}(\mathbf{y_{1}},\mathbf{y_{2}})\wedge\{\mathbf{x_{1}},\mathbf{x_{2}}\}{\cap}\{\mathbf{y_{1}},\mathbf{y_{2}}\}{=}\emptyset\wedge
$\linebreak$
\bigvee\{\exists^{{=}1}\mathbf{z}(\mathbf{R}(\mathbf{z}){=}\{\mathbf{x_{i}},\mathbf{y_{j}}\}\wedge\neg\exists\mathbf{t}\mathbf{R}(\mathbf{t},\mathbf{z})):\ 1{\leq}i,j{\leq}2\}$.
\end{itemize}
We claim that for all $(a_{1},a_{2}),(b_{1},b_{2}),(c_{1},c_{2}),(d_{1},d_{2}){\in}X$, if $(a_{1},a_{2}){\bowtie}(c_{1},c_{2})$ and $(b_{1},b_{2}){\bowtie}(d_{1},d_{2})$ then $(a_{1},a_{2}){S}(b_{1},b_{2})$ if and only if $(c_{1},c_{2}){S}(d_{1},d_{2})$.
If not, let $(a_{1},a_{2}),(b_{1},b_{2}),(c_{1},c_{2}),(d_{1},d_{2}){\in}X$ be such that $(a_{1},a_{2}){\bowtie}(c_{1},c_{2})$, $(b_{1},b_{2}){\bowtie}(d_{1},
$\linebreak$
d_{2})$ and either $(a_{1},a_{2}){S}(b_{1},b_{2})$ and not $(c_{1},c_{2}){S}(d_{1},d_{2})$, or not $(a_{1},a_{2}){S}(b_{1},b_{2})$ and $(c_{1},c_{2}){S}(d_{1},d_{2})$.
Without loss of generality, suppose $(a_{1},a_{2}){S}(b_{1},b_{2})$ and not $(c_{1},c_{2}){S}(d_{1},d_{2})$.
Thus, ${\mathcal F}^{\rho}_{A,B}{\models}\mathbf{A}(\mathbf{x_{1}},\mathbf{x_{2}},\mathbf{y_{1}},\mathbf{y_{2}})\ \lbrack a_{1},a_{2},b_{1},b_{2}\rbrack$ and ${\mathcal F}^{\rho}_{A,B}{\not\models}
$\linebreak$
\mathbf{A}(\mathbf{x_{1}},\mathbf{x_{2}},\mathbf{y_{1}},\mathbf{y_{2}})\ \lbrack c_{1},c_{2},d_{1},d_{2}\rbrack$.
Consequently,
\begin{description}
\item[$(C_{1})$] $\{a_{1},a_{2}\}{\cap}\{b_{1},b_{2}\}{=}\emptyset$,
\item[$(C_{2})$] there exists $i,j{\in}\{1,2\}$ such that for exactly $1$~$e$ in ${\mathcal F}_{A,B}^{\rho}$, $R_{A,B}^{\rho}(e){=}\{a_{i},b_{j}\}$ and ${R_{A,B}^{\rho}}^{{-}1}(e){=}\emptyset$,
\item[$(C_{3})$] at least one of the following conditions holds:
\begin{description}
\item[$(C_{3}^{a})$] $\{c_{1},c_{2}\}{\cap}\{d_{1},d_{2}\}{\not=}\emptyset$,
\item[$(C_{3}^{b})$] for all $k,l{\in}\{1,2\}$, either for all $f$ in ${\mathcal F}_{A,B}^{\rho}$, if $R_{A,B}^{\rho}(f){=}\{c_{k},d_{l}\}$ then ${R_{A,B}^{\rho}}^{{-}1}(f){\not=}\emptyset$, or there exists $f^{\prime},f^{\prime\prime}$ in ${\mathcal F}_{A,B}^{\rho}$ such that $f^{\prime}{\not=}f^{\prime\prime}$,
\linebreak$
R_{A,B}^{\rho}(f^{\prime}){=}\{c_{k},d_{l}\}$, $R_{A,B}^{\rho}(f^{\prime\prime}){=}\{c_{k},d_{l}\}$, ${R_{A,B}^{\rho}}^{{-}1}(f^{\prime}){=}\emptyset$ and
\linebreak$
{R_{A,B}^{\rho}}^{{-}1}(f^{\prime\prime}){=}\emptyset$.
\end{description}
\end{description}
%
%
Let $s,t,u,v{\in}W$ be such that $\{a_{1},a_{2}\}{=}\{(s,1),(s,2)\}$, $\{b_{1},b_{2}\}{=}\{(t,1),(t,2)\}$, $\{c_{1},c_{2}\}{=}\{(u,1),(u,2)\}$ and $\{d_{1},d_{2}\}{=}\{(v,1),(v,2)\}$.
Since $(a_{1},a_{2}){\bowtie}(c_{1},c_{2})$ and $(b_{1},b_{2}){\bowtie}(d_{1},d_{2})$, then ${\mathcal F}^{\rho}_{A,B}{\models}\mathbf{E}(\mathbf{x_{1}},\mathbf{x_{2}},\mathbf{y_{1}},\mathbf{y_{2}})\ \lbrack a_{1},a_{2},c_{1},c_{2}\rbrack$ and ${\mathcal F}^{\rho}_{A,B}{\models}\mathbf{E}(\mathbf{x_{1}},
$\linebreak$
\mathbf{x_{2}},\mathbf{y_{1}},\mathbf{y_{2}})\ \lbrack b_{1},b_{2},d_{1},d_{2}\rbrack$.
Hence, 
$\{a_{1},a_{2}\}{=}\{c_{1},c_{2}\}$ and $\{b_{1},b_{2}\}{=}\{d_{1},d_{2}\}$.
Since $\{a_{1},a_{2}\}{=}\{(s,1),(s,2)\}$, $\{b_{1},b_{2}\}{=}\{(t,1),(t,2)\}$, $\{c_{1},c_{2}\}{=}\{(u,1),(u,2)\}$ and $\{d_{1},
$\linebreak$
d_{2}\}{=}\{(v,1),(v,2)\}$, then $s{=}u$ and $t{=}v$.
Moreover, by $(C_{1})$, $s{\not=}t$.
In other respect, by $(C_{2})$, there exists $k^{\prime},l^{\prime}{\in}\{1,2\}$ such that $a_{i}{=}(s,k^{\prime})$, $b_{j}{=}(t,l^{\prime})$ and for exactly $1$~$e$ in ${\mathcal F}_{A,B}^{\rho}$, $R_{A,B}^{\rho}(e){=}\{(s,k^{\prime}),(t,l^{\prime})\}$ and ${R_{A,B}^{\rho}}^{{-}1}(e){=}\emptyset$.
Since $\{c_{1},c_{2}\}{=}\{(u,1),(u,2)\}$, $s{=}u$, $\{d_{1},d_{2}\}{=}\{(v,1),(v,2)\}$, $t{=}v$ and $s{\not=}t$, then $\{c_{1},
$\linebreak$
c_{2}\}{\cap}\{d_{1},d_{2}\}{=}\emptyset$.
Moreover, there exists $k,l{\in}\{1,2\}$ such that $c_{k}{=}(s,k^{\prime})$ and $d_{l}{=}(t,l^{\prime})$.
Hence, $(C_{3}^{a})$ does not hold.
Thus, by $(C_{3})$, $(C_{3}^{b})$ holds and either for all $f$ in ${\mathcal F}_{A,B}^{\rho}$, if $R_{A,B}^{\rho}(f){=}\{(s,k^{\prime}),(t,l^{\prime})\}$ then ${R_{A,B}^{\rho}}^{{-}1}(f){\not=}\emptyset$, or there exists $f^{\prime},f^{\prime\prime}$ in ${\mathcal F}_{A,B}^{\rho}$ such that $f^{\prime}{\not=}f^{\prime\prime}$, $R_{A,B}^{\rho}(f^{\prime}){=}\{(s,k^{\prime}),(t,l^{\prime})\}$, $R_{A,B}^{\rho}(f^{\prime\prime}){=}\{(s,k^{\prime}),(t,l^{\prime})\}$, ${R_{A,B}^{\rho}}^{{-}1}(f^{\prime}){=}\emptyset$ and ${R_{A,B}^{\rho}}^{{-}1}(f^{\prime\prime}){=}\emptyset$.
In the former case, since $R_{A,B}^{\rho}(e){=}\{(s,k^{\prime}),
$\linebreak$
(t,l^{\prime})\}$, then ${R_{A,B}^{\rho}}^{{-}1}(e){\not=}\emptyset$: a contradiction.
In the latter case, since $a_{i}{=}(s,k^{\prime})$, $b_{j}{=}(t,l^{\prime})$, then there exists $f^{\prime},f^{\prime\prime}$ in ${\mathcal F}_{A,B}^{\rho}$ such that $f^{\prime}{\not=}f^{\prime\prime}$, $R_{A,B}^{\rho}(f^{\prime}){=}\{a_{i},b_{j}\}$, $R_{A,B}^{\rho}(f^{\prime\prime}){=}\{a_{i},b_{j}\}$, ${R_{A,B}^{\rho}}^{{-}1}(f^{\prime}){=}\emptyset$ and ${R_{A,B}^{\rho}}^{{-}1}(f^{\prime\prime}){=}\emptyset$: a contradiction.
Consequently, for all $(a_{1},a_{2}),(b_{1},b_{2}),(c_{1},c_{2}),(d_{1},d_{2}){\in}X$, if $(a_{1},a_{2}){\bowtie}(c_{1},c_{2})$ and $(b_{1},b_{2}){\bowtie}(d_{1},d_{2})$ then $(a_{1},a_{2}){S}(b_{1},b_{2})$ if and only if $(c_{1},c_{2}){S}(d_{1},d_{2})$.
\\
\\
Let $(W^{\prime},R^{\prime})$ be the quotient of $(X,S)$ modulo $\bowtie$~---~i.e.
\begin{itemize}
\item $W^{\prime}$ is the quotient set of $X$ modulo $\bowtie$,
\item $R^{\prime}$ is the binary relation on $W^{\prime}$ such that for all $(a_{1},a_{2}),(b_{1},b_{2}){\in}X$, $\lbrack(a_{1},a_{2})\rbrack{R^{\prime}}\lbrack(b_{1},b_{2})\rbrack$ if and only if $(a_{1},a_{2}){S}(b_{1},b_{2})$.
\end{itemize}
Since for all $(a_{1},a_{2}){\in}X$, there exists $s{\in}W$ such that $\{a_{1},a_{2}\}{=}\{(s,1),(s,2)\}$ and for all $(a_{1},a_{2}),(b_{1},b_{2}){\in}X$, $(a_{1},a_{2}){\bowtie}(b_{1},b_{2})$ if and only if there exists $s{\in}W$ such that $\{a_{1},a_{2}\}{=}\{(s,1),(s,2)\}$ and $\{b_{1},b_{2}\}{=}\{(s,1),(s,2)\}$, then let $\kappa:\ W^{\prime}{\longrightarrow}W$ be the function such that for all $(a_{1},a_{2}){\in}X$, $\kappa(\lbrack(a_{1},a_{2})\rbrack)$ is the unique $s{\in}W$ such that $\{a_{1},a_{2}\}{=}\{(s,1),(s,2)\}$.
Obviously, $\kappa$ is bijective.
%
\\
\\
We claim that for all $(a_{1},a_{2}),(b_{1},b_{2}){\in}X$, $\lbrack(a_{1},a_{2})\rbrack{R^{\prime}}\lbrack(b_{1},b_{2})\rbrack$ if and only if $\kappa(\lbrack(a_{1},a_{2})\rbrack){R}\kappa(\lbrack(b_{1},b_{2})\rbrack)$.
Let $(a_{1},a_{2}),(b_{1},b_{2}){\in}X$.
Let $s$ be the unique element in $W$ such that $\{a_{1},a_{2}\}{=}\{(s,1),(s,2)\}$ and $t$ be the unique element in $W$ such that $\{b_{1},b_{2}\}{=}\{(t,1),(t,2)\}$.
Obviously, the following conditions are equivalent:
\begin{itemize}
\item $\lbrack(a_{1},a_{2})\rbrack{R^{\prime}}\lbrack(b_{1},b_{2})\rbrack$,
\item $(a_{1},a_{2}){S}(b_{1},b_{2})$,
\item ${\mathcal F}^{\rho}_{A,B}{\models}\mathbf{A}(\mathbf{x_{1}},\mathbf{x_{2}},\mathbf{y_{1}},\mathbf{y_{2}})\ \lbrack a_{1},a_{2},b_{1},b_{2}\rbrack$,
\item $s{\not=}t$ and there exists $i,j{\in}\{1,2\}$ such that for exactly $1$~$e$ in ${\mathcal F}_{A,B}^{\rho}$, $R_{A,B}^{\rho}(e){=}
$\linebreak$
\{a_{i},b_{j}\}$ and ${R_{A,B}^{\rho}}^{{-}1}(e){=}\emptyset$,
\item $\kappa(\lbrack(a_{1},a_{2})\rbrack){R}\kappa(\lbrack(b_{1},b_{2})\rbrack)$.
\end{itemize}
Hence, for all $(a_{1},a_{2}),(b_{1},b_{2}){\in}X$, $\lbrack(a_{1},a_{2})\rbrack{R^{\prime}}\lbrack(b_{1},b_{2})\rbrack$ if and only if $\kappa(\lbrack(a_{1},a_{2})\rbrack){R}
$\linebreak$
\kappa(\lbrack(b_{1},b_{2})\rbrack)$.
\\
\\
All in all, we have proved that, $(W^{\prime},R^{\prime})$ is isomorphic with $(W,R)$.
\\
\\
Thus, we have shown that the class of all irreflexive symmetric frames with at least $2$ elements is relatively elementary definable in ${\mathcal K}_{2}$.
\medskip
\end{proof}
\begin{lemma}\label{lemma:relatively:elementary:definable:L2}
The class of all irreflexive symmetric frames with at least $2$~elements is relatively elementary definable in ${\mathcal L}_{2}$.
\end{lemma}
\begin{proof}
Let $(W,R)$ be an arbitrary irreflexive symmetric frame with at least $2$ elements.
Let $E$ be $\{\{s_{1},s_{2}\}:\ s_{1},s_{2}{\in}W$ and $s_{1}{R}s_{2}\}$.
Notice that for all $s_{1},s_{2}{\in}W$, if $\{s_{1},s_{2}\}{\in}E$ then $s_{1}{\not=}s_{2}$.
Without loss of generality, we can assume that $E$ and $W{\times}\{(3,4)\}$ are disjoint.
Let $A$ be $E{\cup}(W{\times}\{3,4\})$.
Let $B$ be $W{\times}\{1,2\}$.
Let ${\mathcal F}^{\rho}_{A,B}$ be the galaxy such that
\begin{itemize}
\item for all $\{s_{1},s_{2}\}{\in}E$, $\rho(\{s_{1},s_{2}\}){=}B{\setminus}\{(s_{1},1),(s_{2},1)\}$,
\item for all $s{\in}W$, $\rho((s,3)){=}B{\setminus}\{(s,1),(s,2)\}$,
\item for all $s{\in}W$, $\rho((s,4)){=}B{\setminus}\{(s,1),(s,2)\}$.
\end{itemize}
%
%
Obviously, ${\mathcal F}^{\rho}_{A,B}$ is in ${\mathcal L}_{2}$.
\\
\\
Let $X$ be $\{(a_{1},a_{2}):\ a_{1},a_{2}{\in}W^{\rho}_{A,B}$ and ${\mathcal F}^{\rho}_{A,B}{\models}\mathbf{U}(\mathbf{x_{1}},\mathbf{x_{2}})\ \lbrack a_{1},a_{2}\rbrack\}$ where
\begin{itemize}
\item $\mathbf{U}(\mathbf{x_{1}},\mathbf{x_{2}})::=\mathbf{x_{1}}{\not=}\mathbf{x_{2}}\wedge\mathbf{R}(\mathbf{x_{1}},\mathbf{x_{2}})\wedge\mathbf{R}(\mathbf{x_{2}},\mathbf{x_{1}})\wedge\exists^{{=}2}\mathbf{y}(\mathbf{R}(\mathbf{y}){=}\{\mathbf{x_{1}},\mathbf{x_{2}}\}^{\mathbf{c}}\wedge\neg\exists\mathbf{z}\mathbf{R}(\mathbf{z},\mathbf{y}))$.
\end{itemize}
Obviously, for all $s{\in}W$, $((s,1),(s,2)),((s,2),(s,1)){\in}X$.
Hence, $X$ is non-empty.
\\
\\
Obviously, for all $(a_{1},a_{2}){\in}X$, there exists $s{\in}W$ such that $\{a_{1},a_{2}\}{=}\{(s,1),(s,
$\linebreak$
2)\}$.
\\
\\
Let $\bowtie$ be the binary relation on $X$ such that for all $(a_{1},a_{2}),(b_{1},b_{2}){\in}X$, $(a_{1},a_{2}){\bowtie}
$\linebreak$
(b_{1},b_{2})$ if and only if ${\mathcal F}^{\rho}_{A,B}{\models}\mathbf{E}(\mathbf{x_{1}},\mathbf{x_{2}},\mathbf{y_{1}},\mathbf{y_{2}})\ \lbrack a_{1},a_{2},b_{1},b_{2}\rbrack$ where
\begin{itemize}
\item $\mathbf{E}(\mathbf{x_{1}},\mathbf{x_{2}},\mathbf{y_{1}},\mathbf{y_{2}})::=\mathbf{U}(\mathbf{x_{1}},\mathbf{x_{2}})\wedge\mathbf{U}(\mathbf{y_{1}},\mathbf{y_{2}})\wedge\{\mathbf{x_{1}},\mathbf{x_{2}}\}{=}\{\mathbf{y_{1}},\mathbf{y_{2}}\}$.
\end{itemize}
Obviously, $\bowtie$ is an equivalence relation on $X$.
For all $(a_{1},a_{2}){\in}X$, the equivalence class of $(a_{1},a_{2})$ modulo $\bowtie$ will be denoted $\lbrack(a_{1},a_{2})\rbrack$.
\\
\\
Since for all $(a_{1},a_{2}){\in}X$, there exists $s{\in}W$ such that $\{a_{1},a_{2}\}{=}\{(s,1),(s,2)\}$, then for all $(a_{1},a_{2}),(b_{1},b_{2}){\in}X$, $(a_{1},a_{2}){\bowtie}(b_{1},b_{2})$ if and only if there exists $s{\in}W$ such that $\{a_{1},a_{2}\}{=}\{(s,1),(s,2)\}$ and $\{b_{1},b_{2}\}{=}\{(s,1),(s,2)\}$.
\\
\\
Let $S$ be the binary relation on $X$ such that for all $(a_{1},a_{2}),(b_{1},b_{2}){\in}X$, $(a_{1},a_{2}){S}
$\linebreak$
(b_{1},b_{2})$ if and only if ${\mathcal F}^{\rho}_{A,B}{\models}\mathbf{A}(\mathbf{x_{1}},\mathbf{x_{2}},\mathbf{y_{1}},\mathbf{y_{2}})\ \lbrack a_{1},a_{2},b_{1},b_{2}\rbrack$ where
\begin{itemize}
\item $\mathbf{A}(\mathbf{x_{1}},\mathbf{x_{2}},\mathbf{y_{1}},\mathbf{y_{2}})::=\mathbf{U}(\mathbf{x_{1}},\mathbf{x_{2}})\wedge\mathbf{U}(\mathbf{y_{1}},\mathbf{y_{2}})\wedge\{\mathbf{x_{1}},\mathbf{x_{2}}\}{\cap}\{\mathbf{y_{1}},\mathbf{y_{2}}\}{=}\emptyset\wedge
$\linebreak$
\bigvee\{\exists^{{=}1}\mathbf{z}(\mathbf{R}(\mathbf{z}){=}\{\mathbf{x_{i}},\mathbf{y_{j}}\}^{\mathbf{c}}\wedge\neg\exists\mathbf{t}\mathbf{R}(\mathbf{t},\mathbf{z})):\ 1{\leq}i,j{\leq}2\}$.
\end{itemize}
We claim that for all $(a_{1},a_{2}),(b_{1},b_{2}),(c_{1},c_{2}),(d_{1},d_{2}){\in}X$, if $(a_{1},a_{2}){\bowtie}(c_{1},c_{2})$ and $(b_{1},b_{2}){\bowtie}(d_{1},d_{2})$ then $(a_{1},a_{2}){S}(b_{1},b_{2})$ if and only if $(c_{1},c_{2}){S}(d_{1},d_{2})$.
If not, let $(a_{1},a_{2}),(b_{1},b_{2}),(c_{1},c_{2}),(d_{1},d_{2}){\in}X$ be such that $(a_{1},a_{2}){\bowtie}(c_{1},c_{2})$, $(b_{1},b_{2}){\bowtie}(d_{1},
$\linebreak$
d_{2})$ and either $(a_{1},a_{2}){S}(b_{1},b_{2})$ and not $(c_{1},c_{2}){S}(d_{1},d_{2})$, or not $(a_{1},a_{2}){S}(b_{1},b_{2})$ and $(c_{1},c_{2}){S}(d_{1},d_{2})$.
Without loss of generality, suppose $(a_{1},a_{2}){S}(b_{1},b_{2})$ and not $(c_{1},c_{2}){S}(d_{1},d_{2})$.
Thus, ${\mathcal F}^{\rho}_{A,B}{\models}\mathbf{A}(\mathbf{x_{1}},\mathbf{x_{2}},\mathbf{y_{1}},\mathbf{y_{2}})\ \lbrack a_{1},a_{2},b_{1},b_{2}\rbrack$ and ${\mathcal F}^{\rho}_{A,B}{\not\models}
$\linebreak$
\mathbf{A}(\mathbf{x_{1}},\mathbf{x_{2}},\mathbf{y_{1}},\mathbf{y_{2}})\ \lbrack c_{1},c_{2},d_{1},d_{2}\rbrack$.
Consequently,
\begin{description}
\item[$(C_{1})$] $\{a_{1},a_{2}\}{\cap}\{b_{1},b_{2}\}{=}\emptyset$,
\item[$(C_{2})$] there exists $i,j{\in}\{1,2\}$ such that for exactly $1$~$e$ in ${\mathcal F}_{A,B}^{\rho}$, $R_{A,B}^{\rho}(e){=}B{\setminus}\{a_{i},
$\linebreak$
b_{j}\}$ and ${R_{A,B}^{\rho}}^{{-}1}(e){=}\emptyset$,
\item[$(C_{3})$] at least one of the following conditions holds:
\begin{description}
\item[$(C_{3}^{a})$] $\{c_{1},c_{2}\}{\cap}\{d_{1},d_{2}\}{\not=}\emptyset$,
\item[$(C_{3}^{b})$] for all $k,l{\in}\{1,2\}$, either for all $f$ in ${\mathcal F}_{A,B}^{\rho}$, if $R_{A,B}^{\rho}(f){=}B{\setminus}\{c_{k},d_{l}\}$ then ${R_{A,B}^{\rho}}^{{-}1}(f){\not=}\emptyset$, or there exists $f^{\prime},f^{\prime\prime}$ in ${\mathcal F}_{A,B}^{\rho}$ such that $f^{\prime}{\not=}f^{\prime\prime}$, $R_{A,B}^{\rho}(f^{\prime}){=}B{\setminus}\{c_{k},d_{l}\}$, $R_{A,B}^{\rho}(f^{\prime\prime}){=}B{\setminus}\{c_{k},d_{l}\}$, ${R_{A,B}^{\rho}}^{{-}1}(f^{\prime}){=}\emptyset$ and
\linebreak$
{R_{A,B}^{\rho}}^{{-}1}(f^{\prime\prime}){=}\emptyset$.
\end{description}
\end{description}
%
%
Let $s,t,u,v{\in}W$ be such that $\{a_{1},a_{2}\}{=}\{(s,1),(s,2)\}$, $\{b_{1},b_{2}\}{=}\{(t,1),(t,2)\}$, $\{c_{1},c_{2}\}{=}\{(u,1),(u,2)\}$ and $\{d_{1},d_{2}\}{=}\{(v,1),(v,2)\}$.
Since $(a_{1},a_{2}){\bowtie}(c_{1},c_{2})$ and $(b_{1},b_{2}){\bowtie}(d_{1},d_{2})$, then ${\mathcal F}^{\rho}_{A,B}{\models}\mathbf{E}(\mathbf{x_{1}},\mathbf{x_{2}},\mathbf{y_{1}},\mathbf{y_{2}})\ \lbrack a_{1},a_{2},c_{1},c_{2}\rbrack$ and ${\mathcal F}^{\rho}_{A,B}{\models}\mathbf{E}(\mathbf{x_{1}},
$\linebreak$
\mathbf{x_{2}},\mathbf{y_{1}},\mathbf{y_{2}})\ \lbrack b_{1},b_{2},d_{1},d_{2}\rbrack$.
Hence, 
$\{a_{1},a_{2}\}{=}\{c_{1},c_{2}\}$ and $\{b_{1},b_{2}\}{=}\{d_{1},d_{2}\}$.
Since $\{a_{1},a_{2}\}{=}\{(s,1),(s,2)\}$, $\{b_{1},b_{2}\}{=}\{(t,1),(t,2)\}$, $\{c_{1},c_{2}\}{=}\{(u,1),(u,2)\}$ and $\{d_{1},
$\linebreak$
d_{2}\}{=}\{(v,1),(v,2)\}$, then $s{=}u$ and $t{=}v$.
Moreover, by $(C_{1})$, $s{\not=}t$.
In other respect, by $(C_{2})$, there exists $k^{\prime},l^{\prime}{\in}\{1,2\}$ such that $a_{i}{=}(s,k^{\prime})$, $b_{j}{=}(t,l^{\prime})$ and for exactly $1$~$e$ in ${\mathcal F}_{A,B}^{\rho}$, $R_{A,B}^{\rho}(e){=}B{\setminus}\{(s,k^{\prime}),(t,l^{\prime})\}$ and ${R_{A,B}^{\rho}}^{{-}1}(e){=}\emptyset$.
Since $\{c_{1},c_{2}\}{=}\{(u,1),(u,2)\}$, $s{=}u$, $\{d_{1},d_{2}\}{=}\{(v,1),(v,2)\}$, $t{=}v$ and $s{\not=}t$, then $\{c_{1},
$\linebreak$
c_{2}\}{\cap}\{d_{1},d_{2}\}{=}\emptyset$.
Moreover, there exists $k,l{\in}\{1,2\}$ such that $c_{k}{=}(s,k^{\prime})$ and $d_{l}{=}(t,l^{\prime})$.
Thus $(C_{3}^{a})$ does not hold.
Consequently, by $(C_{3})$, $(C_{3}^{b})$ holds and either for all $f$ in ${\mathcal F}_{A,B}^{\rho}$, if $R_{A,B}^{\rho}(f){=}B{\setminus}\{(s,k^{\prime}),(t,l^{\prime})\}$ then ${R_{A,B}^{\rho}}^{{-}1}(f){\not=}\emptyset$, or there exists $f^{\prime},f^{\prime\prime}$ in ${\mathcal F}_{A,B}^{\rho}$ such that $f^{\prime}{\not=}f^{\prime\prime}$, $R_{A,B}^{\rho}(f^{\prime}){=}B{\setminus}\{(s,k^{\prime}),(t,l^{\prime})\}$, $R_{A,B}^{\rho}(f^{\prime\prime}){=}
$\linebreak$
B{\setminus}\{(s,k^{\prime}),(t,l^{\prime})\}$, ${R_{A,B}^{\rho}}^{{-}1}(f^{\prime}){=}\emptyset$ and ${R_{A,B}^{\rho}}^{{-}1}(f^{\prime\prime}){=}\emptyset$.
In the former case, since $R_{A,B}^{\rho}(e){=}B{\setminus}\{(s,k^{\prime}),(t,l^{\prime})\}$, then ${R_{A,B}^{\rho}}^{{-}1}(e){\not=}\emptyset$: a contradiction.
In the latter case, since $a_{i}{=}(s,k^{\prime})$, $b_{j}{=}(t,l^{\prime})$, then there exists $f^{\prime},f^{\prime\prime}$ in ${\mathcal F}_{A,B}^{\rho}$ such that $f^{\prime}{\not=}f^{\prime\prime}$, $R_{A,B}^{\rho}(f^{\prime}){=}B{\setminus}\{a_{i},b_{j}\}$, $R_{A,B}^{\rho}(f^{\prime\prime}){=}B{\setminus}\{a_{i},b_{j}\}$, ${R_{A,B}^{\rho}}^{{-}1}(f^{\prime}){=}\emptyset$ and ${R_{A,B}^{\rho}}^{{-}1}(f^{\prime\prime}){=}
$\linebreak$
\emptyset$: a contradiction.
Hence, for all $(a_{1},a_{2}),(b_{1},b_{2}),(c_{1},c_{2}),(d_{1},d_{2}){\in}X$, if $(a_{1},a_{2})
$\linebreak$
{\bowtie}(c_{1},c_{2})$ and $(b_{1},b_{2}){\bowtie}(d_{1},d_{2})$ then $(a_{1},a_{2}){S}(b_{1},b_{2})$ if and only if $(c_{1},c_{2}){S}(d_{1},
$\linebreak$
d_{2})$.
\\
\\
Let $(W^{\prime},R^{\prime})$ be the quotient of $(X,S)$ modulo $\bowtie$~---~i.e.
\begin{itemize}
\item $W^{\prime}$ is the quotient set of $X$ modulo $\bowtie$,
\item $R^{\prime}$ is the binary relation on $W^{\prime}$ such that for all $(a_{1},a_{2}),(b_{1},b_{2}){\in}X$, $\lbrack(a_{1},a_{2})\rbrack{R^{\prime}}\lbrack(b_{1},b_{2})\rbrack$ if and only if $(a_{1},a_{2}){S}(b_{1},b_{2})$.
\end{itemize}
Since for all $(a_{1},a_{2}){\in}X$, there exists $s{\in}W$ such that $\{a_{1},a_{2}\}{=}\{(s,1),(s,2)\}$ and for all $(a_{1},a_{2}),(b_{1},b_{2}){\in}X$, $(a_{1},a_{2}){\bowtie}(b_{1},b_{2})$ if and only if there exists $s{\in}W$ such that $\{a_{1},a_{2}\}{=}\{(s,1),(s,2)\}$ and $\{b_{1},b_{2}\}{=}\{(s,1),(s,2)\}$, then let $\kappa:\ W^{\prime}{\longrightarrow}W$ be the function such that for all $(a_{1},a_{2}){\in}X$, $\kappa(\lbrack(a_{1},a_{2})\rbrack)$ is the unique $s{\in}W$ such that $\{a_{1},a_{2}\}{=}\{(s,1),(s,2)\}$.
Obviously, $\kappa$ is bijective.
%
\\
\\
We claim that for all $(a_{1},a_{2}),(b_{1},b_{2}){\in}X$, $\lbrack(a_{1},a_{2})\rbrack{R^{\prime}}\lbrack(b_{1},b_{2})\rbrack$ if and only if $\kappa(\lbrack(a_{1},a_{2})\rbrack){R}\kappa(\lbrack(b_{1},b_{2})\rbrack)$.
Let $(a_{1},a_{2}),(b_{1},b_{2}){\in}X$.
Let $s$ be the unique element in $W$ such that $\{a_{1},a_{2}\}{=}\{(s,1),(s,2)\}$ and $t$ be the unique element in $W$ such that $\{b_{1},b_{2}\}{=}\{(t,1),(t,2)\}$.
Obviously, the following conditions are equivalent:
\begin{itemize}
\item $\lbrack(a_{1},a_{2})\rbrack{R^{\prime}}\lbrack(b_{1},b_{2})\rbrack$,
\item $(a_{1},a_{2}){S}(b_{1},b_{2})$,
\item ${\mathcal F}^{\rho}_{A,B}{\models}\mathbf{A}(\mathbf{x_{1}},\mathbf{x_{2}},\mathbf{y_{1}},\mathbf{y_{2}})\ \lbrack a_{1},a_{2},b_{1},b_{2}\rbrack$,
\item $s{\not=}t$ and there exists $i,j{\in}\{1,2\}$ such that for exactly $1$~$e$ in ${\mathcal F}_{A,B}^{\rho}$, $R_{A,B}^{\rho}(e){=}
$\linebreak$
B{\setminus}\{a_{i},b_{j}\}$ and ${R_{A,B}^{\rho}}^{{-}1}(e){=}\emptyset$,
\item $\kappa(\lbrack(a_{1},a_{2})\rbrack){R}\kappa(\lbrack(b_{1},b_{2})\rbrack)$.
\end{itemize}
Thus, for all $(a_{1},a_{2}),(b_{1},b_{2}){\in}X$, $\lbrack(a_{1},a_{2})\rbrack{R^{\prime}}\lbrack(b_{1},b_{2})\rbrack$ if and only if $\kappa(\lbrack(a_{1},a_{2})\rbrack){R}
$\linebreak$
\kappa(\lbrack(b_{1},b_{2})\rbrack)$.
\\
\\
All in all, we have proved that, $(W^{\prime},R^{\prime})$ is isomorphic with $(W,R)$.
\\
\\
Consequently, we have shown that the class of all irreflexive symmetric frames with at least $2$ elements is relatively elementary definable in ${\mathcal L}_{2}$.
\medskip
\end{proof}
\begin{lemma}\label{lemma:first:case:2:N:moins}
Let $\L$ be an Euclidean modal logic.
If $\{2\}{\times}\mathbf{N}^{-}{\subseteq}\mathtt{S}_{\L}$ then $\Th(\Fr(\L))$ is undecidable.
\end{lemma}
\begin{proof}
Suppose $\{2\}{\times}\mathbf{N}^{-}{\subseteq}\mathtt{S}_{\L}$.
Since the first-order theory of the class of all irreflexive symmetric frames with at least $2$ elements is hereditarily undecidable~\cite[Page~$273$]{Ershov:1980}, then by Proposition~\ref{Syntactic:Restriction:Lemma:hereditarily:undecidable} and Lemma~\ref{lemma:relatively:elementary:definable:K2}, the first-order theory of ${\mathcal K}_{2}$ is hereditarily undecidable.
Since $\{2\}{\times}\mathbf{N}^{-}{\subseteq}\mathtt{S}_{\L}$, then by Lemma~\ref{lemma:gamma}, $\Th(\Fr(\L))$ is contained in the first-order theory of ${\mathcal K}_{2}$.
Since the first-order theory of ${\mathcal K}_{2}$ is hereditarily undecidable, then $\Th(\Fr(\L))$ is undecidable.
\medskip
\end{proof}
\begin{lemma}\label{lemma:second:case:N:plus:1}
Let $\L$ be an Euclidean modal logic.
If $\mathbf{N}^{+}{\times}\{2\}{\subseteq}\mathtt{S}_{\L}$ then $\Th(\Fr(\L))$ is undecidable.
\end{lemma}
\begin{proof}
Suppose $\mathbf{N}^{+}{\times}\{2\}{\subseteq}\mathtt{S}_{\L}$.
Since the first-order theory of the class of all irreflexive symmetric frames with at least $2$ elements is hereditarily undecidable~\cite[Page~$273$]{Ershov:1980}, then by Proposition~\ref{Syntactic:Restriction:Lemma:hereditarily:undecidable} and Lemma~\ref{lemma:relatively:elementary:definable:L2}, the first-order theory of ${\mathcal L}_{2}$ is hereditarily undecidable.
Since $\mathbf{N}^{+}{\times}\{2\}{\subseteq}\mathtt{S}_{\L}$, then by Lemma~\ref{lemma:delta}, $\Th(\Fr(\L))$ is contained in the first-order theory of ${\mathcal L}_{2}$.
Since the first-order theory of ${\mathcal L}_{2}$ is hereditarily undecidable, then $\Th(\Fr(\L))$ is undecidable.
\medskip
\end{proof}
\section{Relativizations}\label{section:about:relativization}
In this section, we adapt to our purpose a theorem about relativizations.
See~\cite[Theorem $5.1.1$ at Page~$203$]{Hodges:1993} for a more general presentation.
As far as we know, the first systematic treatment of the concept of relativization can be seen in the book of Tarski, Mostowski and Robinson about {\it Undecidable Theories}\/~\cite{Tarski:Mostowski:Robinson:1953}.
We will use it later for the proof that the modal definability problem with respect to classes of frames determined by some Euclidean modal logics is undecidable.
\\
\\
The {\it relativization of a first-order formula $C$ with respect to a first-order formula $A$ and an individual variable $\mathbf{x}$}\/ (denoted $(C)_{\mathbf{x}}^{A}$) is inductively defined as follows:
\begin{itemize}
\item $(\mathbf{R}(\mathbf{y},\mathbf{z}))_{\mathbf{x}}^{A}$ is $\mathbf{R}(\mathbf{y},\mathbf{z})$,
\item $(\mathbf{y}=\mathbf{z})_{\mathbf{x}}^{A}$ is $\mathbf{y}=\mathbf{z}$,
\item $(\neg C)_{\mathbf{x}}^{A}$ is $\neg(C)_{\mathbf{x}}^{A}$,
\item $(C\vee D)_{\mathbf{x}}^{A}$ is $(C)_{\mathbf{x}}^{A}\vee(D)_{\mathbf{x}}^{A}$,
\item $(\forall\mathbf{y}C)_{\mathbf{x}}^{A}$ is $\forall\mathbf{y}(A\lbrack\mathbf{x}/\mathbf{y}\rbrack\rightarrow(C)_{\mathbf{x}}^{A})$.
\end{itemize}
In the above definition, $A\lbrack\mathbf{x}/\mathbf{y}\rbrack$ denotes the first-order formula obtained from the first-order formula $A$ by replacing every free occurrence of the individual variable $\mathbf{x}$ in $A$ by the individual variable $\mathbf{y}$.
\\
\\
From now on, {\bf when we write $(C)_{\mathbf{x}}^{A}$, we shall assume that the sets of individual variables occurring in $A$ and $C$ are disjoint.}
\\
\\
The reader may easily verify by induction on the first-order formula $C$ that $fiv((C)_{\mathbf{x}}^{A})\subseteq(fiv(A)\setminus\{\mathbf{x}\}){\cup}fiv(C)$.
Hence, if $C$ is a sentence then $fiv((C)_{\mathbf{x}}^{A})\subseteq fiv(A)\setminus\{\mathbf{x}\}$.
\\
\\
Let $(W,R),(W^{\prime},R^{\prime})$ be frames.
\\
\\
$(W^{\prime},R^{\prime})$ is a {\it relativized reduct of $(W,R)$}\/ if there exists a first-order formula $A(\mathbf{x_{1}},\ldots,\mathbf{x_{m}},\mathbf{y})$ and there exists $s_{1},\ldots,s_{m}{\in}W$ such that $(W^{\prime},R^{\prime})$ is the restriction of $(W,R)$ to the set of all $t{\in}W$ such that $(W,R)\models A(\mathbf{x_{1}},\ldots,\mathbf{x_{m}},\mathbf{y})\ \lbrack s_{1},\ldots,
$\linebreak$
s_{m},t\rbrack$.
In this case, $(W^{\prime},R^{\prime})$ is the {\it relativized reduct of $(W,R)$ with respect to the first-order formula $A(\mathbf{x_{1}},\ldots,\mathbf{x_{m}},\mathbf{y})$ and $s_{1},\ldots,s_{m}{\in}W$.}
Obviously, $(W,R)$ possesses a relativized reduct with respect to a first-order formula $A(\mathbf{x_{1}},\ldots,\mathbf{x_{m}},\mathbf{y})$ and $s_{1},\ldots,s_{m}{\in}W$ if and only if $(W,R)\models\exists\mathbf{y}A(\mathbf{x_{1}},\ldots,\mathbf{x_{m}},\mathbf{y})\ \lbrack s_{1},\ldots,s_{m}\rbrack$.
\\
\\
As proved in~\cite[Theorem $5.1.1$ at Page~$203$]{Hodges:1993}, for all first-order formulas $A(\mathbf{x_{1}},
$\linebreak$
\ldots,\mathbf{x_{m}},\mathbf{y})$ and for all $s_{1},\ldots,s_{m}{\in}W$, if $(W^{\prime},R^{\prime})$ is the relativized reduct of $(W,R)$ with respect to $A(\mathbf{x_{1}},\ldots,\mathbf{x_{m}},\mathbf{y})$ and $s_{1},\ldots,s_{m}$ then for all first-order formulas $C(\mathbf{z_{1}},\ldots,\mathbf{z_{n}})$ and for all $u^{\prime}_{1},\ldots,u^{\prime}_{n}{\in}W^{\prime}$, the following conditions are equivalent:
\begin{itemize}
\item $(W,R)\models(C(\mathbf{z_{1}},\ldots,\mathbf{z_{n}}))_{\mathbf{y}}^{A(\mathbf{x_{1}},\ldots,\mathbf{x_{m}},\mathbf{y})}\ \lbrack s_{1},\ldots,s_{m},u^{\prime}_{1},\ldots,u^{\prime}_{n}\rbrack$,
\item $(W^{\prime},R^{\prime})\models C(\mathbf{z_{1}},\ldots,\mathbf{z_{n}})\ \lbrack u^{\prime}_{1},\ldots,u^{\prime}_{n}\rbrack$.
\end{itemize}
%
%
%
%
%
%
%
%
%
%
%
%
%
%
The importance of the concept of relativization lies in the fact that it is at the heart of the concept of stability~\cite{Balbiani:Tinchev:2017,Balbiani:Tinchev:2019}.
\\
\\
A class ${\cal C}$ of frames is {\it stable}\/ if there exists a first-order formula $A(\mathbf{x_{1}},\ldots,\mathbf{x_{m}},\mathbf{y})$ and there exists a sentence $B$ such that
\begin{itemize}
\item for all frames $(W^{\prime},R^{\prime})$ in ${\cal C}$ and for all $s^{\prime}_{1},\ldots,s^{\prime}_{m}{\in}W^{\prime}$, the relativized reduct of $(W^{\prime},R^{\prime})$ with respect to $A(\mathbf{x_{1}},\ldots,\mathbf{x_{m}},\mathbf{y})$ and $s^{\prime}_{1},\ldots,s^{\prime}_{m}$ is in ${\cal C}$,
\item for all frames $(W_{0},R_{0})$ in ${\cal C}$, there exists frames $(W^{\prime},R^{\prime}),(W^{\prime\prime},R^{\prime\prime})$ in ${\cal C}$ and there exists $s^{\prime}_{1},\ldots,s^{\prime}_{m}{\in}W^{\prime}$ such that $(W_{0},R_{0})$ is the relativized reduct of $(W^{\prime},R^{\prime})$ with respect to $A(\mathbf{x_{1}},\ldots,\mathbf{x_{m}},\mathbf{y})$ and $s^{\prime}_{1},\ldots,s^{\prime}_{m}$, $(W^{\prime},R^{\prime})\models B$, $(W^{\prime\prime},R^{\prime\prime})\not\models B$ and $(W^{\prime},R^{\prime})\preceq(W^{\prime\prime},R^{\prime\prime})$.
\end{itemize}
In this case, $(A(\mathbf{x_{1}},\ldots,\mathbf{x_{m}},\mathbf{y}),B)$ is called {\it witness of the stability of ${\cal C}$.}
\\
\\
The importance of the concept of stability lies in its capacity to connect the problem of deciding the validity of sentences to the problem of deciding the modal definability of sentences.
\begin{proposition}\label{Validity:And:Definability}
Let ${\mathcal C}$ be a class of frames.
If ${\cal C}$ is stable then the problem of deciding the validity of sentences in ${\cal C}$ is reducible to the problem of deciding the modal definability of sentences with respect to ${\cal C}$.
\end{proposition}
\begin{proof}
See~\cite{Balbiani:Tinchev:2017,Balbiani:Tinchev:2019}.
\medskip
\end{proof}
By Proposition~\ref{Validity:And:Definability}, Balbiani and Tinchev~\cite{Balbiani:Tinchev:2017,Balbiani:Tinchev:2019} have proved that the problem of deciding the modal definability of sentences is undecidable with respect to classes of frames such as the class of all transitive frames, the class of all symmetric frames, etc.
\begin{lemma}\label{lemma:L:modal:logic:Fr:L:is:stable}
Let $\L$ be an Euclidean modal logic.
Then, $\Fr(\L)$ is stable.
\end{lemma}
\begin{proof}
We consider the following $2$~cases.
\\
\\
{\bf Case where $\Diamond\top{\in}\L$:}
Let
\begin{itemize}
\item $A(\mathbf{x_{1}},\mathbf{x_{2}},\mathbf{y})$ be the first-order formula $\forall\mathbf{x}(\mathbf{R}(\mathbf{x},\mathbf{x_{1}})\vee\mathbf{R}(\mathbf{x_{1}},\mathbf{x})\leftrightarrow\mathbf{x}=\mathbf{x_{1}})\wedge\forall\mathbf{x}(\mathbf{R}(\mathbf{x},\mathbf{x_{2}})\vee\mathbf{R}(\mathbf{x_{2}},\mathbf{x})\leftrightarrow\mathbf{x}=\mathbf{x_{2}})\rightarrow\mathbf{y}\not=\mathbf{x_{1}}\wedge\mathbf{y}\not=\mathbf{x_{2}}$,
\item $B$ be the sentence $\exists\mathbf{x_{1}}\exists\mathbf{x_{2}}(\forall\mathbf{x}(\mathbf{R}(\mathbf{x},\mathbf{x_{1}})\vee\mathbf{R}(\mathbf{x_{1}},\mathbf{x})\leftrightarrow\mathbf{x}=\mathbf{x_{1}})\wedge\forall\mathbf{x}(\mathbf{R}(\mathbf{x},
$\linebreak$
\mathbf{x_{2}})\vee\mathbf{R}(\mathbf{x_{2}},\mathbf{x})\leftrightarrow\mathbf{x}=\mathbf{x_{2}})\wedge\mathbf{x_{1}}\not=\mathbf{x_{2}})$.
\end{itemize}
Obviously, for all Euclidean frames $(W,R)$ in $\Fr(\L)$, for all $s_{1},s_{2}{\in}W$ and for all Euclidean frames $(W^{\prime},R^{\prime})$, if $(W^{\prime},R^{\prime})$ is the relativized reduct of $(W,R)$ with respect to $A(\mathbf{x_{1}},\mathbf{x_{2}},\mathbf{y})$ and $s_{1},s_{2}$ then $(W^{\prime},R^{\prime})$ is in $\Fr(\L)$.
Now, let $(W_{0},R_{0})$ be an Euclidean frame in $\Fr(\L)$.
Let $(W,R)$ be the Euclidean frame in $\Fr(\L)$ such that
\begin{itemize}
\item $W{=}W_{0}{\cup}\{s_{1},s_{2}\}$,
\item $R{=}R_{0}{\cup}\{(s_{1},s_{1}),(s_{2},s_{2})\}$,
\end{itemize}
$s_{1}$ and $s_{2}$ being $2$~new distinct elements.
Let $(W^{\prime},R^{\prime})$ be the Euclidean frame in $\Fr(\L)$ such that
\begin{itemize}
\item $W^{\prime}{=}\{s^{\prime}\}$,
\item $R^{\prime}{=}\{(s^{\prime},s^{\prime})\}$.
\end{itemize}
Obviously, $(W_{0},R_{0})$ is the relativized reduct of $(W,R)$ with respect to $A(\mathbf{x_{1}},\mathbf{x_{2}},
$\linebreak$
\mathbf{y})$ and $s_{1},s_{2}$, $(W,R)\models B$, $(W^{\prime},R^{\prime})\not\models B$ and $(W,R)\preceq(W^{\prime},R^{\prime})$.
Hence, $(A(\mathbf{x_{1}},\mathbf{x_{2}},\mathbf{y}),B)$ is a witness of the stability of $\Fr(\L)$.
\\
\\
{\bf Case where $\Diamond\top{\not\in}\L$:}
Let
\begin{itemize}
\item $A(\mathbf{x_{1}},\mathbf{x_{2}},\mathbf{y})$ be the first-order formula $\forall\mathbf{x}\neg\mathbf{R}(\mathbf{x_{1}},\mathbf{x})\wedge\forall\mathbf{x}\neg\mathbf{R}(\mathbf{x_{2}},\mathbf{x})\rightarrow\mathbf{y}\not=\mathbf{x_{1}}\wedge\mathbf{y}\not=\mathbf{x_{2}}$,
\item $B$ be the sentence $\exists\mathbf{x_{1}}\exists\mathbf{x_{2}}(\forall\mathbf{x}\neg\mathbf{R}(\mathbf{x_{1}},\mathbf{x})\wedge\forall\mathbf{x}\neg\mathbf{R}(\mathbf{x_{2}},\mathbf{x})\wedge\mathbf{x_{1}}\not=\mathbf{x_{2}})$.
\end{itemize}
Obviously, for all Euclidean frames $(W,R)$ in $\Fr(\L)$, for all $s_{1},s_{2}{\in}W$ and for all Euclidean frames $(W^{\prime},R^{\prime})$, if $(W^{\prime},R^{\prime})$ is the relativized reduct of $(W,R)$ with respect to $A(\mathbf{x_{1}},\mathbf{x_{2}},\mathbf{y})$ and $s_{1},s_{2}$ then $(W^{\prime},R^{\prime})$ is in $\Fr(\L)$.
Now, let $(W_{0},R_{0})$ be an Euclidean frame in $\Fr(\L)$.
Let $(W,R)$ be the Euclidean frame in $\Fr(\L)$ such that
\begin{itemize}
\item $W{=}W_{0}{\cup}\{s_{1},s_{2}\}$,
\item $R{=}R_{0}$,
\end{itemize}
$s_{1}$ and $s_{2}$ being $2$~new distinct elements.
Let $(W^{\prime},R^{\prime})$ be the Euclidean frame in $\Fr(\L)$ such that
\begin{itemize}
\item $W^{\prime}{=}\{s^{\prime}\}$,
\item $R^{\prime}{=}\emptyset$.
\end{itemize}
Obviously, $(W_{0},R_{0})$ is the relativized reduct of $(W,R)$ with respect to $A(\mathbf{x_{1}},\mathbf{x_{2}},
$\linebreak$
\mathbf{y})$ and $s_{1},s_{2}$, $(W,R)\models B$, $(W^{\prime},R^{\prime})\not\models B$ and $(W,R)\preceq(W^{\prime},R^{\prime})$.
Thus, $(A(\mathbf{x_{1}},\mathbf{x_{2}},\mathbf{y}),B)$ is a witness of the stability of $\Fr(\L)$.
\medskip
\end{proof}
\begin{lemma}\label{Lemma:53:bis}
Let $\L$ be an Euclidean modal logic.
If $\Th(\Fr(\L))$ is undecidable then the problem of deciding the modal definability of sentences with respect to $\Fr(\L)$ is undecidable.
\end{lemma}
\begin{proof}
By Proposition~\ref{Validity:And:Definability} and Lemma~\ref{lemma:L:modal:logic:Fr:L:is:stable}.
\medskip
\end{proof}
%
%
%
%
%
%
%
%
%
%
\section{Computability of modal definability}\label{section:about:computability:of:modal:definability}
The following results show that for all Euclidean modal logics $\L$, the computability of $\Th(\Fr(\L))$ only depends on the cardinality of $\mathtt{S}_{\L}{\setminus}((\{1\}{\times}\mathbf{N}^{-}){\cup}(\mathbf{N}^{+}{\times}\{{-}1,
$\linebreak$
0,1\}))$.
\begin{lemma}\label{decidability:if:setminus:is:finite}
%
%
%
%
%
%
%
%
%
%
%
%
%
%
%
%
%
%
%
%
%
%
%
%
%
%
Let $\L$ be an Euclidean modal logic such that $\mathtt{S}_{\L}{\setminus}((\{1\}{\times}
$\linebreak$
\mathbf{N}^{-}){\cup}(\mathbf{N}^{+}{\times}\{{-}1,0,1\}))$ is finite.
Let $\varphi$ be a modal formula such that $\L{=}\K5{\oplus}\varphi$.
Let $k$ be the least integer such that $k{\geq}4$ and $(\ast)$~for all $m{\in}\mathbf{N}^{+}$ and for all $n{\in}\mathbf{N}^{-}$, if $(m,n){\in}\mathtt{S}_{\L}{\setminus}((\{1\}{\times}\mathbf{N}^{-}){\cup}(\mathbf{N}^{+}{\times}\{{-}1,0,1\}))$ then $m{+}n{\leq}k$.
For all sentences $A$, the following conditions are equivalent:
\begin{itemize}
\item $A{\in}\Th(\Fr(\L))$,
\item for all finite Euclidean frames $(W,R)$, if ${\parallel}W{\parallel}{\leq}2{\times}q{\times}((Q{\times}K){+}1)^{2}{\times}
$\linebreak$
2^{(Q{\times}K)^{2}}{\times}Q{\times}K$ and $(W,R)$ validates $\varphi$ then $(W,R)$ validates $A$,
\end{itemize}
$q$ denoting $\qdd(A)$, $Q$ denoting $2{\times}q{\times}(q{\times}(q{+}1)^{2}+1)$ and $K$ denoting $2^{k}$.
\end{lemma}
\begin{proof}
Since~$(\ast)$, then for all $m{\in}\mathbf{N}^{+}$ and for all $n{\in}\mathbf{N}^{-}$, if ${\mathcal F}_{m}^{n}{\models}\L$ then either $m{=}1$, or $n{=}{-}1$, or $n{=}0$, or $n{=}1$, or $m{\not=}1$, $n{\not=}{-}1$, $n{\not=}0$, $n{\not=}1$ and $m{+}n{\leq}k$.
Thus, by Lemma~\ref{lemma:about:flowers:simple:galaxies}, for all $m{\in}\mathbf{N}^{+}$ and for all $n{\in}\mathbf{N}^{-}$, if ${\mathcal F}_{m}^{n}{\models}\L$ then either ${\mathcal F}_{m}^{n}$ is simple, or ${\mathcal F}_{m}^{n}$ is non-simple and $m{+}n{\leq}k$.
\\
\\
We claim that for all galaxies ${\mathcal F}_{A,B}^{\rho}$, if ${\mathcal F}_{A,B}^{\rho}{\models}\L$ then either ${\mathcal F}_{A,B}^{\rho}$ is simple, or ${\mathcal F}_{A,B}^{\rho}$ is non-simple and ${\parallel}B{\parallel}{\leq}k$.
If not, let ${\mathcal F}_{A,B}^{\rho}$ be a galaxy such that ${\mathcal F}_{A,B}^{\rho}{\models}\L$, ${\mathcal F}_{A,B}^{\rho}$ is non-simple and ${\parallel}B{\parallel}{>}k$.
Consequently, there exists $s{\in}A$ such that ${\parallel}\rho(s){\parallel}{>}1$ and ${\parallel}B{\setminus}\rho(s){\parallel}{>}1$.
Let $(W_{s},R_{s})$ be the least generated subframe of ${\mathcal F}_{A,B}^{\rho}$ containing $s$.
Obviously, $W_{s}{=}\{s\}{\cup}B$ and $R_{s}{=}(\{s\}{\times}\rho(s)){\cup}(B{\times}B)$.
Moreover, since ${\mathcal F}_{A,B}^{\rho}{\models}\L$, then by Lemma~\ref{Generated:Subframe:Lemma}, $(W_{s},R_{s}){\models}\L$.
Let $m^{\prime}{=}\min\{{\parallel}\rho(s){\parallel},k\}$ and $n^{\prime}{=}\min\{{\parallel}B{\setminus}\rho(s){\parallel},k\}$.
Obviously, $m^{\prime}{>}1$, $n^{\prime}{>}1$ and $m^{\prime}{+}n^{\prime}{>}k$.
Moreover, since ${\parallel}B{\parallel}{>}k$, ${\parallel}\rho(s){\parallel}{>}1$ and ${\parallel}B{\setminus}\rho(s){\parallel}{>}1$, then by Lemma~\ref{lemma:about:k:s:new:element:bmi:of:W:R}, ${\mathcal F}_{m^{\prime}}^{n^{\prime}}$ is a bounded morphic image of $(W_{s},R_{s})$.
Since $(W_{s},R_{s}){\models}\L$, then by Lemma~\ref{Bounded:Morphism:Lemma}, ${\mathcal F}_{m^{\prime}}^{n^{\prime}}{\models}\L$.
Since $m^{\prime}{>}1$ and $n^{\prime}{>}1$, then by Lemma~\ref{lemma:about:flowers:simple:galaxies}, ${\mathcal F}_{m^{\prime}}^{n^{\prime}}$ is non-simple.
Since for all $m{\in}\mathbf{N}^{+}$ and for all $n{\in}\mathbf{N}^{-}$, if ${\mathcal F}_{m}^{n}{\models}\L$ then either ${\mathcal F}_{m}^{n}$ is simple, or ${\mathcal F}_{m}^{n}$ is non-simple and $m{+}n{\leq}k$, then either ${\mathcal F}_{m^{\prime}}^{n^{\prime}}{\not\models}\L$, or $m^{\prime}{+}n^{\prime}{\leq}k$: a contradiction.
Hence, for all galaxies ${\mathcal F}_{A,B}^{\rho}$, if ${\mathcal F}_{A,B}^{\rho}{\models}\L$ then either ${\mathcal F}_{A,B}^{\rho}$ is simple, or ${\mathcal F}_{A,B}^{\rho}$ is non-simple and ${\parallel}B{\parallel}{\leq}k$.
\\
\\
Thus, by Lemmas~\ref{lemma:about:generated:subframe:and:disjoint:unions} and~\ref{Generated:Subframe:Lemma},
\begin{itemize}
\item for all disjoint indexed families $\{{\mathcal F}_{A_{i},B_{i}}^{\rho_{i}}:\ i{\in}I\}$ of galaxies, if the union of $\{{\mathcal F}_{A_{i},B_{i}}^{\rho_{i}}:\ i{\in}I\}$ validates $\L$ then for all $i{\in}I$, either ${\mathcal F}_{A_{i},B_{i}}^{\rho_{i}}$ is simple, or ${\mathcal F}_{A_{i},B_{i}}^{\rho_{i}}$ is non-simple and ${\parallel}B_{i}{\parallel}{\leq}k$.
\end{itemize}
Now, let $A$ be an arbitrary sentence.
Suppose $A{\not\in}\Th(\Fr(\L))$.
Consequently, there exists an Euclidean frame $(W,R)$ such that $(W,R)$ validates $\L$ and $(W,R)$ does not validate $A$.
Hence, by Lemma~\ref{lemma:every:frame:disjoint:union:galaxies}, there exists a disjoint indexed family $\{{\mathcal F}_{A_{i},B_{i}}^{\rho_{i}}:\ i{\in}I\}$ of galaxies such that the union $(W,R)$ of $\{{\mathcal F}_{A_{i},B_{i}}^{\rho_{i}}:\ i{\in}I\}$ validates $\L$ and $(W,R)$ does not validate $A$.
Let $q{=}\qdd(A)$.
%
%
\\
\\
Let $\{{\mathcal F}_{A_{i}^{\prime},B_{i}^{\prime}}^{\rho_{i}^{\prime}}:\ i{\in}I\}$ be a $q$-$\alpha$-reduction of $\{{\mathcal F}_{A_{i},B_{i}}^{\rho_{i}}:\ i{\in}I\}$.
Since $(W,R)$ validates $\L$, then for all $i{\in}I$, either ${\mathcal F}_{A_{i},B_{i}}^{\rho_{i}}$ is simple, or ${\mathcal F}_{A_{i},B_{i}}^{\rho_{i}}$ is non-simple and ${\parallel}B_{i}{\parallel}{\leq}k$.
Let $I^{ns}$ be the set of all $i{\in}I$ such that ${\mathcal F}_{A_{i},B_{i}}^{\rho_{i}}$ is non-simple.
Thus, $\{{\mathcal F}_{A_{i},B_{i}}^{\rho_{i}}:\ i{\in}I^{ns}\}$ is kernel-bounded.
Consequently, by Lemma~\ref{lemma:the:alpha:reduction:is:always:dust:finite}, for all $i{\in}I$,
\begin{itemize}
\item $A_{i}^{\prime}{\subseteq}A_{i}$ and if ${\parallel}A_{i}^{\prime}{\parallel}{\leq}2$ then $A_{i}^{\prime}{=}A_{i}$,
\item $B_{i}^{\prime}{=}B_{i}$,
\item for all $s^{\prime}{\in}A_{i}^{\prime}$, $\rho_{i}^{\prime}(s^{\prime}){=}\rho_{i}(s^{\prime})$,
\item ${\parallel}{\rho_{i}^{\prime}}^{{-}1}(\emptyset){\parallel}{\leq}q$ and ${\parallel}{\rho_{i}^{\prime}}^{{-}1}(B_{i}^{\prime}){\parallel}{\leq}q$,
\item for all $t^{\prime}{\in}B_{i}^{\prime}$, ${\parallel}{\rho_{i}^{\prime}}^{{-}1}(\{t^{\prime}\}){\parallel}{\leq}q$ and for all $t^{\prime}{\in}B_{i}^{\prime}$, ${\parallel}{\rho_{i}^{\prime}}^{{-}1}(B_{i}^{\prime}{\setminus}\{t^{\prime}\}){\parallel}{\leq}q$,
\item either ${\mathcal F}_{A_{i}^{\prime},B_{i}^{\prime}}^{\rho_{i}^{\prime}}$ is simple, or ${\mathcal F}_{A_{i}^{\prime},B_{i}^{\prime}}^{\rho_{i}^{\prime}}$ is non-simple, ${\parallel}B_{i}^{\prime}{\parallel}{\leq}k$ and ${\parallel}A_{i}^{\prime}{\parallel}{\leq}q{\times}2^{k}$.
%
%
%
%
%
%
\end{itemize}
Moreover, by Proposition~\ref{lemma:good:properties:of:alpha:reductions:indexed}, $\{{\mathcal F}_{A_{i}^{\prime},B_{i}^{\prime}}^{\rho_{i}^{\prime}}:\ i{\in}I\}$ is dust-bounded and $\{{\mathcal F}_{A_{i}^{\prime},B_{i}^{\prime}}^{\rho_{i}^{\prime}}:\ i{\in}I^{ns}\}$ is root-bounded and kernel-bounded.
In other respect, $(W^{\prime},R^{\prime})$ being the union of $\{{\mathcal F}_{A_{i}^{\prime},B_{i}^{\prime}}^{\rho_{i}^{\prime}}:\ i{\in}I\}$, $(W^{\prime},R^{\prime})$ is a bounded morphic image of $(W,R)$ and the second player has a winning strategy in the Ehrenfeucht-Fra\"\i ss\'e game ${\mathcal G}_{q}((W,R),(W^{\prime},R^{\prime}))$.
As a result, by Lemmas~\ref{lemma:alpha:reduction:modal:logic} and~\ref{lemma:alpha:reduction:fol}, $(W,R){\preceq}(W^{\prime},R^{\prime})$ and $(W,R){\equiv_{q}}(W^{\prime},R^{\prime})$.
\\
\\
Let $\{{\mathcal F}_{A_{i}^{\prime\prime},B_{i}^{\prime\prime}}^{\rho_{i}^{\prime\prime}}:\ i{\in}I\}$ be a $q$-$\gamma$-reduction of $\{{\mathcal F}_{A_{i}^{\prime},B_{i}^{\prime}}^{\rho_{i}^{\prime}}:\ i{\in}I\}$.
Hence, by Lemma~\ref{lemma:gamma:reduction:prebounded:dust:quasi:root:bounded:OK}, for all $i{\in}I$,
\begin{itemize}
\item $A_{i}^{\prime\prime}{\subseteq}A_{i}^{\prime}$,
\item $B_{i}^{\prime\prime}{\subseteq}B_{i}^{\prime}$ and if ${\parallel}B_{i}^{\prime\prime}{\parallel}{\leq}2$ then $B_{i}^{\prime\prime}{=}B_{i}^{\prime}$,
\item for all $s^{\prime\prime}{\in}A_{i}^{\prime\prime}$, $\rho_{i}^{\prime\prime}(s^{\prime\prime}){\subseteq}\rho_{i}^{\prime}(s^{\prime\prime})$,
\item ${\parallel}{\rho_{i}^{\prime\prime}}^{{-}1}(\emptyset){\parallel}{\leq}q$ and ${\parallel}{\rho_{i}^{\prime\prime}}^{{-}1}(B_{i}^{\prime\prime}){\parallel}{\leq}q$,
\item for all $t^{\prime\prime}{\in}B_{i}^{\prime\prime}$, ${\parallel}{\rho_{i}^{\prime\prime}}^{{-}1}(\{t^{\prime\prime}\}){\parallel}{\leq}q$ and for all $t^{\prime\prime}{\in}B_{i}^{\prime\prime}$, ${\parallel}{\rho_{i}^{\prime\prime}}^{{-}1}(B_{i}^{\prime\prime}{\setminus}\{t^{\prime\prime}\}){\parallel}{\leq}q$,
\item either ${\mathcal F}_{A_{i}^{\prime\prime},B_{i}^{\prime\prime}}^{\rho_{i}^{\prime\prime}}$ is simple, ${\parallel}B_{i}^{\prime\prime}{\parallel}{\leq}q{\times}(q{+}1)^{2}$ and ${\parallel}A_{i}^{\prime\prime}{\parallel}{\leq}2{\times}q{\times}(q{\times}(q{+}1)^{2}+1)$, or ${\mathcal F}_{A_{i}^{\prime\prime},B_{i}^{\prime\prime}}^{\rho_{i}^{\prime\prime}}$ is non-simple, ${\parallel}B_{i}^{\prime\prime}{\parallel}{\leq}k$ and ${\parallel}A_{i}^{\prime\prime}{\parallel}{\leq}q{\times}2^{k}$.
%
%
%
%
%
%
\end{itemize}
As a result, $\{{\mathcal F}_{A_{i}^{\prime\prime},B_{i}^{\prime\prime}}^{\rho_{i}^{\prime\prime}}:\ i{\in}I\}$ is dust-bounded, root-bounded and kernel-bounded, seeing that for all $i{\in}I$, ${\parallel}A_{i}^{\prime\prime}{\parallel}{\leq}Q{\times}K$ and ${\parallel}B_{i}^{\prime\prime}{\parallel}{\leq}Q{\times}K$ where $Q{=}2{\times}q{\times}(q{\times}(q{+}
$\linebreak$
1)^{2}+1)$ and $K{=}2^{k}$.
Moreover, $(W^{\prime\prime},R^{\prime\prime})$ being the union of $\{{\mathcal F}_{A_{i}^{\prime\prime},B_{i}^{\prime\prime}}^{\rho_{i}^{\prime\prime}}:\ i{\in}I\}$, by Lemmas~\ref{lemma:gamma:reduction:modal:logic} and~\ref{lemma:gamma:reduction:fol}, $(W^{\prime},R^{\prime}){\preceq}(W^{\prime\prime},R^{\prime\prime})$ and $(W^{\prime},R^{\prime}){\equiv_{q}}(W^{\prime\prime},R^{\prime\prime})$.
%
%
%
%
%
%
%
%
%
%
\\
\\
Let $\{{\mathcal F}_{A_{j}^{\prime\prime},B_{j}^{\prime\prime}}^{\rho_{j}^{\prime\prime}}:\ j{\in}J\}$ be a $q$-$\delta$-reduction of $\{{\mathcal F}_{A_{i}^{\prime\prime},B_{i}^{\prime\prime}}^{\rho_{i}^{\prime\prime}}:\ i{\in}I\}$.
Since $\{{\mathcal F}_{A_{i}^{\prime\prime},B_{i}^{\prime\prime}}^{\rho_{i}^{\prime\prime}}:\ i{\in}I\}$ is dust-bounded, root-bounded and kernel-bounded, then by Proposition~\ref{lemma:good:properties:of:delta:reductions:indexed}, $\{{\mathcal F}_{A_{j}^{\prime\prime},B_{j}^{\prime\prime}}^{\rho_{j}^{\prime\prime}}:\ j{\in}J\}$ is dust-bounded, root-bounded and kernel-bounded.
Moreover, since for all $i{\in}I$, ${\parallel}A_{i}^{\prime\prime}{\parallel}{\leq}Q{\times}K$ and ${\parallel}B_{i}^{\prime\prime}{\parallel}{\leq}Q{\times}K$, then by Proposition~\ref{lemma:good:properties:of:delta:reductions:indexed}, ${\parallel}J{\parallel}{\leq}q{\times}(Q{\times}K{+}1)^{2}{\times}2^{(Q{\times}K)^{2}}$.
In other respect, $(W^{\prime\prime\prime},R^{\prime\prime\prime})$ being the union of $\{{\mathcal F}_{A_{j}^{\prime\prime},B_{j}^{\prime\prime}}^{\rho_{j}^{\prime\prime}}:\ j{\in}J\}$, $(W^{\prime\prime\prime},R^{\prime\prime\prime})$ is a bounded morphic image of $(W^{\prime\prime},R^{\prime\prime})$ and the second player has a winning strategy in the Ehrenfeucht-Fra\"\i ss\'e game ${\mathcal G}_{q}((W^{\prime\prime},R^{\prime\prime}),(W^{\prime\prime\prime},R^{\prime\prime\prime}))$.
Hence, by Lemmas~\ref{lemma:delta:reduction:modal:logic} and~\ref{lemma:delta:reduction:fol}, $(W^{\prime\prime},R^{\prime\prime}){\preceq}(W^{\prime\prime\prime},R^{\prime\prime\prime})$ and $(W^{\prime\prime},R^{\prime\prime}){\equiv_{q}}(W^{\prime\prime\prime},R^{\prime\prime\prime})$.
\\
\\
Since $(W,R)$ validates $\L$, $(W,R)$ does not validate $A$, $(W,R){\preceq}(W^{\prime},R^{\prime})$, $(W,R)
$\linebreak$
{\equiv_{q}}(W^{\prime},R^{\prime})$, $(W^{\prime},R^{\prime}){\preceq}(W^{\prime\prime},R^{\prime\prime})$ and $(W^{\prime},R^{\prime}){\equiv_{q}}(W^{\prime\prime},R^{\prime\prime})$, then $(W^{\prime\prime\prime},R^{\prime\prime\prime})$ validates $\L$ and $(W^{\prime\prime\prime},R^{\prime\prime\prime})$ does not validate $A$.
In other respect, ${\parallel}W^{\prime\prime\prime}{\parallel}{\leq}2{\times}q{\times}(Q{\times}
$\linebreak$
K{+}1)^{2}{\times}2^{(Q{\times}K)^{2}}{\times}Q{\times}K$.
\\
\\
All in all, we have proved that for all sentences $A$, $q$ denoting $\qdd(A)$, $Q$ denoting $2{\times}q{\times}(q{\times}(q{+}1)^{2}+1)$ and $K$ denoting $2^{k}$, the following conditions are equivalent:
\begin{itemize}
\item $A{\in}\Th(\Fr(\L))$,
\item for all finite Euclidean frames $(W,R)$, if ${\parallel}W{\parallel}{\leq}2{\times}q{\times}((Q{\times}K){+}1)^{2}{\times}
$\linebreak$
2^{(Q{\times}K)^{2}}{\times}Q{\times}K$ and $(W,R)$ validates $\varphi$ then $(W,R)$ validates $A$.
\end{itemize}
\medskip
\end{proof}
\begin{lemma}\label{lemma:infinite:undecidable:case}
Let $\L$ be an Euclidean modal logic.
If $\mathtt{S}_{\L}{\setminus}((\{1\}{\times}\mathbf{N}^{-}){\cup}(\mathbf{N}^{+}{\times}\{{-}1,
$\linebreak$
0,1\}))$ is infinite then $\Th(\Fr(\L))$ is undecidable.
\end{lemma}
\begin{proof}
Suppose $\mathtt{S}_{\L}{\setminus}((\{1\}{\times}\mathbf{N}^{-}){\cup}(\mathbf{N}^{+}{\times}\{{-}1,0,1\}))$ is infinite.
Hence, by Lemmas~\ref{lemma:if:S:L:setminus:some:pairs:is:infinite:then:two:consequences} and~\ref{lemma:S:L:is:closed:subset:of:N:plus:times:N:minus}, either $\{2\}{\times}\mathbf{N}^{-}{\subseteq}\mathtt{S}_{\L}$, or $\mathbf{N}^{+}{\times}\{2\}{\subseteq}\mathtt{S}_{\L}$.
Thus, by Lemmas~\ref{lemma:first:case:2:N:moins} and~\ref{lemma:second:case:N:plus:1}, $\Th(\Fr(\L))$ is undecidable.
\medskip
\end{proof}
\begin{proposition}\label{proposition:beta:beta:nsc:for:decidability:theory}
For all Euclidean modal logics $\L$, $\Th(\Fr(\L))$ is decidable if and only if $\mathtt{S}_{\L}{\setminus}((\{1\}{\times}\mathbf{N}^{-}){\cup}(\mathbf{N}^{+}{\times}\{{-}1,0,1\}))$ is finite.
In that case, $\Th(\Fr(\L))$ is in $\EXPSPACE$.
\end{proposition}
\begin{proof}
By Lemmas~\ref{lemma:validity:in:a:given:frame:is:in:coNP}, \ref{lemma:validity:in:a:given:frame:is:in:coNP:fol}, \ref{decidability:if:setminus:is:finite} and~\ref{lemma:infinite:undecidable:case}.
\medskip
\end{proof}
%
%
%
%
%
%
%
%
%
%
%
%
%
%
%
%
%
%
%
%
%
%
%
%
%
%
%
%
%
%
%
%
%
%
%
%
%
%
Remind that for all Euclidean modal logics $\L$, the problem of deciding the first-order definability of modal formulas with respect to $\Fr(\L)$ is trivial~\cite{Balbiani:Georgiev:Tinchev:2018}.
The following results show that for all Euclidean modal logics $\L$, the computability of the problem of deciding the modal definability of sentences with respect to $\Fr(\L)$ only depends on the cardinality of $\mathtt{S}_{\L}{\setminus}((\{1\}{\times}\mathbf{N}^{-}){\cup}(\mathbf{N}^{+}{\times}\{{-}1,0,1\}))$.
\begin{lemma}\label{proposition:gamma:gamma:L:finite:sl:setminus:ensemble:TFCAE:sentence:A}
Let $\L$ be an Euclidean modal logic such that $\mathtt{S}_{\L}{\setminus}((\{1\}{\times}\mathbf{N}^{-}){\cup}(\mathbf{N}^{+}{\times}
$\linebreak$
\{{-}1,0,1\}))$ is finite.
Let $\varphi$ be a modal formula such that $\L{=}\K5{\oplus}\varphi$.
Let $k$ be the least integer such that $k{\geq}4$ and $(\ast)$~for all $m{\in}\mathbf{N}^{+}$ and for all $n{\in}\mathbf{N}^{-}$, if $(m,n){\in}\mathtt{S}_{\L}{\setminus}((\{1\}{\times}\mathbf{N}^{-}){\cup}(\mathbf{N}^{+}{\times}\{{-}1,0,1\}))$ then $m{+}n{\leq}k$.
For all sentences $A$ and for all individual variables $\mathbf{x}$ not occurring in $A$, the following conditions are equivalent:
\begin{itemize}
\item $A$ is modally definable in $\Fr(\L)$,
\item $A{\leftrightarrow}\forall\mathbf{x}\tau(\mathbf{x},A)$ is in $\Th(\Fr(\L))$ and for all $m,m^{\prime}{\in}\mathbf{N}^{+}$ and for all $n,n^{\prime}{\in}\mathbf{N}^{-}$, if ${\mathcal F}_{m}^{n}{\models}\varphi$, $(m,n){\in}\Pi_{k}^{A}$, $m^{\prime}{\leq}m$, $n^{\prime}{\leq}n$ and ${\mathcal F}_{m}^{n}{\models}A$ then ${\mathcal F}_{m^{\prime}}^{n^{\prime}}{\models}A$.
\end{itemize}
\end{lemma}
\begin{proof}
Let $A$ be a sentence and $\mathbf{x}$ be an individual variable not occurring in $A$.
\\
\\
$\mathbf{(\Rightarrow)}$~Suppose $A$ is modally definable in $\Fr(\L)$.
Hence, by Proposition~\ref{proposition:characterizing:modal:definability}, $\mathbf{(i)}$~for all Euclidean frames $(W,R)$, if $(W,R){\models}\L$ then $(W,R){\models}A$ if and only if for all $s{\in}W$, $(W_{s},R_{s}){\models}A$ and $\mathbf{(ii)}$~for all flowers ${\mathcal F}_{m^{\prime\prime}}^{n^{\prime\prime}},{\mathcal F}_{m^{\prime\prime\prime}}^{n^{\prime\prime\prime}}$, if ${\mathcal F}_{m^{\prime\prime}}^{n^{\prime\prime}}{\models}\L$, ${\mathcal F}_{m^{\prime\prime}}^{n^{\prime\prime}}{\models}A$ and ${\mathcal F}_{m^{\prime\prime\prime}}^{n^{\prime\prime\prime}}$ is a bounded morphic image of ${\mathcal F}_{m^{\prime\prime}}^{n^{\prime\prime}}$ then ${\mathcal F}_{m^{\prime\prime\prime}}^{n^{\prime\prime\prime}}{\models}A$.
Thus, by Lemma~\ref{lemma:about:the:rooted:translation}, for all Euclidean frames $(W,R)$, if $(W,R){\models}\L$ then $(W,R){\models}A$ if and only if $(W,R){\models}\forall\mathbf{x}\tau(\mathbf{x},A)$.
Consequently, $A{\leftrightarrow}\forall\mathbf{x}\tau(\mathbf{x},A)$ is in $\Th(\Fr(\L))$.
For the sake of the contradiction, suppose there exists $m,m^{\prime}{\in}\mathbf{N}^{+}$ and there exists $n,n^{\prime}{\in}\mathbf{N}^{-}$ such that ${\mathcal F}_{m}^{n}{\models}\varphi$, $(m,n){\in}\Pi_{k}^{A}$, $m^{\prime}{\leq}m$, $n^{\prime}{\leq}n$, ${\mathcal F}_{m}^{n}{\models}A$ and ${\mathcal F}_{m^{\prime}}^{n^{\prime}}{\not\models}A$.
Hence, ${\mathcal F}_{m}^{n}{\models}\L$.
Moreover, by Lemma~\ref{lemma:about:bounded:morphisms:in:the:situation:of F:m:n:frames}, ${\mathcal F}_{m^{\prime}}^{n^{\prime}}$ is a bounded morphic image of ${\mathcal F}_{m}^{n}$.
Since ${\mathcal F}_{m}^{n}{\models}A$, then by $\mathbf{(ii)}$, ${\mathcal F}_{m^{\prime}}^{n^{\prime}}{\models}A$: a contradiction.
\\
\\
$\mathbf{(\Leftarrow)}$~Suppose $\mathbf{(iii)}$~$A{\leftrightarrow}\forall\mathbf{x}\tau(\mathbf{x},A)$ is in $\Th(\Fr(\L))$ and $\mathbf{(iv)}$~for all $m,m^{\prime}{\in}\mathbf{N}^{+}$ and for all $n,n^{\prime}{\in}\mathbf{N}^{-}$, if ${\mathcal F}_{m}^{n}{\models}\varphi$, $(m,n){\in}\Pi_{k}^{A}$, $m^{\prime}{\leq}m$, $n^{\prime}{\leq}n$ and ${\mathcal F}_{m}^{n}{\models}A$ then ${\mathcal F}_{m^{\prime}}^{n^{\prime}}{\models}A$.
For the sake of the contradiction, suppose $A$ is not modally definable in $\Fr(\L)$.
Since $\mathbf{(iii)}$, then for all Euclidean frames $(W,R)$, if $(W,R){\models}\L$ then $(W,R){\models}A$ if and only if $(W,R){\models}\forall\mathbf{x}\tau(\mathbf{x},A)$.
Thus, by Lemma~\ref{lemma:about:the:rooted:translation}, for all Euclidean frames $(W,R)$, if $(W,R){\models}\L$ then $(W,R){\models}A$ if and only if for all $s{\in}W$, $(W_{s},R_{s}){\models}A$.
Since $A$ is not modally definable in $\Fr(\L)$, then by Proposition~\ref{proposition:characterizing:modal:definability}, there exists flowers ${\mathcal F}_{m^{\prime\prime}}^{n^{\prime\prime}},{\mathcal F}_{m^{\prime\prime\prime}}^{n^{\prime\prime\prime}}$ such that ${\mathcal F}_{m^{\prime\prime}}^{n^{\prime\prime}}{\models}\L$, ${\mathcal F}_{m^{\prime\prime}}^{n^{\prime\prime}}{\models}A$, ${\mathcal F}_{m^{\prime\prime\prime}}^{n^{\prime\prime\prime}}$ is a bounded morphic image of ${\mathcal F}_{m^{\prime\prime}}^{n^{\prime\prime}}$ and ${\mathcal F}_{m^{\prime\prime\prime}}^{n^{\prime\prime\prime}}{\not\models}A$.
Consequently, ${\mathcal F}_{m^{\prime\prime}}^{n^{\prime\prime}}{\models}\varphi$.
Moreover, by Lemma~\ref{lemma:about:bounded:morphisms:in:the:situation:of F:m:n:frames}, $m^{\prime\prime\prime}{\leq}m^{\prime\prime}$ and $n^{\prime\prime\prime}{\leq}n^{\prime\prime}$.
Let ${\mathcal S}{=}\{(m,n):\ m{\in}\mathbf{N}^{+},\ n{\in}\mathbf{N}^{-},\ {\mathcal F}_{m}^{n}{\models}\varphi,\ {\mathcal F}_{m}^{n}{\models}A$ and there exists $m^{\prime}{\in}\mathbf{N}^{+}$ and $n^{\prime}{\in}\mathbf{N}^{-}$ such that $m^{\prime}{\leq}m$, $n^{\prime}{\leq}n$ and ${\mathcal F}_{m^{\prime}}^{n^{\prime}}{\not\models}A\}$.
Since ${\mathcal F}_{m^{\prime\prime}}^{n^{\prime\prime}}{\models}A$, ${\mathcal F}_{m^{\prime\prime\prime}}^{n^{\prime\prime\prime}}{\not\models}A$, ${\mathcal F}_{m^{\prime\prime}}^{n^{\prime\prime}}{\models}\varphi$, $m^{\prime\prime\prime}{\leq}m^{\prime\prime}$ and $n^{\prime\prime\prime}{\leq}n^{\prime\prime}$, then $(m^{\prime\prime},n^{\prime\prime}){\in}{\mathcal S}$.
Let $(m,n)$ be a $\ll$-minimal element in ${\mathcal S}$.
Hence, ${\mathcal F}_{m}^{n}{\models}\varphi$, ${\mathcal F}_{m}^{n}{\models}A$ and there exists $m^{\prime}{\in}\mathbf{N}^{+}$ and $n^{\prime}{\in}\mathbf{N}^{-}$ such that $m^{\prime}{\leq}m$, $n^{\prime}{\leq}n$ and ${\mathcal F}_{m^{\prime}}^{n^{\prime}}{\not\models}A$.
Thus, $(m,n){\in}\mathtt{S}_{\L}$.
Moreover,
by $\mathbf{(iv)}$, $(m,n){\not\in}\Pi_{k}^{A}$.
Thus, either $m{\geq}2$, $n{\geq}2$ and $m{+}n{>}k$, or $m{=}1$ and $n{>}\qdd(A)$, or either $n{=}{-}1$, or $n{=}0$, or $n{=}1$ and $m{>}\qdd(A)$.
In the first case, since $(m,n){\in}\mathtt{S}_{\L}$, then $(m,n){\in}\mathtt{S}_{\L}{\setminus}((\{1\}{\times}\mathbf{N}^{-}){\cup}(\mathbf{N}^{+}{\times}\{{-}1,0,1\}))$.
Consequently, by~$(\ast)$, $m{+}n{\leq}k$: a contradiction.
In the second case, since ${\mathcal F}_{m}^{n}{\models}A$ and ${\mathcal F}_{m^{\prime}}^{n^{\prime}}{\not\models}A$, then ${\mathcal F}_{m}^{n}{\not\equiv_{\qdd(A)}}{\mathcal F}_{m^{\prime}}^{n^{\prime}}$.
Since $m^{\prime}{\leq}m$, $n^{\prime}{\leq}n$, $m{=}1$ and $n{>}\qdd(A)$, then $(m,n{-}1){\in}{\mathcal S}$: a contradiction with the $\ll$-minimality of $(m,n)$ in ${\mathcal S}$.
In the third case, since ${\mathcal F}_{m}^{n}{\models}A$ and ${\mathcal F}_{m^{\prime}}^{n^{\prime}}{\not\models}A$, then ${\mathcal F}_{m}^{n}{\not\equiv_{\qdd(A)}}{\mathcal F}_{m^{\prime}}^{n^{\prime}}$.
Since $m^{\prime}{\leq}m$, $n^{\prime}{\leq}n$, either $n{=}{-}1$, or $n{=}0$, or $n{=}1$ and $m{>}\qdd(A)$, then $(m{-}1,n){\in}{\mathcal S}$: a contradiction with the $\ll$-minimality of $(m,n)$ in ${\mathcal S}$.
%
%
\medskip
\end{proof}
%
%
%
%
%
%
%
%
%
%
%
%
%
%
%
%
%
%
%
%
%
%
%
%
%
%
%
%
%
%
%
%
\begin{lemma}\label{proposition:undecidability:of:modal:definability}
Let $\L$ be an Euclidean modal logic.
If $\mathtt{S}_{\L}{\setminus}((\{1\}{\times}\mathbf{N}^{-}){\cup}(\mathbf{N}^{+}{\times}\{{-}1,
$\linebreak$
0,1\}))$ is infinite then the problem of deciding the modal definability of sentences with respect to $\Fr(\L)$ is undecidable.
\end{lemma}
\begin{proof}
By Lemmas~\ref{Lemma:53:bis} and~\ref{lemma:infinite:undecidable:case}.
\medskip
\end{proof}
\begin{theorem}\label{proposition:finale}
For all Euclidean modal logics $\L$, the problem of deciding the modal definability of sentences with respect to $\Fr(\L)$ is decidable if and only if $\mathtt{S}_{\L}{\setminus}((\{1\}{\times}\mathbf{N}^{-}){\cup}(\mathbf{N}^{+}{\times}\{{-}1,0,1\}))$ is finite.
In that case,
%
%
%
%
%
%
%
%
%
%
%
%
%
%
%
%
%
%
%
%
%
%
%
%
%
%
%
%
%
%
the problem of deciding the modal definability of sentences with respect to $\Fr(\L)$ is in $\EXPSPACE$.
\end{theorem}
\begin{proof}
By Lemmas~\ref{lemma:validity:in:a:given:frame:is:in:coNP}, \ref{lemma:ll:finitely:many:PI:PI:pairs}, \ref{lemma:validity:in:a:given:frame:is:in:coNP:fol}, \ref{proposition:gamma:gamma:L:finite:sl:setminus:ensemble:TFCAE:sentence:A} and~\ref{proposition:undecidability:of:modal:definability} and Proposition~\ref{proposition:beta:beta:nsc:for:decidability:theory}.
\medskip
\end{proof}
%
%
%
%
%
%
%
%
%
%
%
%
%
%
%
%
%
%
%
%
%
%
%
%
%
%
%
%
%
%
%
%
%
%
%
%
%
%
%
%
%
%
%
%
%
%
%
%
%
%
%
%
%
%
%
%
%
%
%
%
%
%
%
%
%
%
%
%
\begin{theorem}\label{proposition:finale:correspondence:problem}
For all Euclidean modal logics $\L$, the problem of deciding the correspondence of modal formulas and sentences with respect to $\Fr(\L)$ is decidable if and only if $\mathtt{S}_{\L}{\setminus}((\{1\}{\times}\mathbf{N}^{-}){\cup}(\mathbf{N}^{+}{\times}\{{-}1,0,1\}))$ is finite.
In that case,
%
%
%
%
%
%
%
%
%
%
%
%
%
%
%
%
%
%
%
%
%
%
%
%
%
%
%
%
%
%
the problem of deciding the correspondence of modal formulas and sentences with respect to $\Fr(\L)$ is in $\EXPSPACE$.
\end{theorem}
\begin{proof}
Let $\L$ be an Euclidean modal logics.
\\
\\
Suppose $\mathtt{S}_{\L}{\setminus}((\{1\}{\times}\mathbf{N}^{-}){\cup}(\mathbf{N}^{+}{\times}\{{-}1,0,1\}))$ is infinite.
Hence, by Proposition~\ref{proposition:beta:beta:nsc:for:decidability:theory}, $\Th(\Fr(\L))$ is undecidable.
Since for all sentences $A$, the modal formula $\top$ and $A$ correspond in $\Fr(\L)$ if and only if $A$ is in $\Th(\Fr(\L))$, then the problem of deciding the correspondence of modal formulas and sentences with respect to $\Fr(\L)$ is undecidable.
\\
\\
Suppose Suppose $\mathtt{S}_{\L}{\setminus}((\{1\}{\times}\mathbf{N}^{-}){\cup}(\mathbf{N}^{+}{\times}\{{-}1,0,1\}))$ is finite.
Thus, by Proposition~\ref{proposition:beta:beta:nsc:for:decidability:theory}, $\Th(\Fr(\L))$ is in $\EXPSPACE$.
According to Balbiani {\it et al.}\/~\cite[Lem\-ma~$18$]{Balbiani:Georgiev:Tinchev:2018}, given an arbitrary modal formula $\varphi$, one can easily construct a sentence $A(\varphi)$ such that $\size(A(\varphi)){=}{\mathcal O}(\size(\varphi))$ and $\varphi$ and $A(\varphi)$ correspond in the class of all Euclidean frames.
Thus, for all modal formulas $\varphi$ and for all sentences $A$, $\varphi$ and $A$ correspond in $\Fr(\L)$ if and only if $A\leftrightarrow A(\varphi)$ is in $\Th(\Fr(\L))$.
Since $\Th(\Fr(\L))$ is in $\EXPSPACE$, then the problem of deciding the correspondence of modal formulas and sentences with respect to $\Fr(\L)$ is in $\EXPSPACE$.
%
%
\medskip
\end{proof}
Now, it is time to assert and prove the following additional results.
\begin{corollary}\label{corollary:modal:definability:complexity}
The following decision problem is in $\NEXPTIME$:
\begin{description}
\item[input:] a modal formula $\varphi$,
\item[output:] determine whether the problem of deciding the modal definability of sentences with respect to $\Fr(\K5{\oplus}\varphi)$ is decidable.
\end{description}
\end{corollary}
\begin{proof}
Let $\varphi$ be a modal formula.
Obviously, $l$ denoting $2^{{\parallel}\vvar(\varphi){\parallel}}$, the following conditions are equivalent:
\begin{enumerate}
\item the modal definability of sentences with respect to $\Fr(\K5{\oplus}\varphi)$ is decidable,
\item $\mathtt{S}_{\K5{\oplus}\varphi}{\setminus}((\{1\}{\times}\mathbf{N}^{-}){\cup}(\mathbf{N}^{+}{\times}\{{-}1,0,1\}))$ is finite,
\item $\{2\}{\times}\mathbf{N}^{-}{\not\subseteq}\mathtt{S}_{\K5{\oplus}\varphi}$ and $\mathbf{N}^{+}{\times}\{2\}{\not\subseteq}\mathtt{S}_{\K5{\oplus}\varphi}$,
\item there exists $m{\in}\mathbf{N}^{+}$ such that ${\mathcal F}_{m}^{2}{\not\models}\varphi$ and there exists $n{\in}\mathbf{N}^{-}$ such that ${\mathcal F}_{2}^{n}{\not\models}\varphi$,
\item there exists $m{\in}\lsem1,l\rsem$ such that ${\mathcal F}_{m}^{2}{\not\models}\varphi$ and there exists $n{\in}\lsem{-}1,l\rsem$ such that ${\mathcal F}_{2}^{n}{\not\models}\varphi$,
\item ${\mathcal F}_{l}^{2}{\not\models}\varphi$ and either ${\mathcal F}_{2}^{{-}1}{\not\models}\varphi$, or ${\mathcal F}_{2}^{0}{\not\models}\varphi$, or ${\mathcal F}_{2}^{l}{\not\models}\varphi$.
\end{enumerate}
Indeed, the equivalence between~$\mathbf{(1)}$ and~$\mathbf{(2)}$ is a consequence of Theorem~\ref{proposition:finale}, the equivalence between~$\mathbf{(2)}$ and~$\mathbf{(3)}$ is a consequence of Lemmas~\ref{lemma:if:S:L:setminus:some:pairs:is:infinite:then:two:consequences} and~\ref{lemma:S:L:is:closed:subset:of:N:plus:times:N:minus}, the equivalence between~$\mathbf{(4)}$ and~$\mathbf{(5)}$ is a consequence of Lemmas~\ref{lemma:about:upper:bound:on:the:kernel:1} and~\ref{lemma:about:upper:bound:on:the:kernel:2} and the equivalence between~$\mathbf{(5)}$ and~$\mathbf{(6)}$ is a consequence of Lemmas~\ref{lemma:about:bounded:morphisms:in:the:situation:of F:m:n:frames} and~\ref{Bounded:Morphism:Lemma}.
Hence, by Lemma~\ref{lemma:validity:in:a:given:frame:is:in:coNP}, the decision problem considered in Corollary~\ref{corollary:modal:definability:complexity} is in $\NEXPTIME$.
\medskip
\end{proof}
\begin{corollary}\label{corollary:correspondence:complexity}
The following decision problem is in $\NEXPTIME$:
\begin{description}
\item[input:] a modal formula $\varphi$,
\item[output:] determine whether the problem of deciding the correspondence of modal formulas and sentences with respect to $\Fr(\K5{\oplus}\varphi)$ is decidable.
\end{description}
\end{corollary}
\begin{proof}
Let $\varphi$ be a modal formula.
Obviously, $l$ denoting $2^{{\parallel}\vvar(\varphi){\parallel}}$, the following conditions are equivalent:
\begin{enumerate}
\item the correspondence of modal formulas and sentences with respect to $\Fr(\K5
$\linebreak$
{\oplus}\varphi)$ is decidable,
\item $\mathtt{S}_{\K5{\oplus}\varphi}{\setminus}((\{1\}{\times}\mathbf{N}^{-}){\cup}(\mathbf{N}^{+}{\times}\{{-}1,0,1\}))$ is finite,
\item $\{2\}{\times}\mathbf{N}^{-}{\not\subseteq}\mathtt{S}_{\K5{\oplus}\varphi}$ and $\mathbf{N}^{+}{\times}\{2\}{\not\subseteq}\mathtt{S}_{\K5{\oplus}\varphi}$,
\item there exists $m{\in}\mathbf{N}^{+}$ such that ${\mathcal F}_{m}^{2}{\not\models}\varphi$ and there exists $n{\in}\mathbf{N}^{-}$ such that ${\mathcal F}_{2}^{n}{\not\models}\varphi$,
\item there exists $m{\in}\lsem1,l\rsem$ such that ${\mathcal F}_{m}^{2}{\not\models}\varphi$ and there exists $n{\in}\lsem{-}1,l\rsem$ such that ${\mathcal F}_{2}^{n}{\not\models}\varphi$,
\item ${\mathcal F}_{l}^{2}{\not\models}\varphi$ and either ${\mathcal F}_{2}^{{-}1}{\not\models}\varphi$, or ${\mathcal F}_{2}^{0}{\not\models}\varphi$, or ${\mathcal F}_{2}^{l}{\not\models}\varphi$.
\end{enumerate}
Indeed, the equivalence between~$\mathbf{(1)}$ and~$\mathbf{(2)}$ is a consequence of Theorem~\ref{proposition:finale:correspondence:problem}, the equivalence between~$\mathbf{(2)}$ and~$\mathbf{(3)}$ is a consequence of Lemmas~\ref{lemma:if:S:L:setminus:some:pairs:is:infinite:then:two:consequences} and~\ref{lemma:S:L:is:closed:subset:of:N:plus:times:N:minus}, the equivalence between~$\mathbf{(4)}$ and~$\mathbf{(5)}$ is a consequence of Lemmas~\ref{lemma:about:upper:bound:on:the:kernel:1} and~\ref{lemma:about:upper:bound:on:the:kernel:2} and the equivalence between~$\mathbf{(5)}$ and~$\mathbf{(6)}$ is a consequence of Lemmas~\ref{lemma:about:bounded:morphisms:in:the:situation:of F:m:n:frames} and~\ref{Bounded:Morphism:Lemma}.
Hence, by Lemma~\ref{lemma:validity:in:a:given:frame:is:in:coNP}, the decision problem considered in Corollary~\ref{corollary:correspondence:complexity} is in $\NEXPTIME$.
\medskip
\end{proof}
%
%
%
%
%
%
%
%
%
%
%
%
%
%
%
%
%
%
%
%
%
%
%
%
%
%
%
%
%
%
%
%
%
%
%
%
%
%
%
%
%
%
%
%
%
%
%
%
\section{Conclusion}\label{section:conclusion}
This paper was about the computability of the modal definability problem in classes of frames determined by Euclidean modal logics.
We have characterized those Euclidean modal logics such that the classes of frames they determine give rise to an undecidable modal definability problem.
Much remains to be done.
\\
\\
A first obvious question is whether the problem of deciding the modal definability of sentences considered in Theorem~\ref{proposition:finale} and the problem of deciding the correspondence of modal formulas and sentences considered in Theorem~\ref{proposition:finale:correspondence:problem}~---~when they are decidable~---~are $\EXPSPACE$-hard.
\\
\\
A second obvious question is whether there exists other classes of frames for which modal definability is decidable.
Is the modal definability problem in classes of frames determined by extensions of $\S4.3$ decidable?
What about the first-order definability problem in classes of frames determined by extensions of $\S4.3$?
Another question is whether there exists other modal logics for which the first-order definability problem is trivial in the classes of frames they determine.
\\
\\
It is also worth to consider restrictions or extensions of the ordinary language of modal logic.
For example, one can consider either the implication fragment based on $\rightarrow$ and $\Box$, or the positive fragment based on $\vee$, $\wedge$, $\Box$ and $\Diamond$, or the tense extension, or the extension with the universal modality.
For such restrictions or extensions of the ordinary language of modal logic, what is the computability of the first-order definability problem and the modal definability problem?
\\
\\
And in the end, there is the open question of the existence of modal logics determining classes of frames in which the modal definability problem is decidable and the first-order definability problem is undecidable.
We conjecture that such modal logics do not exist.
\section*{Funding}
The preparation of this paper has been supported by the {\it Rila Programme} (project $\mathbf{KP{-}06{-}RILA{/}4}$ of the Bulgarian Ministry of Education and Science and project $\mathbf{48055QG}$ of the French Ministry for Europe and Foreign Affairs and the French Ministry of Higher Education and Research).
\section*{Acknowledgement}
Special acknowledgement is heartily granted to our colleagues of the Faculty of Mathematics and Informatics (Sofia, Bulgaria) and the Toulouse Institute of Computer Science Research (Toulouse, France) for many stimulating discussions about the subject of this paper.
%
%
%
%
\bibliographystyle{named}
\end{document}